\documentclass[11pt, reqno]{amsart}

\usepackage[dcucite]{harvard}

\usepackage{graphicx,amsmath,amssymb,epsfig,harvard}
\usepackage{graphicx,amsmath,amssymb,epsfig}
\usepackage{colortbl}

\usepackage{xr}

\long\def\comment#1{} \oddsidemargin +0.2in \evensidemargin +0.2in
\usepackage{amssymb}
\usepackage{amsmath}
\usepackage{amsfonts}
\usepackage{graphicx}

\long\def\comment#1{}
\newtheorem{algorithm}{Algorithm}
\newtheorem{theorem}{Theorem}

\newtheorem{corollary}{Corollary}
\newtheorem{lemma}{Lemma}

\theoremstyle{definition}
\newtheorem{definition}{Definition}
\numberwithin{definition}{section}
\numberwithin{theorem}{section}
\numberwithin{lemma}{section}
\numberwithin{proposition}{section}
\numberwithin{corollary}{section}
\numberwithin{algorithm}{section}
\newtheorem{design}{Design}
\newtheorem{step}{Step}
\newtheorem{remark}{Remark}[section]
\numberwithin{remark}{section}
\newtheorem{example}{Example}
\newcommand{\citen}{\citeasnoun}

\newcommand{\eps}{\varepsilon}

\newcommand{\be}{\begin{eqnarray}}
\newcommand{\ee}{\end{eqnarray}}

\newcommand{\Ss}{\mathcal{S}}

\newcommand{\DD}{\mathcal{D}}
\newcommand{\II}{\mathcal{I}}
\newcommand{\JJ}{\mathcal{J}}
\newcommand{\XX}{\mathcal{X}}
\newcommand{\VV}{\mathcal{V}}
\newcommand{\UU}{\mathcal{U}}
\newcommand{\MM}{\mathcal{M}}
\newcommand{\KK}{\mathcal{K}}

\newcommand{\Vol}{\mathrm{Vol}}

\newcommand{\ba}{\begin{array}}
\newcommand{\ea}{\end{array}}
\newcommand{\bs}{\begin{align}\begin{split}\nonumber}
\newcommand{\bsnumber}{\begin{align}\begin{split}}
\newcommand{\es}{\end{split}\end{align}}

\renewcommand{\[}{\left[}
\renewcommand{\]}{\right]}
\renewcommand{\hat}{\widehat}

\newcommand{\C}{\mathcal{C}}

\newcommand{\X}{\mathcal{X}}
\newcommand{\TT}{\mathcal{T}}

\newcommand{\Ep}{{\mathrm{E}}}

\newcommand{\Q}{{\mathrm{Q}}}
\newcommand{\transp}{^\mathsf{T}}

\renewcommand{\Pr}{{\mathrm{P}}}

\def\RR{ {\Bbb{R}}}

\renewcommand{\hat}{\widehat}
\renewcommand{\leq}{\leqslant}
\renewcommand{\geq}{\geqslant}
\newcommand{\Keywords}[1]{\par\noindent{\small{\em Keywords\/}: #1}}

\title[The Sorted Effects Method]{The Sorted Effects Method: Discovering Heterogeneous Effects Beyond Their Averages}
\author{Victor Chernozhukov, Ivan Fernandez-Val, Ye Luo}
\date{\today}
\thanks{MIT, BU and University of Florida. First version of October, 2011. We thank the coeditor, three anonymous referees, Alberto Abadie, Daron Acemoglu, Isaiah Andrews, Joshua Angrist, Debopam Bhattacharya, Jaroslav Borovicka, Kasey Buckles, Denis Chetverikov, Jerry Hausman, Tetsuya Kaji, Pat Kline, Roger Koenker, Siyi Luo, Maddie McKelway,  Anna Mikusheva,
Whitney Newey, Claudia Olivetti, Andres Santos, Camille Terrier, Frank Windmeijer, Abigail Wozniak,  and seminar participants at  multiple institutions
for
helpful discussions. We gratefully acknowledge research support from the National Science Foundation.}
\begin{document}
\maketitle
\begin{abstract}
{\tiny The  partial (ceteris paribus) effects of interest in nonlinear  and interactive linear models are heterogeneous as they can vary dramatically with the underlying observed or unobserved covariates. Despite the apparent importance of heterogeneity, a common practice in modern empirical work is to largely \textit{ignore it} by reporting  average partial effects (or, at best, average effects for some groups).
While average effects provide very convenient scalar summaries of  typical effects, by definition they fail to reflect the entire variety of the heterogeneous effects.  In order to discover these effects much more fully, we propose to estimate and report \textit{sorted effects} -- a collection of estimated partial effects sorted in increasing order and indexed by percentiles. By construction the sorted effect curves \textit{completely represent} and help visualize the range of the heterogeneous effects in one plot. They are as convenient and easy to report in practice as the conventional average partial effects. They also serve as a basis for \textit{classification analysis}, where we divide the observational units into most or least affected groups and  summarize their characteristics. We  provide a quantification of uncertainty (standard errors and confidence bands) for the estimated sorted effects and related classification analysis, and provide confidence sets for the most and least affected groups.
The  derived statistical results rely on establishing  key, new mathematical results on Hadamard differentiability of a multivariate sorting operator and a related classification operator, which are of independent interest.  \\

We apply the \textit{sorted effects method}  and classification analysis to demonstrate several  striking patterns in the gender wage gap.  We find that this gap is particularly
strong for married  women, ranging  from $-60\%$ to $0\%$ between the $2\%$ and $98\%$ percentiles, as a function
of observed and unobserved characteristics; while the gap for never married women ranges
from $-40\%$ to $+20\%$. The most adversely affected women tend to be married, do not have college degrees,
work in sales, and have high levels of potential experience.
 \\

\Keywords{Sorting, Partial Effect, Marginal Effect,  Sorted Effect, Classification
Analysis, Nonlinear Model, Functional Analysis, Differential Geometry, Gender Wage Gap}}
\end{abstract}

\section{introduction}\label{sec:intro}

In nonlinear and interactive linear models the partial (ceteris paribus) effects  of interest  often vary with respect to the underlying covariates. For example, consider a binary response model with conditional choice
probability  $\Pr(Y = 1 \mid X) = F(X\transp \beta)$, where $Y$ is a binary response variable, $X$ is a vector of covariates, $F$ is a distribution function such as the standard normal or logistic, and $\beta$ is a vector of coefficients. The partial or predictive  effect (PE) of a marginal change in a continuous covariate $X_j$ with coefficient $\beta_j$ on the conditional choice probability is  $$\Delta(X) = f(X\transp \beta) \beta_j, \quad  f(v) = \partial F(v)/\partial v,$$ which generally varies in the population of interest with the covariate vector $X$, as $X$ varies according to some distribution, say $\mu$.   A common empirical practice is to report the average partial effect (APE), $$\Ep [\Delta (X)] = \int \Delta(x) d\mu(x),$$ as a single summary measure of the PE  (e.g., \citen[Chap. 2]{wooldridge:text}),  or to report effects for some groups  (e.g., \citen{angrist2008mostly}). However, the APE completely disregards the heterogeneity of the PE and may give a very incomplete picture of the  impact of the covariates.


In this paper we propose complementing the APE by reporting
the entire set of PEs sorted in increasing order and indexed by a ranking with respect to the distribution of the covariates in the population of interest. These sorted effects  correspond to percentiles of the PE,
\begin{equation*}\label{define:SPE}
\Delta^*_{\mu}(u) = \text{$u^{th}$-quantile  of  }  \Delta(X), \quad X \sim \mu,
\end{equation*}
 and provide a more complete representation of the heterogeneity of $\Delta(X)$.
%
We shall call these effects as sorted predictive or partial effects (SPE) by default, as most models are predictive.\footnote{When the underlying model has a structural or causal
 interpretation, we may use the name sorted structural effects or sorted treatment effects.} We also show how to use the SPEs to carry out  classifications analysis (CA). This analysis consists of classifying the observational  units into  most or least affected  depending on whether their PEs are above or below some tail SPE, and then comparing the moments or distribution of the covariates of the most and least affected groups.


 Heterogeneous effects also arise in the most basic  linear models with interactions \cite{oaxaca73,cox84}. Consider a conditional mean model for the Mincer earnings function:
 $$
 Y =   P(T,W)\transp  \beta + \epsilon,  \quad \Ep[\epsilon \mid T, W] = 0, \quad X = (T, W),
 $$
 where $Y$ is log wage, $T$ is an indicator of gender (or race, treatment, or program participation), and $W$ is a vector of labor market  characteristics. The vector
$P(T,W)$ is a collection of transformations of $T$ and $W$, involving some interaction between $T$ and $W$. For example, \citen{oaxaca73} used the specification  $P(T,W) = (TW, (1-T)W)$. Then, the PE of changing $T=0$ to $T=1$ is
 $$
\Delta(X) =  P(1,W)\transp  \beta -  P(0,W)\transp  \beta,$$
which is a measure of the gender wage gap conditional on worker characteristics.
 The
 function $u \mapsto \Delta^*_\mu(u)$
 provides again a complete
summary of the entire range of PEs. The left panel of Figure \ref{fig:gender} illustrates the SPE of the conditional gender wage gap for women.
\begin{figure}
\centering
\includegraphics[width=.45\textwidth,height=.45\textwidth]{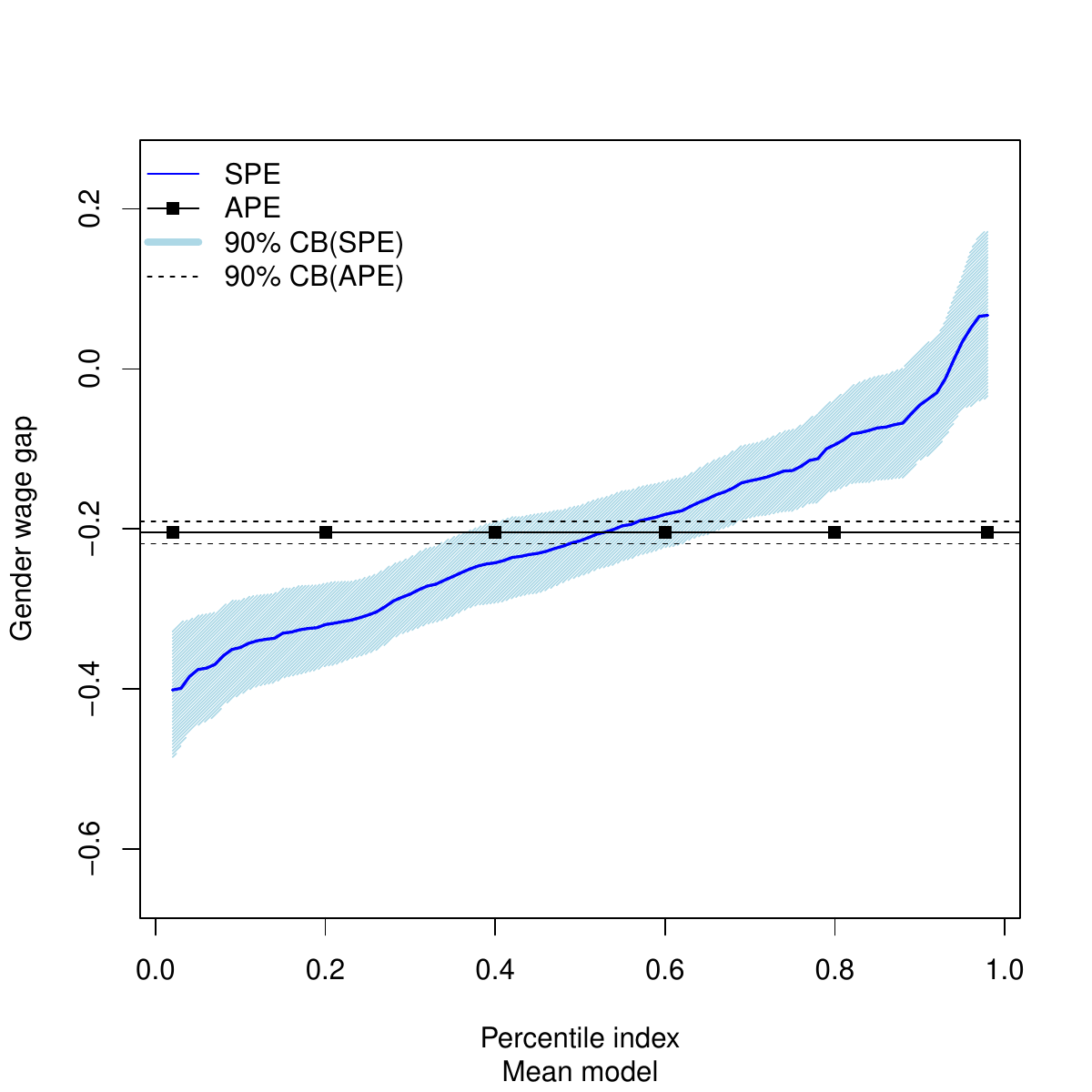}
\includegraphics[width=.45\textwidth,height=.45\textwidth]{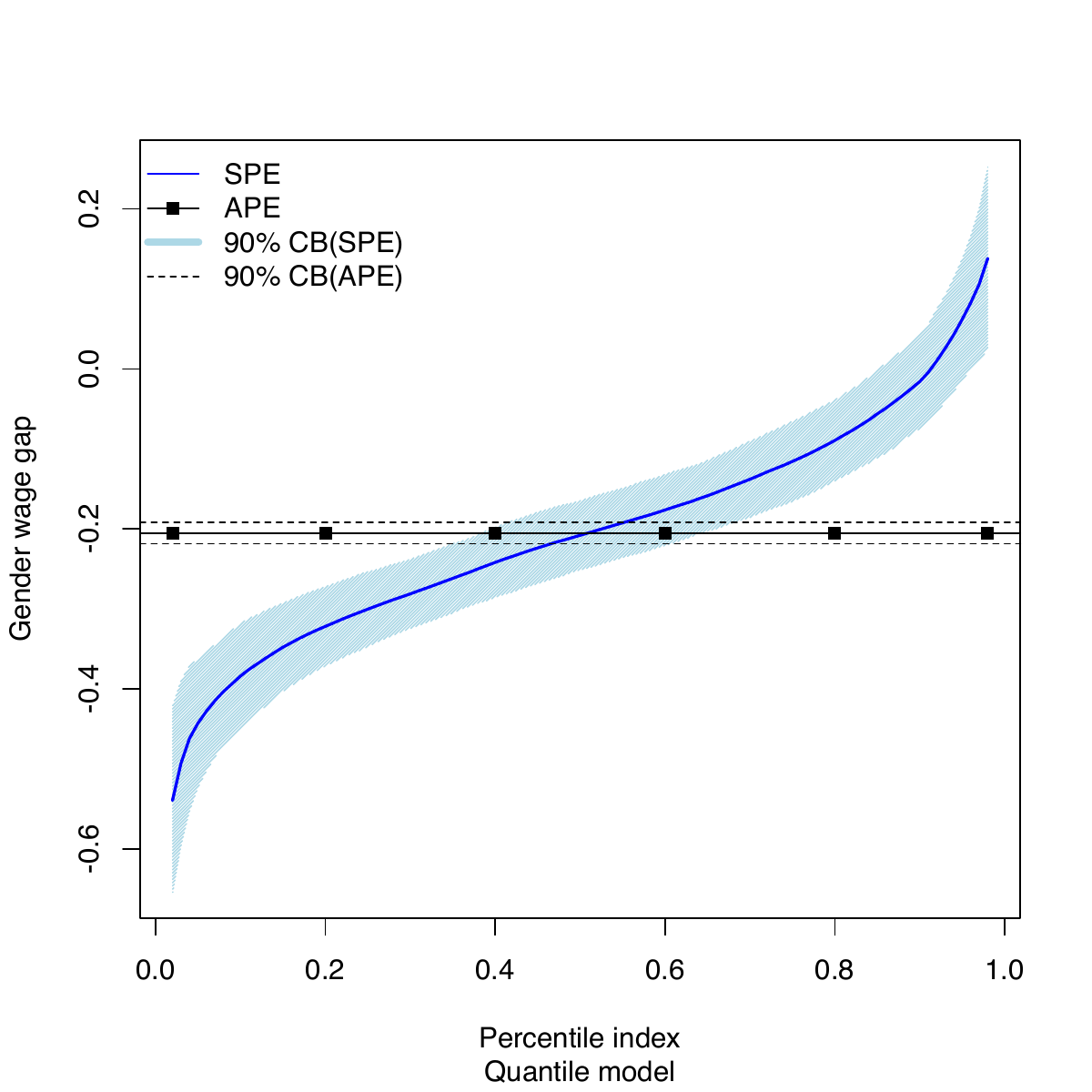}
\caption{APE and SPE (introduced in this paper) of the gender wage gap for women.  Estimates and 90\% bootstrap uniform confidence bands (derived in this paper) based on a linear model with interactions for the conditional expectation (left) and quantile (right) functions. }\label{fig:gender}
\end{figure}
The SPE  varies sharply from around $-40$ to $6.5\%$, and does not coincide with the average PE of $-20\%$.  The PE is especially
(negatively) large for women who have any of the following characteristics: married, low educated,   high  experience, and working on sales occupations  --  this follows from the classification analysis, where we compare the average characteristics of the subpopulations of women with covariate values $X$ such that $\Delta(X)$ is above the 90\% percentile and below the 10\% percentile.  We refer the reader to Section \ref{sec:empirics} for a detailed discussion of this example.


The general settings that we deal with in this paper as well as the specific results we obtain are as follows: Let $X$ denote  a covariate vector, $\Delta(X)$ denote a generic PE of interest, $\mu$ denote the distribution of $X$ in the population of interest, and $\XX$ denote the interior of the support of $X$ in this population. The SPE is obtained by sorting the multivariate function $x \mapsto \Delta(x)$ in increasing order with respect to $\mu$.
Using tools from differential geometry, we prove that this multivariate sorting operator is Hadamard differentiable with respect to the PE function $\Delta$ and the distribution $\mu$ at the regular values of $x \mapsto \Delta(x)$ on $\XX$. This key and new mathematical result allows us to derive the large sample properties of the empirical SPE, which replace $\Delta$ and $\mu$ by sample analogs, obtained from parametric or semi-parametric estimators, using the functional delta method. In particular, we derive a functional central limit theorem and a bootstrap functional central limit theorem for the empirical SPE. The main requirement of these theorems is that the empirical $\Delta$ and $\mu$ also satisfy functional central limit  theorems, which hold for many estimators used in empirical economics under general sampling conditions. We use the properties of the empirical SPE to construct confidence sets for the SPE that hold uniformly over quantile indices. We also show under the same conditions that the empirical version of the objects in the classification analysis follow functional central limit theorems and bootstrap functional central limit theorems. We derive these result by establishing the Hadamard differentiability of a classification operator related to the multivariate sorting operator.

\paragraph{\textbf{Related technical literature:} }    Previously,  \citen{CFG-10} derived the properties of the rearrangement (sorting) operator  in the univariate case with known $\mu$ (standard uniform distribution). Those results were motivated by a completely different problem -- namely, restoration of monotonicity in conditional quantile estimation -- rather than the problem of summarizing heterogeneous effects by the SPEs. These prior technical results are not applicable  to our case  as soon as the dimension of $X$ is greater than one, which is the case in all modern applications where  effects are of interest. Moreover, the previous results are not applicable even in the univariate case since the measure $\mu$ is not known in all envisioned applications.   The properties of the sorting operator are different in the multivariate case and require tools from differential geometry:  computation of functional (Hadamard) derivatives of the sorting operator with respect to perturbations of $\Delta$ require us to work with integration on $(d_x-1)$-dimensional manifolds of the type $\{\Delta(x) = \delta\}$, where $d_x = \dim X$.   Moreover, we also need to compute functional derivatives with respect to suitable perturbations of the measure $\mu$.    In econometrics or statistics, \citen{sasaki-15} also used differential geometry to characterize the structural properties of derivatives of conditional quantile functions in nonseparable models; and \citen{kim:pollard} used tools from differential geometry to derive the large sample properties of the maximum score and other cube root consistent estimators.  Relative to these papers,  we share the use of differential geometry tools as a general proof strategy, but we apply these tools to establish the analytical properties of different functionals -- namely, the SPEs.  Moreover, our results on the functional differentiability of the sorting and classification operators in the multivariate case constitute new mathematical results, which are of interest in their own right.



\paragraph{\textbf{Organization of the paper:} }  In Section \ref{sec:model} we discuss the quantities of interest in nonlinear and interactive linear models
with examples;  introduce the SPE and related  CA, along with their empirical counterparts; and outline the main inferential results.   In Section \ref{sec:empirics}
we provide an empirical application to the gender wage gap in the U.S. in 2015.    We derive the properties of the empirical SPE and CA in large samples and
show how to use these properties to make inference in Section \ref{sec:theory2}. 
Appendix  \ref{sec:theory1} provides some key mathematical results on the differentiability of the  multivariate sorting and classification operators and Appendix \ref{app:proofs}  contains the proof of the main results. All other proofs are given in  the online appendix with supplementary material (SM), which also contains additional technical material, and results from Monte Carlo simulations and an empirical application to mortgage denials  using binary response models \cite{cfy17sup}.

\section{Sorted Effects and Classification Analysis}\label{sec:model}

We start by discussing the objects of interest in nonlinear and interactive linear models.


\subsection{Effects of Interest} We consider a general model characterized by a predictive function $g(X)$, where $X$ is a $d_x$-vector  of covariates that may contain unobserved components, as in quantile regression models. The function $g$ usually arises from a model for a response variable $Y$, which can be discrete or continuous. We call the function $g$ predictive because the underlying model can be either predictive or causal under additional assumptions, but we do not insist on estimands having a causal interpretation.   For example, in a  mean regression model, $g(X) = \Ep[Y \mid X]$  corresponds to the
expectation function of $Y$ conditional on $X$;  in a
binary response model, $g(X) = \Pr[Y = 1 \mid X]$ corresponds to
the choice probability of $Y=1$ conditional on $X$; in a quantile
regression model, $g(X) = \Q_{Y}[\epsilon \mid Z]$, where the covariate $X = (\epsilon,Z)$
consists of the unobservable rank variable $\epsilon$ with a uniform distribution, $\epsilon \mid Z \sim U(0,1)$, and the observed covariate vector $Z$, and where $\Q_{Y}[ \tau \mid Z]$  is the
conditional $\tau^{th}$-quantile of $Y$ given $Z$.

Let $X = (T,W)$, where $T$ is the key covariate or treatment of interest, and $W$ is a vector of  control variables.  We are interested in the effects of changes in $T$ on the  function $g$ holding $W$ constant. These effects are usually called partial effects, marginal effects, or treatment effects. We call them   predictive effects (PE) throughout the paper, as such a name most accurately describes the meaning of the estimand (especially when
a causal interpretation is not available).  If $T$ is discrete, the PE is
\begin{equation}\label{eq: PE-D}
\Delta(x) = \Delta(t,w) = g(t_1,w) - g(t_0,w),
\end{equation}
where $t_1$ and $t_0$ are two values of $T$ that might depend on $t$ (e.g., $t_0 = 0$ and $t_1 = 1$, or $t_0 = t$ and $t_1 = t + 1$). This PE measures the effect of changing $T$ from $t_0$ to $t_1$ holding $W$ constant at $w$. If $T$ is continuous and $t \mapsto g(t,w)$ is differentiable, the PE is
\begin{equation}\label{eq: PE-D}
\Delta(x) = \Delta(t,w) = \partial_{t} g(t,w),
\end{equation}
where $\partial_{t}$ denotes $\partial/ \partial t$, the partial derivative with respect to $t$. This PE measures the effect of a marginal change of $T$ from the level $t$ holding $W$ constant at $w$.\footnote{We can also consider high-order and crossed effects. For example, $\Delta(x) =  \partial^2_{t^2} g(t,w)$ gives the second-order PE of the continuous treatment $T$ if $t \mapsto g(t,w)$ is twice differentiable; and, letting $X=(T,S,W)$ where $T$ and $S$ are discrete, $\Delta(x)  = g(t_1,s_1,w) - g(t_0,s_1,w)-g(t_1,s_0,w) + g(t_0,s_0,w)$ gives the crossed  effect or interaction of $T$ and $S$.}

We consider the following examples in the empirical
applications of Section \ref{sec:empirics} and SM.

\begin{example}[Binary response model]\label{example:binary} Let $Y$ be a binary response variable such as an indicator for
mortgage denial, and $X$ be a vector of covariates related to $Y$.  The predictive function of the probit or logit model takes the form:
$$
g(X) = \Pr(Y = 1 \mid X) = F(P(X)\transp\beta),
$$
where $P(X)$ is  a vector of known transformations of $X$, $\beta$ is a parameter vector,
and $F$ is a known distribution function (the standard normal distribution function
in the probit model or standard logistic distribution function in the logit model).  If $T$ is a binary variable such as an  indicator for black applicant and $W$ is a vector of controls such as the applicant characteristics relevant for the bank decision, the
 PE,
$$
\Delta(x) = F(P(1,w)\transp\beta) - F(P(0,w)\transp\beta),
$$
describes the difference in predicted  probability of mortgage denial for a black applicant
and a white applicant, conditional on a specific value $w$ of
the  observable characteristics $W$.   \qed \end{example}

\begin{example}[Interactive linear model with additive error]\label{example:cef} Let $Y$ be the logarithm of
the wage.  Suppose $X = (T, W)$, where $T$ is an indicator for  female worker
and $W$ are other worker characteristics. We can model the conditional expectation function of log wage using the linear interactive model:
 $$
 Y = g(X) + \epsilon =  P(T,W)\transp  \beta + \epsilon,  \quad \Ep[\epsilon \mid T, W] = 0, \quad X = (T,W),
 $$
 where $P(T,W)$ is a collection of transformations of $T$ and $W$, involving some interaction between $T$ and $W$. For example, $P(T,W) = (TW, (1-T)W)$. Then the
 PE
  $$
\Delta(x) =  P(1,w)\transp  \beta -  P(0,w)\transp  \beta$$
is the (average) gender wage gap or difference between the expected log wage of a woman and a man,
conditional on a specific value $w$ of the characteristics $W$.   \qed
 \end{example}

\begin{example}[Linear model with non-additive error, or QR model]\label{example:cqf}
Let $Y$ be log wage, $T$ be an
 indicator for female worker, and $W$ be a vector of worker
characteristics as in the previous example. Suppose we model the conditional quantile function of log wage using the linear interactive model:
 $$
 Y = g(X) =  P(T,W)\transp  \beta(\epsilon),  \quad  \epsilon \mid T, W  \sim U(0,1), \quad X = (T,W,\epsilon),
 $$
 where  $P(T,W)\transp  \beta(\tau)$ is the conditional
 $\tau^{th}$-quantile of $Y$ given $T$ and $W$.  Thus
 the covariate vector $X = (T,W,\epsilon)
$
 includes the observed covariates $(T,W)$ as well as the
 rank variable $\epsilon$, which is an unobserved factor (e.g., ``ability rank").  Here
 $P(T,W)$ is a collection of transformations of $T$ and $W$, e.g., $P(T,W) = (TW, (1-T)W)$. Then the  PE
  $$
\Delta(x) =  P(1,w)\transp  \beta(\tau) -  P(0,w)\transp  \beta(\tau), \quad x = (t,w,\tau),$$
is the ($\tau^{th}$-quantile) gender wage gap or difference between the conditional $\tau^{th}$-quantile of log-wage of a woman and a man, conditional on a specific value $w$ of the characteristics $W$.

Note that
in this case,
$$
X \sim \mu,  \quad  \mu =  F_{T,W} \times F_\epsilon,
$$
where $F_{\epsilon}$ is the distribution function of the standard uniform random variable,
and $F_{T,W}$ is the distribution of $(T,W)$. For estimation purposes, we will have to exclude the tail quantile indices,
 so $F_{\epsilon}$ will be redefined to have support on a set of the form $[\ell, 1-\ell]$, where $0 < \ell < 0.5$ is a small positive number.
\qed \end{example}

The set of examples listed above are the most basic, leading cases, arising
mostly in predictive analysis and program evaluation.  Our theoretical
results are rather general and are not limited to these cases. Thus, they allow for both $\Delta$ and $\mu$ to originate from causal or structural models and to be estimated by structural methods. For example, in treatment effects models with selection on observables \cite{rr83}, the PE is the conditional average treatment effect $\Delta(x) = \Ep(Y_1 - Y_0 \mid X=x),$ where $Y_1$ and $Y_0$ are potential outcomes in the treated and non-treated statuses and $X$ is a vector of covariates. The standard approach is to aggregate the conditional average treatment effects by integration with respect to the distribution of the covariates  in the population of interest $\mu$. This yields the average treatment effect if $\mu$ is the distribution in the entire population or the average treatment effect on the treated if $\mu$ is the distribution in the treated subpopulation.  The SPE can be used to complement the analysis by reporting the entire range of conditional average treatment effects, and also to determine the optimal treatment allocation with budget constraints.  Thus, \citen{bd12} showed that under some conditions  this optimal allocation has a cutoff determined by a tail percentile of the conditional average treatment effects, i.e. by a SPE. Another example is the welfare analysis described in \citen{hn17}, where $\Delta(x)$ is the compensating or equivalent variation of a price change conditional on covariates such as income and demographic characteristics, and $\mu$ is the distribution of covariates in the population of interest. 

In all the previous examples,  the PE  $\Delta(x)$ is a function of $x$ and therefore can be different for each observational unit. To summarize this effect in a single measure, a common practice in empirical economics  is to average the PEs.   Averaging, however, masks most of the heterogeneity in the PE allowed by nonlinear or interactive linear models. We propose reporting the entire set of values of the PE sorted in increasing order and indexed by a ranking  $u \in [0,1]$ with respect to the population of interest. 
These sorted effects provide a more complete representation of the heterogeneity in the PE than the average effects.
\begin{definition}[$u$-SPE] The \textit{$u^{th}$-sorted predictive effect} with respect to $\mu$ is
$$
\Delta_{\mu}^*(u)  :=  \inf\{ \delta \in \mathbb{R} : F_{\Delta,  \mu}(\delta)  \geq u \},  \ \  F_{\Delta, \mu}(\delta) := \Ep_{\mu}[1\{\Delta(X) \leq \delta \}],
$$
where $\Ep_{\mu}$ denotes expectation with respect to $\mu$.
\end{definition}
The $u$-SPE is the $u^{th}$-quantile of $\Delta(X)$ when $X$ is distributed according to $\mu$. As for the average effect, $\mu$ can be chosen to select a target subpopulation from the entire population. For example, when $T$ is a treatment indicator:
\begin{itemize}
\item If $\mu$ is set to the marginal distribution of $X$ in the entire population, then $\Delta_{\mu}^*(u)$ is the \textit{population $u$-SPE}.
\item If $\mu$ is set to the distribution of $X$ conditional on $T=1$, then   $\Delta_{\mu}^*(u)$ is the  \textit{$u$-SPE on the treated  or exposed}.
\end{itemize}

By considering $\Delta_{\mu}^*(u) $ at multiple quantile indices, we obtain a one-dimensional representation of the heterogeneity of the PE. Accordingly, our object of interest is the \textit{SPE-function}
$$
\{u \mapsto \Delta_{\mu}^*(u) : u \in \mathcal{U} \}, \ \ \mathcal{U} \subseteq [0,1],
$$
where $\mathcal{U}$ is the set of quantile indices of interest.

We also show how to use the $u$-SPE for classification analysis.  Let $u \in \mathcal{U},$ with $u < 1/2$, and $Z$ be a $d_z$-dimensional random vector that includes $X$ and possibly other variables such as $Y$ in Examples 1--3. By abuse of notation, we also denote the distribution of $Z$ over its support $\mathcal{Z}$ as $\mu$.
\begin{definition}[$u$-CA]\label{def:uca} The $u^{th}$-classification analysis consists of 2 steps: (i)  Assign all $Z$ with $\Delta(X) < \Delta_{\mu}^*(u)$ to the \textit{$u$-least affected} subpopulation, and all $Z$ with $\Delta(X) > \Delta_{\mu}^*(1-u)$ to  the \textit{$u$-most affected} subpopulation. (ii) Obtain the moments and distribution of  $Z$ in the least and most affected subpopulations. We denote by $\Lambda_{\Delta,\mu}^{-u}(t)$ and $\Lambda_{\Delta,\mu}^{+u}(t)$ generic objects indexed by $t \in \mathbb{R}^{d_z}$ in the least and most affected subpopulations, respectively. For example, $\Lambda_{\Delta,\mu}^{-u}(t) = \Ep_{\mu}[Z^t \mid \Delta(X) < \Delta_{\mu}^*(u)]$ corresponds to the $t$-moment  of $Z$ in the $u$-least affected subpopulation, and $\Lambda_{\Delta,\mu}^{-u}(t) = \Ep_{\mu}[1(Z \leq t) \mid \Delta(X) < \Delta_{\mu}^*(u)]$  to the distribution  of $Z$ at $t$  in the $u$-least affected subpopulation.\footnote{For a $d_z$-dimensional random variable $Z = (Z_1, \ldots, Z_{d_z})$ and $t = (t_1, \ldots, t_{d_z}) \in \mathbb{R}^{d_z}$, we denote $Z^t := \prod_{j=1}^{d_z} Z_j^{t_j}$ and $\{Z \leq t\} := \{Z_1 \leq t_1, \ldots, Z_{d_z} \leq t_{d_z}\}$.}
 We define the same quantities in the $u$-most affected subpopulation replacing $\Delta(X) < \Delta_{\mu}^*(u)$ by $ \Delta(X) > \Delta_{\mu}^*(1-u)$ in the conditioning set.
\end{definition}

\subsection{Empirical SPE}\label{sec:estimation}
In practice, we replace the PE $\Delta$
and the distribution $\mu$ by sample analogs to construct plug-in estimators of the SPE. Let $\widehat{\Delta}(x)$ and $\widehat{\mu}(x)$ be  estimators of  $\Delta(x)$ and $\mu(x)$ obtained from $\{(Y_i,T_i,W_i) : 1 \leq i \leq n \}$, an independent and identically distributed sample of size $n$ from $(Y,T,W)$.

\begin{definition}[Empirical $u$-SPE]
The estimator of $\Delta_{\mu}^*$ is
$$
\widehat{\Delta_{\mu}^*}(u)  := \hat{\Delta}_{\hat \mu}^*(u) = \inf\{ \delta \in \mathbb{R} : F_{\hat \Delta, \hat \mu}(\delta)  \geq u \}, \ \ F_{\hat \Delta, \hat \mu}(\delta) = \Ep_{\widehat \mu}[1\{\widehat{\Delta}(X) \leq \delta \}] =: \widehat{F_{\Delta,\mu}}(\delta).
$$
Then the empirical SPE-function is
$$
\{u \mapsto \widehat{\Delta_{\mu}^*}(u) : u \in \mathcal{U} \}, \ \ \mathcal{U} \subseteq [0,1],
$$
where $\mathcal{U}$ is the set of indices of interest that typically excludes tail indices and satisfies other technical conditions stated in Section \ref{sec:theory2}.
\end{definition}

\noindent \textbf{Example 1} (Binary response model, cont.) \
The estimator of the PE is
\begin{equation*}\label{eq:pe-probit}
\hat \Delta(x) = F\left(P(1,w) \transp \hat \beta\right) - F\left(P(0,w)\transp \hat \beta \right),
\end{equation*}
where $\hat \beta$ is the maximum likelihood (ML) estimator  of $\beta$,
$$
\hat \beta \in \arg \max_{b \in \mathbb{R}^{d_p}} \sum_{i=1}^n [Y_i \log F(P(T_i,W_i) \transp b) + (1-Y_i)\log \{1 - F(P(T_i,W_i) \transp b)\}], \ \ d_p = \dim P(T,W).
$$\qed

\noindent \textbf{Example 2} (Interactive linear model with additive error, cont.) \
The estimator of the PE is
\begin{equation*}\label{eq:pe-probit}
\hat \Delta(x) =  P(1,w) \transp \hat \beta -  P(0,w)\transp \hat \beta,
\end{equation*}
where $\hat \beta$ is the ordinary least squares (OLS) estimator of $\beta$,
$$
\hat \beta \in \arg \min_{b \in \mathbb{R}^{d_p}} \sum_{i=1}^n [Y_i - P(T_i,W_i) \transp b]^2, \ \ d_p = \dim P(T,W). 
$$\qed

\noindent \textbf{Example 3} (Linear model with non-additive error, cont.) \
The estimator of the PE is
\begin{equation*}\label{eq:pe-probit}
\hat \Delta(x) =  P(1,w) \transp \hat \beta(\tau) -  P(0,w)\transp \hat \beta(\tau),
\end{equation*}
where $\hat \beta(\tau)$ is the \citen{Koenker:1978} quantile regression (QR) estimator of $\beta(\tau)$,
$$
\hat  \beta(\tau) \in \arg \min_{b \in \mathbb{R}^{d_p}} \sum_{i=1}^n \rho_{\tau}(Y_i - P(T_i,W_i) \transp b), \ \ d_p = \dim P(T,W), \ \ \rho_{\tau}(v) = (\tau - 1\{v < 0\})v.
$$
\qed

\begin{remark}[Estimation of $\mu$]\label{re:cov}
Let $S$ denote the indicator for an observational unit belonging to the subpopulation of interest. For example, if $S=T$,  then $S=1$ indicates  the unit is in the subpopulation of the treated and $S = 0$ indicates  the unit is in the subpopulation of the untreated.   The indicator $S$ can also incorporate other restrictions, for example $S = 1\{X \in \mathcal{X} \}$ restricts the support of covariate $X$ to the region $\mathcal{X}$.    Finally, if $S$ is always $1$, then this means that we work with the entire population. Estimation of $\mu$ can be done using the empirical distribution:
$$
\hat \mu(x) =  \sum_{i = 1}^n S_i1\{X_i \leq x \} / \sum_{i = 1}^nS_i,
$$
provided that $ \sum_{i = 1}^nS_i > 0$. An alternative would be to use  the smoothed empirical distribution.

If $\mu$ can be decomposed into known and unknown parts, then we only need to estimate the unknown parts. Thus,   $\mu = F_{T,W} \times F_\epsilon$ in Example 3, where $F_\epsilon$ is known to be the uniform distribution and $F_{T,W}$ is unknown, but can be estimated by the empirical distribution of $(T,W)$ in the part of the population of interest. \qed
\end{remark}

\subsection{Empirical CA} The empirical version of the $u$-CA classifies the observations in the sample using the empirical PEs and $u$-SPE, and computes the moments and distributions in the resulting most and least affected subsamples.

\begin{definition}[Empirical $u$-CA] The empirical $u^{th}$-classification analysis consists of 2 steps: (1) Assign all $Z_i$ with $\widehat \Delta(X_i) < \widehat{\Delta_{\mu}^*}(u)$ to the $u$-least affected subsample, and all $Z_i$ with $\widehat \Delta(X_i) > \widehat{\Delta_{\mu}^*}(1-u) $ to  the $u$-most affected subsample.
(2) Estimate the moments and distribution of $Z$ in the least and most affected subpopulations by the empirical analogs in the least and most affected subsamples, i.e. $\widehat{\Lambda}_{\Delta,\mu}^{-u}(t) = \Lambda_{\widehat \Delta,\widehat \mu}^{-u}(t)$ and $\widehat{\Lambda}_{\Delta,\mu}^{+u}(t) = \Lambda_{\widehat \Delta,\widehat \mu}^{+u}(t)$. For example, $\widehat{\Lambda}_{\Delta,\mu}^{-u}(t) = \Ep_{\widehat \mu}\left[Z^t \mid \widehat \Delta(X) <  \widehat{\Delta_{\mu}^*}(u)\right]$ estimates the $t$-moment of $Z$ in the $u$-least affected subpopulation and $\widehat{\Lambda}_{\Delta,\mu}^{-u}(t) = \Ep_{\widehat \mu}\left[1(Z \leq t) \mid \widehat \Delta(X)<  \widehat{\Delta_{\mu}^*}(u)\right]$ the  distribution  of $Z$ at $t$ in the $u$-least affected subpopulation. The corresponding estimators in the $u$-most affected subpopulation are constructed replacing  $\widehat \Delta(X)<  \widehat{\Delta_{\mu}^*}(u)$ by  $\widehat \Delta(X) >  \widehat{\Delta_{\mu}^*}(1-u)$ in the conditioning set.   Here we use the same notation as in Definition \ref{def:uca}.
\end{definition}

\subsection{Inference on SPE}

The main inferential result  for the SPE can be previewed as follows.  Assume that the PE function  $x \mapsto \Delta(X)$ is not locally flat in the sense that the norm of its gradient does not vanish anywhere over the support, and other regularity conditions stated in Section 4. Then,  the empirical SPE-process is $\sqrt{n}$-consistent and converges in distribution to a centered Gaussian process, namely
\begin{eqnarray*}\label{eq:fclt}
\sqrt{n}(\hat\Delta^*_{\hat \mu}(u)-\Delta^*_{\mu}(u))\rightsquigarrow Z_{\infty}(u) \ \ \text{in}  \ \ \ell^\infty(\UU),
\end{eqnarray*}
the metric space of bounded functions on $\UU$, as a stochastic process  indexed by $u \in \UU$, where $\UU$ is  a compact subset of $(0,1)$.  Moreover, the exchangeable bootstrap algorithm specified in Algorithm \ref{alg:computation} estimates consistently the law of $Z_{\infty}(u)$. 

The next corollary to Theorem \ref{thm:fclt} in Section \ref{sec:theory2} provides uniform bands that cover the  SPE-function simultaneously over a region of values of $u$ with prespecified probability in large samples. It does cover pointwise confidence bands for the SPE-function at a specific quantile index $u$ as a special case by simply taking $\UU$ to be the singleton set $\{u\}$.

\begin{corollary}[Inference on SPE-function using Limit Theory and Bootstrap]\label{cor:cb} Under the assumptions of Theorem \ref{thm:fclt}, for any $0 < \alpha < 1,$
$$\Pr \left\{ \Delta^*_{\mu}(u) \in \left[\widehat{\Delta^*_{\mu}}(u) - \widehat t_{1-\alpha}(\UU) \hat \Sigma(u)^{1/2} / \sqrt{n},
\widehat{\Delta^*_{\mu}}(u) + \widehat t_{1-\alpha}(\UU) \hat \Sigma(u)^{1/2} /\sqrt{n}\right] : u \in \UU \right\} \to 1 - \alpha,$$
where $\widehat t_{1-\alpha}(\UU)$ is any consistent estimator of $t_{1-\alpha}(\UU)$, the $(1-\alpha)$-quantile of  $$t(\UU) := \sup_{u \in \UU}|Z_{\infty}(u)|   \Sigma(u)^{-1/2},$$ and $u \mapsto \hat \Sigma(u)$ is a uniformly consistent estimator of $u \mapsto \Sigma(u)$, the variance function of $u \mapsto Z_{\infty}(u)$. We provide consistent estimators of $t_{1-\alpha}(\UU)$ and $u \mapsto \Sigma(u)$ in Algorithm \ref{alg:computation}.
\end{corollary}

We  now describe a practical bootstrap algorithm to estimate the quantiles of  $t(\UU)$. Let $(\omega_{1}, \ldots, \omega_{n})$ denote the bootstrap weights, which are nonnegative random variables  independent of the data obeying the conditions stated in \citen{vdV-W}. For example, $(\omega_{1}, \ldots, \omega_{n})$ is a multinomial vector with dimension $n$ and probabilities $(1/n,\ldots,1/n)$ in the empirical bootstrap.
 In what follows $B$ is the number of bootstrap draws, such that $B \to \infty$.  In our experience, setting $B \geq 500$ suffices for good accuracy.

\begin{algorithm}[Bootstrap law of $t(\UU)$ and its quantiles]\label{alg:computation}
\textbf{1)} Draw a realization of the bootstrap weights $(\omega_{1}, \ldots, \omega_{n})$.
\textbf{2)}  For each $u \in \UU$, compute $\widetilde{\Delta_{\mu}^*}(u) = \widetilde \Delta_{\widetilde \mu}^*(u)$, a bootstrap draw of $\widehat{\Delta_{\mu}^*}(u) = \widehat \Delta_{\widehat \mu}^*(u)$, where $\widetilde \Delta$ and $\widetilde \mu$ are the bootstrap versions of $\widehat \Delta$ and $\widehat \mu$ that use $(\omega_{1}, \ldots, \omega_{n})$ as sampling weights in the computation of the estimators. Construct a bootstrap draw of $Z_{\infty}(u)$ as $\widetilde Z_{\infty}(u) = \sqrt{n}( \widetilde{\Delta_{\mu}^*}(u) - \widehat{\Delta_{\mu}^*}(u))$.
\textbf{3)} Repeat steps (1)-(2) $B$ times. \textbf{4)} For each $u \in \UU$, compute a bootstrap estimator  of $\Sigma(u)^{1/2}$ such as the bootstrap interquartile range rescaled with the normal distribution, $\widehat \Sigma(u)^{1/2} = (q_{0.75}(u) - q_{0.25}(u))/(z_{0.75} - z_{0.25}),$ where $q_p(u)$ is the $p$th sample quantile of $\widetilde Z_{\infty}(u)$ in the $B$ draws and $z_p$ is the $p$th quantile of $N(0,1)$. \textbf{5)}  Use the empirical distribution of
$\widetilde t(\UU) = \sup_{u \in \UU}  |\widetilde Z_{\infty}(u)| \widehat \Sigma(u)^{-1/2}$ across the $B$ draws to approximate the distribution of $t(\UU) = \sup_{u \in \UU}  |Z_{\infty}(u)| \Sigma(u)^{-1/2}$. In particular, construct $\widehat t_{1-\alpha}(\UU)$, an estimator of $t_{1-\alpha}(\UU)$, as the $(1-\alpha)$-quantile of the $B$ draws of
$\widetilde t(\UU)$.  \end{algorithm}

\begin{remark}[Monotonization of the bands]\label{re:mono} While the SPE-function $u \mapsto \Delta^*_{\mu}(u)$ is increasing by definition,  the end functions of the confidence band $u \mapsto \widehat{\Delta^*_{\mu}}(u) \pm\widehat t_{1-\alpha}(\UU) \widehat \Sigma(u)^{1/2} / \sqrt{n}$ might not be increasing. \citen{cfg-09} showed that monotonizing the end functions via rearrangement reduces the width of the band in uniform norm, while increases coverage in finite-samples. We use this refinement in the empirical examples.\footnote{In practice, the rearrangement simply consists in sorting the two vectors containing the discretized version of the end-functions in increasing order; see  \citen{cfg-09} for more  details.} \qed
\end{remark}

\begin{remark}[Finite-Sample Bias Corrections]\label{re:fsc} The empirical $u$-SPE might be biased in small samples, specially at the the tails. Bootstrap is also useful to improve the estimator and confidence bands.  Thus, a corrected estimator can be formed as $
2 \widehat{\Delta_{\mu}^*} - \overline{\Delta_{\mu}^*},
$ and a corrected $(1-\alpha)$-confidence band as $
\left[2 \widehat{\Delta_{\mu}^*} - \overline{\Delta_{\mu}^*} \pm \widehat t_{1-\alpha}(\UU) \widehat \Sigma(u)^{1/2} / \sqrt{n} \right]
$, where $\overline{\Delta_{\mu}^*}$ is the mean of the bootstrap draw of the estimator. In Appendix \ref{subset:numerical} of the SM, we show that this correction reduces the bias of the estimator and increases the coverage of the confidence bands in a simulation calibrated to the gender wage gap application. \qed
\end{remark}

\subsection{Inference on CA} Let $\Lambda_{\Delta,\mu}^{u}(t) := [\Lambda_{\Delta,\mu}^{-u}(t),\Lambda_{\Delta,\mu}^{+u}(t)]$ and $\widehat{\Lambda}_{\Delta,\mu}^{u}(t) := [\widehat{\Lambda}_{\Delta,\mu}^{-u}(t), \widehat{\Lambda}_{\Delta,\mu}^{+u}(t)]$. The main inferential result for CA can be previewed as follows: the empirical CA-process converges in distribution to a centered bivariate Gaussian  process, namely
\begin{eqnarray}\label{eq:fclt}
\sqrt{n}(\widehat{\Lambda}_{\Delta,\mu}^{u}(t)-\Lambda_{\Delta,\mu}^{u}(t))\rightsquigarrow Z_{\infty}^u(t) \ \  \text{in}  \ \ \ell^\infty(\mathbb{R}^{d_z})^2,
\end{eqnarray}
as a stochastic process  indexed by $t \in \mathbb{R}^{d_z}$.
Moreover,  exchangeable bootstrap  estimates consistently the law of $Z_{\infty}^{u}(t)$.  

The next corollary to Theorem \ref{thm:fclt_supp1} in Section \ref{sec:theory2} provides uniform bands that cover $L$ linear combinations of the $2$-dimensional vector $\Lambda_{\Delta,\mu}^{u}(t)$ with coefficients $c_1, \ldots, c_L$  simultaneously over $t \in \mathcal{T}$ with prespecified probability in large samples. It covers pointwise confidence intervals for the mean of the $k^{th}$ component of $Z$ for least affected as a special case with $L=1$  linear combination, $c_1 = (1,0)'$, and $\mathcal{T} = \{e_k\}$, where $e_k$ is a unit vector with a one in the $k^{th}$ position. Joint confidence intervals for $s$ differences of means of the $k_1^{th}$, $\ldots$, $k_s^{th}$ components of $Z$ between most and least affected are a special case with $L=1$ linear combination, $c_1 = (-1,1)'$, and $\mathcal{T} = \{e_{k_1}, \ldots, e_{k_s}\}$. Joint uniform bands for the distribution of the $k^{th}$ component of $Z$ for most and least affected are also a special case with $L=2$ linear combinations, $c_1 =(1,0)$, $c_2 = (0,1)$,  and $\mathcal{T} = \{ t \in \mathbb{R}^{d_z}: t_j = \bar T, j \neq k \}$, where $\bar T$ is an arbitrarily large number.
By appropriate choice of the linear combinations and the index set $T$, we can therefore conduct multiple tests while preserving the significance level from simultaneous inference problems  \cite{rsw10b,rsw10,lsx16}. We show examples in the empirical application of Section \ref{sec:empirics}.

\begin{corollary}[Inference on CA-function using Limit Theory and Bootstrap]\label{cor:cb2} Under the assumptions of Theorem \ref{thm:fclt_supp1}, for any $0 < \alpha < 1$,
$$\Pr \left\{c_{\ell}'  \Lambda_{\Delta,\mu}^{u}(t) \in c_{\ell}'\widehat{\Lambda}_{\Delta,\mu}^{u}(t) \pm \widehat t^u_{1-\alpha}(\mathcal{T},L) [c_{\ell}'\widehat \Sigma^u(t)c_{\ell}]^{1/2} / \sqrt{n}: t \in \mathcal{T}, \ell = 1,\ldots,L \right\} \to 1 - \alpha,$$
where $\widehat t^u_{1-\alpha}(\mathcal{T},L)$ is any consistent estimator of $t^u_{1-\alpha}(\mathcal{T},L)$, the $(1-\alpha)$-quantile of  $$t^u(\mathcal{T},L) := \sup_{t \in \mathcal{T}, \ell = 1,\ldots,L}|c_{\ell}'Z_{\infty}^u(t)|   [c_{\ell}'\Sigma^u(t)c_{\ell}]^{-1/2},$$ and $t \mapsto \widehat \Sigma^u(t)$ is a uniformly consistent estimator of $t \mapsto \Sigma^u(t)$, the variance function of $t \mapsto Z_{\infty}^u(t)$. A p-value of the null hypothesis $c_{\ell}'  \Lambda_{\Delta,\mu}^{u}(t) = r_{\ell}(t)$ for all $t \in \mathcal{T}$ and $\ell = 1,\ldots,L$ of the realization of the statistic $\sup_{t \in \mathcal{T}, \ell = 1,\ldots,L}|c_{\ell}'\widehat{\Lambda}_{\Delta,\mu}^{u}(t) - r_{\ell}(t)|   [c_{\ell}'\widehat \Sigma^u(t)c_{\ell}]^{-1/2} = s$ is
$$
S_{t^u(\mathcal{T},L)}(s) = \Pr \left(t^u_{1-\alpha}(\mathcal{T},L)>  s  \right).
$$
 We provide consistent estimators of $t^u_{1-\alpha}(\mathcal{T},L)$, $u \mapsto \Sigma^u(t)$ and $S_{t^u(\mathcal{T},L)}(t)$ in Algorithm \ref{alg:ca}.
\end{corollary}

\begin{algorithm}[Bootstrap law of $t(\mathcal{T},L)$, quantiles and p-values]\label{alg:ca}
\textbf{1)} Draw a realization of the bootstrap weights $(\omega_{1}, \ldots, \omega_{n})$.
\textbf{2)}  For each $t \in \TT$, compute $\widetilde{\Lambda}_{\Delta,\mu}^{-u}(t) = \Lambda_{\widetilde \Delta,\widetilde \mu}^{-u}(t)$ and $\widetilde{\Lambda}_{\Delta,\mu}^{+u}(t) = \Lambda_{\widetilde \Delta,\widetilde \mu}^{+u}(t)$,  a bootstrap draw of $\widehat{\Lambda}_{\Delta,\mu}^{-u}(t) = \Lambda_{\widehat \Delta,\widehat \mu}^{-u}(t)$ and $\widehat{\Lambda}_{\Delta,\mu}^{+u}(t) = \Lambda_{\widehat \Delta,\widehat \mu}^{+u}(t)$, where $\widetilde \Delta$ and $\widetilde \mu$ are the bootstrap versions of $\widehat \Delta$ and $\widehat \mu$ that use $(\omega_{1}, \ldots, \omega_{n})$ as sampling weights in the computation of the estimators. Construct a bootstrap draw of $Z_{\infty}^u(t)$ as $\widetilde Z^u_{\infty}(t) = \sqrt{n}(\widetilde{\Lambda}_{\Delta,\mu}^{u}(t) - \widehat{\Lambda}_{\Delta,\mu}^{u}(t))$, where $\widetilde{\Lambda}_{\Delta,\mu}^{u}(t) = [\widetilde{\Lambda}_{\Delta,\mu}^{-u}(t),\widetilde{\Lambda}_{\Delta,\mu}^{+u}(t)]$.
\textbf{3)} Repeat steps (1)-(2) $B$ times. \textbf{4)} For each $t \in \TT$ and $\ell = 1,\ldots,L$, compute a bootstrap estimator  of $[c_{\ell}'\Sigma^u(t)c_{\ell}]^{1/2}$ such as the bootstrap interquartile range rescaled with the normal distribution $[c_{\ell}'\widehat \Sigma^u(t)c_{\ell}]^{1/2} = (q^u_{0.75}(t,\ell) - q^u_{0.25}(t,\ell))/(z_{0.75} - z_{0.25}),$ where $q^u_p(t,\ell)$ is the $p$th sample quantile of $c_{\ell}'\widetilde Z^u_{\infty}(t)$ in the $B$ draws and $z_p$ is the $p$th quantile of $N(0,1)$. \textbf{5)}  Use the empirical distribution of
$\widetilde t(\TT,L) = \sup_{t \in \TT, \ell=1,\ldots,L}  |c_{\ell}'\widetilde Z^u_{\infty}(t)| [c_{\ell}'\widehat \Sigma^u(t)c_{\ell}]^{-1/2}$ across the $B$ draws to approximate the distribution of $t(\TT,L) = \sup_{t \in \TT,\ell = 1,\ldots,L}  |c_{\ell}'Z^u_{\infty}(t)| [c_{\ell}'\Sigma^u(t)c_{\ell}]^{-1/2}$. In particular, construct $\widehat t_{1-\alpha}(\TT,L)$, an estimator of $t_{1-\alpha}(\TT,L)$, as the $(1-\alpha)$-quantile of the $B$ draws of
$\widetilde t(\TT,L)$, and an estimation of  the p-value $S_{t^u(\mathcal{T},L)}(s)$ as the proportion of the $B$ draws of $\widetilde t(\TT,L)$ that are greater than $s$.
\end{algorithm}

\subsection{Inference on Most and Least Affected Subpopulations} In addition to moments and distributions, we can conduct inference on the subpopulations of most and least affected.\footnote{Here
we follow the set inference approach described in  \citen{ckm15}, which builds on \citen{cht:bounds}. In addition our results justify the use of subsapling-based methods as in \citen{cht:bounds} and \citen{RSset}.} Let 
$$
\mathcal{M}^{-u}:=\{ (x,y) \in \mathcal{Z}:  \Delta(x) \leq \Delta^*_\mu (u) \}, \quad \mathcal{M}^{+u} := \{ (x,y) \in \mathcal{Z}
: \Delta(x) \geq \Delta^*_\mu(1-u) \},
$$
be the sets representing the $u$-least and $u$-most affected subpopulation, respectively.  Here we assume
that   $\mathcal{Z}$ is compact or that the support of $(X,Y)$ has been intersected with a compact set to form $\mathcal{Z}$. 
We can construct an outer $(1-\alpha)$-confidence set for $\mathcal{M}^{-u}$ as\footnote{Note that we can also similarly construct  an inner confidence region, which is the complement of the outer confidence region of $\mathcal{X}\setminus \mathcal{M}^{-u}$, see \citen{ckm15} for relevant discussion.}
$$
\mathcal{CM}^{-u}(1-\alpha) = \{ (x,y) \in \mathcal{Z}:   \widehat \Sigma^{-1/2}(x,u) \sqrt{n}[ \widehat \Delta(x) - \widehat \Delta^*_\mu (u) ] \leq  \widehat c(1-\alpha) \},
$$
where $\widehat c(1-\alpha)$ is a consistent estimator of $c(1-\alpha)$, the $(1-\alpha)$-quantile of the random
variable
$$
V_{\infty} = \sup_{ \{x \in \mathcal{X}: \Delta(x) =  \Delta^*_\mu (u)\}}  \Sigma^{-1/2}(x,u)  [ G_\infty (x) - Z_\infty(u) ],
$$
and $x \mapsto  \hat \Sigma(x,u)$ is a uniformly consistent estimator of  
$x \mapsto \Sigma(x,u)$, the variance function of the process $G_\infty(x) - Z_\infty(u)$
defined in Section 4. The estimator $\widehat c(1-\alpha)$ can be obtained as the $(1-\alpha)$-quantile of the bootstrap version of  $V_{\infty}$,
$$
\widetilde V^*_{\infty} = \sup_{ \{ x \in \mathcal{X}:  \widehat \Delta(x) =  \widehat \Delta^*_\mu (u) \}}   \hat \Sigma^{-1/2}(x,u)  \sqrt{n} \Big (  [\widetilde \Delta (x) - \widetilde{\Delta_{\mu}^*}(u)] - [ \widehat \Delta (x) -
 \widehat \Delta^*_\mu(u)] \Big ),
$$
where $\widetilde \Delta(x)$ and $\widetilde{\Delta_{\mu}^*}(u)$ are defined as in Algorithm \ref{alg:computation}.  A similar $(1-\alpha)$-confidence set, $\mathcal{CM}^{+u}(1-\alpha),$ can be constructed for $\mathcal{M}^{+u}$. These sets can be visualized by plotting all 2 or 3 dimensional projections of their elements.   We provide an example of such plots in Section \ref{sec:empirics}. 
An immediate consequence of the set  inference results in \citen{ckm15} and the results of this paper is the following corollary:

\begin{corollary}[Inference on Most and Least Affected Subpopulations]The sets $\mathcal{CM}^{-u}(1-\alpha)$  and $\mathcal{CM}^{+u}(1-\alpha)$  cover $\mathcal{M}^{-u}$ and $\mathcal{M}^{+u}$ with probability approaching $1-\alpha$, and $\mathcal{CM}^{-u}(1-\alpha)$  and $\mathcal{CM}^{+u}(1-\alpha)$  are consistent in the sense that they approach to $\mathcal{M}^{-u}$ and $\mathcal{M}^{+u}$ at a  $\sqrt{n}$-rate with respect to the Hausdorff distance.
\end{corollary}
%



\section{Empirical Analysis of the Gender Wage Gap}\label{sec:empirics}

We report the main results of the application to the gender wage gap using data from the U.S. March Supplement of the Current Population Survey (CPS) in 2015.  In Appendix \ref{subset:numerical} of the SM, we complement the analysis with supporting results from a simulation calibrated to this application. There, we find that our estimation and inference methods perform well in finite samples that closely mimic the characteristics of the CPS data. This exercise serves to indirectly verify the plausibility of the main regularity conditions mentioned in Section \ref{sec:model} and formally stated in Section \ref{sec:theory2}.

Our  sample consists of white, non-hispanic individuals who are aged 25 to 64 years and work more than 35 hours per week during at least 50 weeks of the year. We  exclude self-employed workers; individuals living in group quarters; individuals in the military, agricultural or private household sectors;  individuals with inconsistent reports on earnings and employment status;  individuals with allocated or missing information in any of the variables used in the analysis; and individuals with hourly wage rate below $\$3$.  The resulting  sample  contains $32,523$ workers including $18,137$  men and $14,382$ of women.

We estimate interactive linear models with additive and non-additive errors, using mean and quantile
regressions, respectively. The outcome variable $Y$ is the logarithm of the hourly wage rate constructed as the ratio of the annual earnings to the total number of hours worked, which is constructed in turn as the product of number of weeks worked and the usual number of hours worked per week.
The key covariate $T$ is an indicator for female worker, and the control variables $W$ include 5 marital status indicators (widowed, divorced, separated, never married, and married); 5 educational attainment indicators (less than high school graduate, high school graduate, some college, college graduate, and advanced degree); 4 region indicators (midwest, south, west, and northeast);  a quartic in potential experience constructed as the maximum of age minus years of schooling minus 7 and zero, i.e., $experience = \max(age-education-7,0)$; 5 occupation indicators (management, professional and related; service; sales and office; natural resources, construction and maintenance; and production, transportation and material moving); 12 industry indicators (mining, quarrying, and oil and gas extraction; construction; manufacturing; wholesale and retail trade; transportation and utilities;  information; financial services; professional and business services;  education and health services;  leisure and hospitality; other services; and public administration); and all the two-way interactions between the education, experience, occupation and industry variables except for the occupation-industry interactions.\footnote{The sample selection criteria and the variable construction follow  \citen{Mulligan-Rubinstein-08}. The occupation and industry categories follow the 2010 Census Occupational Classification and 2012 Census Industry Classification, respectively.}  All  calculations use the CPS sampling weights to account for nonrandom sampling in the March CPS.

\begin{table}[ht]\caption{Descriptive Statistics of Workers}\label{table1}
\centering
\footnotesize{\begin{tabular}{lccclccc}
  \hline\hline
 & All & Women & Men &  & All & Women & Men\\
  \hline
 Log wage & 3.15 & 3.02 & 3.25 &  O.manager & 0.48 & 0.55 & 0.43 \\
  Female & 0.44 & 1.00 & 0.00 & O.service & 0.10 & 0.10 & 0.09 \\
  MS.married & 0.65 & 0.61 & 0.68 & O.sales & 0.23 & 0.31 & 0.16 \\
  MS.widowed & 0.01 & 0.02 & 0.01 & O.construction & 0.09 & 0.01 & 0.15 \\
  MS.separated & 0.02 & 0.02 & 0.02 & O.production & 0.11 & 0.04 & 0.17 \\
  MS.divorced & 0.13 & 0.16 & 0.10  &   I.minery & 0.03 & 0.01 & 0.04 \\
  MS.Nevermarried & 0.19 & 0.18 & 0.20 & I.construction & 0.06 & 0.01 & 0.09 \\
  E.lhs & 0.02 & 0.02 & 0.03 & I.manufacture & 0.14 & 0.08 & 0.18 \\
  E.hsg & 0.25 & 0.21 & 0.28 & I.retail & 0.13 & 0.11 & 0.14 \\
  E.sc & 0.28 & 0.29 & 0.27 & I.transport & 0.04 & 0.02 & 0.06 \\
  E.cg & 0.28 & 0.30 & 0.27 & I.information & 0.02 & 0.02 & 0.03 \\
  E.ad & 0.16 & 0.18 & 0.15 &  I.finance & 0.08 & 0.10 & 0.07 \\
  R.northeast & 0.19 & 0.19 & 0.19 & I.professional & 0.11 & 0.10 & 0.13 \\
  R.midwest & 0.27 & 0.28 & 0.27 &   I.education & 0.24 & 0.40 & 0.11 \\
  R.south & 0.35 & 0.35 & 0.35 & I.leisure & 0.05 & 0.05 & 0.04 \\
  R.west & 0.18 & 0.18 & 0.19 & I.services & 0.03 & 0.03 & 0.04 \\
  Experience & 21.68 & 21.72 & 21.65 & I.public & 0.07 & 0.06 & 0.07 \\
   \hline\hline
   \multicolumn{8}{l}{\footnotesize{Source: March Supplement CPS 2015.}}
\end{tabular}}
\end{table}

Table \ref{table1} reports sample means of the variables used in the analysis. Working women are more highly educated than working men,  have about the same potential experience, and are less likely to be married and more likely to be divorced. They work relatively more often in managerial and sales occupations and in the industries providing  education and health services. Working men  are relatively more likely to work in construction and production occupations within non-service industries. The unconditional gender wage gap is 23\%.

Figure \ref{fig:gender} of Section \ref{sec:intro}  plots estimates and 90\% confidence bands for the APE and SPE-function on the treated (women) of the conditional gender wage gap using  additive and non-additive error models. The  PEs are obtained as described in Examples \ref{example:cef} and  \ref{example:cqf} with $P(T,W) = (TW,(1-T)W)$. In this case $\dim P(T,W) = 332$, which makes it very difficult to identify any pattern about the gender wage gap just by looking at the regression coefficients.
The distribution $F_{T,W}$ is estimated by the empirical distribution of $(T,W)$ for women, and $F_{\epsilon}$ is approximated by a uniform distribution over the grid $\{.02, .03, \ldots, .98\}$. The confidence bands are constructed using Algorithm \ref{alg:computation} with standard exponential weights (weighted bootstrap) and $B=500$, and are uniform for the SPE-function over the grid $\mathcal{U} = \{.01, .02, \ldots, .98\}$. We monotonize the bands using the rearrangement method described in Remark \ref{re:mono}, and  implement the finite sample corrections described in Remark \ref{re:fsc}.
After controlling for worker characteristics, the gender wage gap for women remains on average around 20\%.  More importantly, we
uncover a striking amount of heterogeneity, with the PE ranging between -6.5 and  40\% in the additive error model and between -14 and 54\% in the non-additive error model.\footnote{In the 2016 version of the paper we found similar patterns of heterogeneity using CPS 2012 data with a specification that did not include occupation and industry indicators.}


\begin{table}[ht]\caption{\label{table:class} Classification Table --  Average Characteristics of the 10\% Least and Most Affected Women by Gender Wage Gap}
\centering
\footnotesize{\begin{tabular}{lcccclcccc}
  \hline\hline
  &  \multicolumn{2}{c}{10\% Least} &  \multicolumn{2}{c}{10\% Most} & & \multicolumn{2}{c}{10\% Least} &  \multicolumn{2}{c}{10\% Most} \\
 & Est. & S.E. & Est. & S.E. & & Est. & S.E. & Est. & S.E \\
  \hline
  Log wage & 3.08 & 0.03 & 2.97 & 0.03 &   O.manager & 0.67 & 0.04 & 0.38 & 0.04 \\
  M.married & 0.28 & 0.03 & 0.87 & 0.02 & O.service & 0.08 & 0.02 & 0.10 & 0.02 \\
  M.widowed & 0.02 & 0.01 & 0.01 & 0.01 & O.sales & 0.19 & 0.03 & 0.42 & 0.04 \\
  M.separated & 0.02 & 0.01 & 0.01 & 0.00 & O.construction & 0.02 & 0.01 & 0.01 & 0.00 \\
  M.divorced & 0.15 & 0.02 & 0.07 & 0.02 &  O.production & 0.03 & 0.01 & 0.09 & 0.02 \\
  M.nevermarried & 0.52 & 0.03 & 0.04 & 0.01 &  I.minery & 0.01 & 0.01 & 0.02 & 0.01 \\
  E.lhs & 0.01 & 0.01 & 0.06 & 0.01 & I.construction & 0.02 & 0.01 & 0.02 & 0.01 \\
  E.hsg & 0.08 & 0.02 & 0.30 & 0.04 & I.manufacture & 0.02 & 0.01 & 0.11 & 0.02 \\
  E.sc & 0.15 & 0.03 & 0.23 & 0.04 & I.retail & 0.06 & 0.02 & 0.19 & 0.03 \\
  E.cg & 0.37 & 0.04 & 0.17 & 0.03 &   I.transport & 0.01 & 0.01 & 0.04 & 0.01 \\
  E.ad & 0.39 & 0.04 & 0.24 & 0.03 & I.information & 0.04 & 0.01 & 0.05 & 0.01 \\
  R.ne & 0.24 & 0.02 & 0.18 & 0.02 & I.finance & 0.03 & 0.02 & 0.09 & 0.03 \\
  R.mw & 0.26 & 0.02 & 0.28 & 0.02 &  I.professional & 0.06 & 0.02 & 0.04 & 0.02 \\
  R.so & 0.31 & 0.02 & 0.39 & 0.03 & I.education & 0.46 & 0.04 & 0.33 & 0.04 \\
  R.we & 0.19 & 0.02 & 0.15 & 0.02 & I.leisure & 0.10 & 0.03 & 0.01 & 0.01 \\
  Experience & 13.05 & 1.03 & 26.32 & 0.75 &  I.services & 0.09 & 0.02 & 0.01 & 0.01 \\
  & & & & &   I.public & 0.09 & 0.02 & 0.08 & 0.02 \\
   \hline\hline
           \multicolumn{10}{l}{\footnotesize{PE estimated from a linear conditional quantile model with interactions.}} \\
     \multicolumn{10}{l}{\footnotesize{Standard Errors obtained by weighted bootstrap with 500 repetitions.}} \\
\end{tabular}}
\end{table}

Table \ref{table:class} shows the results of a classification analysis, exhibiting characteristics of women that are most and least affected by the gender wage gap together with standard errors obtained by weighted bootstrap. We focus here on the non-additive model, but the results from the additive model are similar. Since the PE are predominantly negative, we define the most affected as $\Delta(X) < \Delta^*_{\mu}(u)$ and the lest affected as $\Delta(X) > \Delta^*_{\mu}(1-u)$ to facilitate the interpretation.  According to this model the 10\% of the women \textit{most affected} by the gender wage gap on average earn lower wages, are much more likely to be married, much less likely to be never married, have lower education, live in the South, possess much more potential experience, are more likely to have sales and non managerial occupations, and work more often in manufacture and retail and less often in education industries than the 10\% least affected women.

\begin{table}[ht]\caption{\label{table:class2} Classification Table -- Difference in the Average Characteristics of the 10\% Most and Least Affected Women by Gender Wage Gap}
\centering
\footnotesize{\begin{tabular}{lcccclcccc}
  \hline
 & Est. & S.E. & P-val.$^1$ & JP-val.$^2$ & & Est. & S.E. & P-val.$^1$ & JP-val.$^2$ \\
  \hline
Log wage & -0.10 & 0.04 & 0.03 & 0.70 &   O.manager & -0.29 & 0.06 & 0.00 & 0.00 \\
  M.married & 0.59 & 0.04 & 0.00 & 0.00 & O.service & 0.02 & 0.03 & 0.99 & 1.00 \\
  M.widowed & -0.02 & 0.02 & 0.93 & 1.00 & O.sales & 0.22 & 0.06 & 0.00 & 0.03 \\
  M.separated & -0.01 & 0.01 & 0.89 & 1.00 & O.construction & -0.01 & 0.01 & 0.67 & 1.00 \\
  M.divorced & -0.08 & 0.04 & 0.46 & 0.86 & O.production & 0.06 & 0.02 & 0.08 & 0.42 \\
  M.nevermarried & -0.48 & 0.04 & 0.00 & 0.00 & I.minery & 0.01 & 0.01 & 0.97 & 1.00 \\
  E.lhs & 0.05 & 0.01 & 0.01 & 0.17 & I.construction & -0.00 & 0.01 & 1.00 & 1.00 \\
  E.hsg & 0.22 & 0.05 & 0.00 & 0.00 &  I.manufacture & 0.08 & 0.02 & 0.02 & 0.15 \\
  E.sc & 0.08 & 0.06 & 0.70 & 1.00 &   I.retail & 0.12 & 0.04 & 0.06 & 0.32 \\
  E.cg & -0.19 & 0.06 & 0.01 & 0.16 &  I.transport & 0.04 & 0.01 & 0.08 & 0.39 \\
  E.ad & -0.15 & 0.06 & 0.07 & 0.46 &  I.information & 0.01 & 0.02 & 1.00 & 1.00 \\
  R.ne & -0.06 & 0.04 & 0.35 & 0.99 &  I.finance & 0.06 & 0.04 & 0.78 & 0.99 \\
  R.mw & 0.02 & 0.04 & 0.95 & 1.00 & I.professional & -0.01 & 0.03 & 1.00 & 1.00 \\
  R.so & 0.08 & 0.04 & 0.23 & 0.97 & I.education & -0.13 & 0.06 & 0.29 & 0.74 \\
  R.we & -0.04 & 0.03 & 0.69 & 1.00 & I.leisure & -0.09 & 0.03 & 0.04 & 0.22 \\
  Experience & 13.27 & 1.54 & 0.00 & 0.00 &  I.services & -0.07 & 0.02 & 0.02 & 0.15 \\
   &  &  &  &  &  I.public & -0.01 & 0.03 & 1.00 & 1.00 \\
   \hline\hline
        \multicolumn{10}{l}{\footnotesize{PE estimated from a linear conditional quantile model with interactions.}} \\
     \multicolumn{10}{l}{\footnotesize{Standard Errors and p-values obtained by weighted bootstrap with 500 repetitions.}} \\
     \multicolumn{10}{l}{\footnotesize{$^1$ These p-values are adjusted for multiplicity to account
     for joint testing of zero coefficients }} \\
        \multicolumn{10}{l}{\footnotesize{ on for all variables within a category: M E, R, O, or I.}} \\
 
     \multicolumn{10}{l}{\footnotesize{$^2$ These p-values are adjusted for multiplicity to account
     for joint testing of zero coefficients }}\\
           \multicolumn{10}{l}{\footnotesize{     on all the variables in the table.}} \\
\end{tabular}}
\end{table}

Table \ref{table:class2} tests if the differences found in table \ref{table:class} are statistically significant. It reports p-values for the test of equality of means for most and least affected women. The first p-value accounts for simultaneous inference on all variables within a given category. For example, it accounts that we are conducting five tests corresponding to the five categories of marital status. For the non categorical variables log wage and experience  the p-values are for one test. The second p-value accounts for simultaneous inference of all the differences displayed in the table.\footnote{We employ the so called "single-step" methods for controlling the family-wise error rate. To generate a (somewhat) higher power, we recommend to employ the p-values generated via "step-down" methods, such as those reported in \citen{RWpvals} and \citen{lsx16}.} These p-values are obtained by Algorithm \ref{alg:ca} with the appropriate choice of vectors of linear combinations and set $\TT$, and 500 weighted bootstrap repetitions. The p-values show that most of the differences from table \ref{table:class} are statistically significant at conventional significant levels after controlling for simultaneous inference. In particular, the most affected women are significantly more likely to be married, high-school graduates, more experienced, and  in sales occupations, and less likely to be never married and in managerial occupations under the most strict simultaneous inference correction. \citen{bk17} have recently documented the importance of differences in occupation and industry to explain the gender wage gap using data from the Panel Study of Income Dynamics (PSID) 1980-2010 and a different methodology based on wage decompositions. Consistent with our findings, they argue that this importance might be due to compensating differentials.  Unlike \citen{bk17} and previous studies in the literature, our analysis uncovers significant heterogeneity in the extent of the gender wage gap and relates this heterogeneity to human capital, occupation, industry and other characteristics.



\begin{figure}[!h]
\centering
\includegraphics[width=\textwidth,height=.5\textwidth]{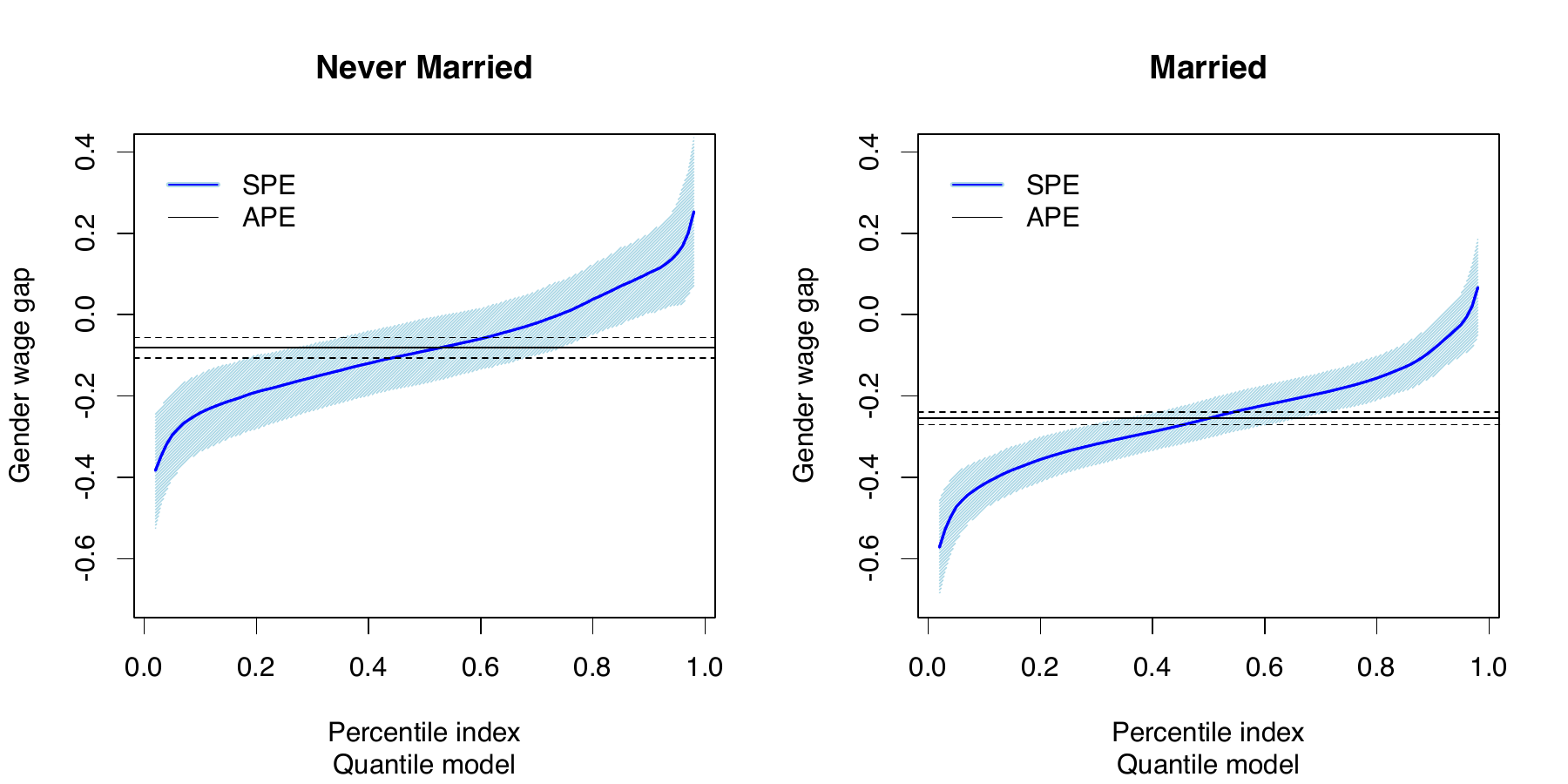}
\caption{APE and SPE  of the gender wage gap for women by marital status.  Estimates and 90\% bootstrap uniform confidence bands  for the conditional quantile function. }\label{fig:gender-ms}
\end{figure}

\begin{figure}[!h]
\centering
\includegraphics[width=\textwidth,height=.5\textwidth]{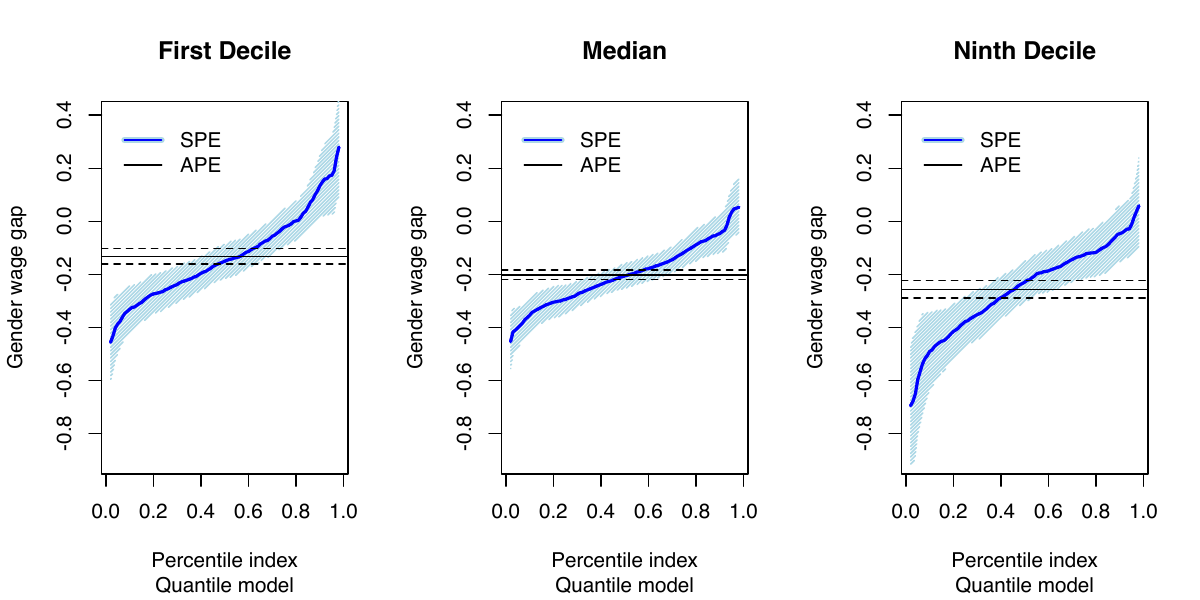}
\caption{APE and SPE  of the gender wage gap for women by unobserved ranking in the conditional distribution.  Estimates and 90\% bootstrap uniform confidence bands  for the conditional quantile function. }\label{fig:gender-rank}
\end{figure}

We further explore these findings by analyzing  the APE and SPE on the treated \textit{conditional} on
 marital status and unobserved rank in the non-additive error model. Figures \ref{fig:gender-ms} and \ref{fig:gender-rank} show estimates and 90\% confidence  bands of the APE and SPE-function of the gender wage gap for  2 subpopulations defined by marital status (married and never married) and 3 subpopulations defined by unobserved rank (first decile, median and ninth decile,  where the unobserved rank is .1, .5 and .9, respectively).  The confidence bands are constructed as in fig. \ref{fig:gender}.   We find significant heterogeneity in  the gender gap within each subpopulation, and also between subpopulations defined by marital status and unobserved rank.
The SPE-function is more negative  for married women and at the tails of the conditional distribution. Married women at the top decile suffer from the highest gender wage gaps.  This pattern is consistent with ``glass-ceiling'' effects behind the gender wage gap \cite{abv03}.

\begin{figure}[!h]
\centering
\includegraphics[width=.45\textwidth,height=.5\textwidth]{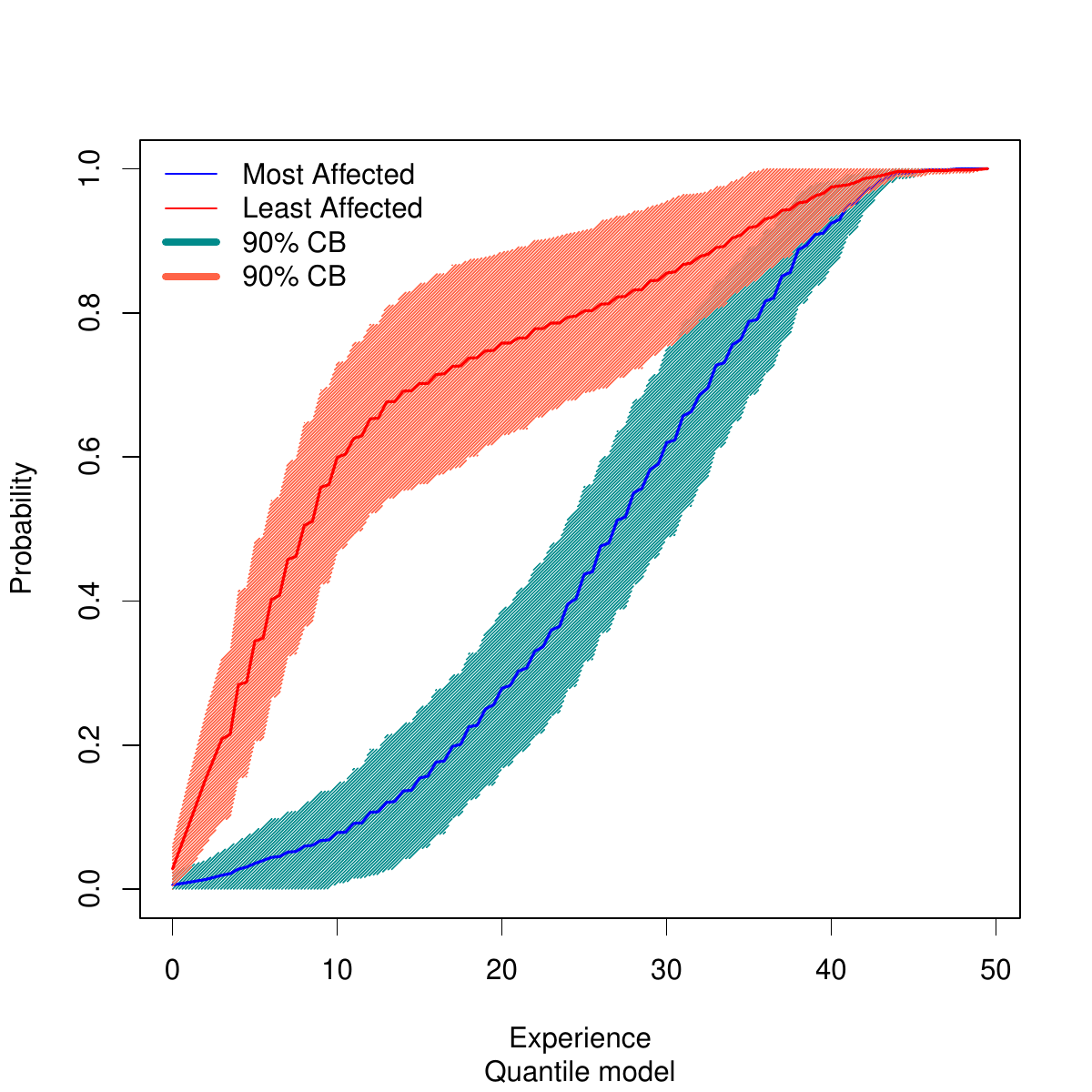}
\includegraphics[width=.45\textwidth,height=.5\textwidth]{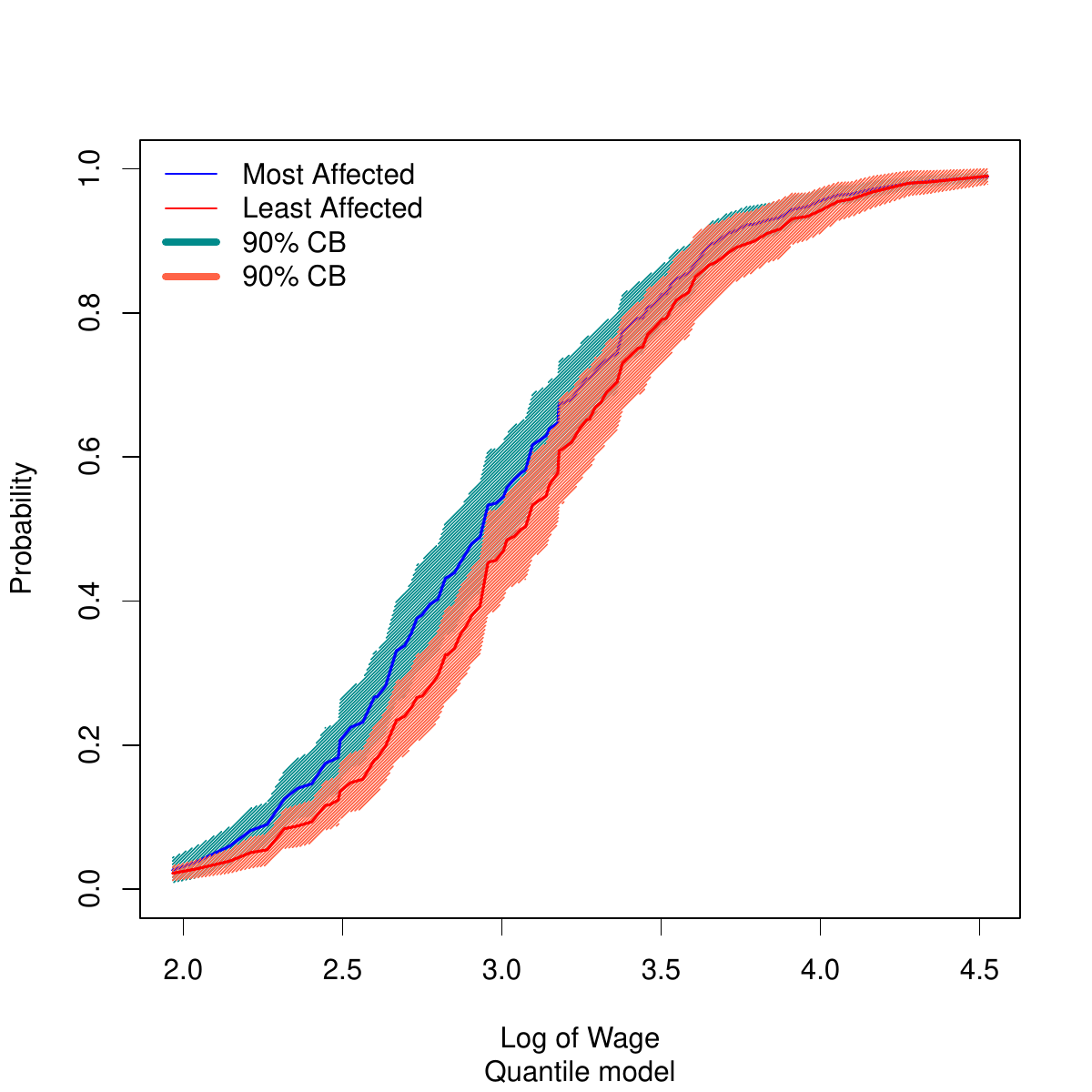}
\caption{Estimates and 90\% weighted bootstrap joint uniform confidence bands for the distributions of experience and log wages of the 10\% most and least affected women by gender wage gap.}\label{fig:ca2}
\end{figure}

Figure \ref{fig:ca2} plots simultaneous 90\% confidence bands for the distribution of experience and log wage for the most and least affected women. They are obtained by Algorithm \ref{alg:ca} with 500 weighted bootstrap replications. The  estimated distribution of experience for the most affected first-order stochastically dominates the same estimated distribution for the least affected women. Moreover, the uniform bands confirm that this dominance is statistically significant at the 90\% confidence level for the underlying distributions. The estimated (marginal) distribution of log wage for the least affected first-order dominates the same estimated distribution for most affected, but we cannot reject that the underlying distributions are equal at the 10\% significance level. The results of the classification analysis are  consistent with  preferences that make never married highly educated young women working on managerial occupations be more career-oriented.\footnote{We find similar results using the additive error model. We do not report these results for the sake of brevity.}

\begin{figure}[!h]
\centering
\includegraphics[width=.45\textwidth,height=.5\textwidth,page=1]{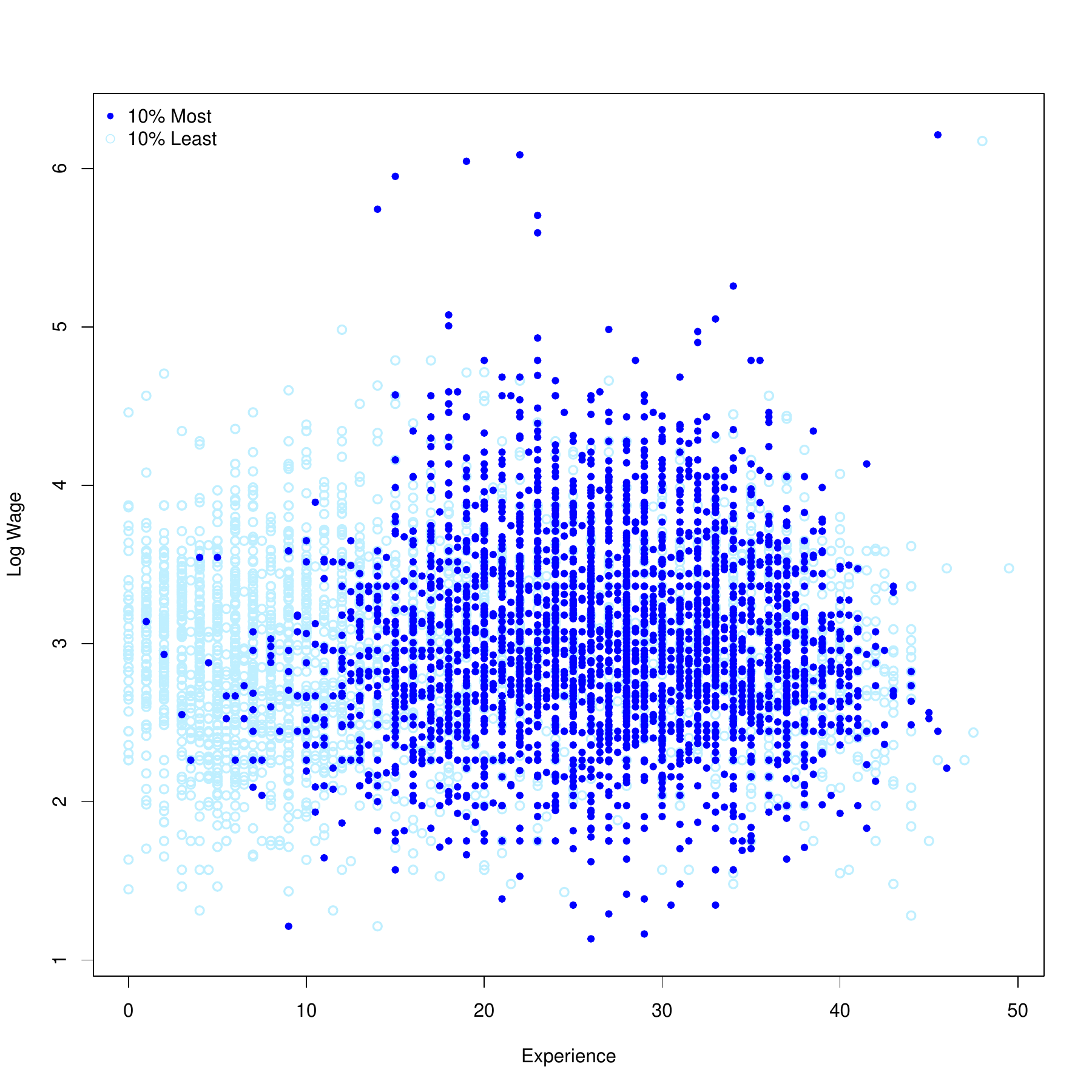}
\includegraphics[width=.45\textwidth,height=.5\textwidth,page=2]{Gender-CSPET-lm-full-ca-by_exp-nointer.pdf}
\caption{Estimates and 90\% weighted bootstrap confidence bands for projections of the confidence sets for characteristics of 10\% most and least affected women by gender wage gap.}\label{fig:ca3}
\end{figure}

Finally, Figure \ref{fig:ca3} plots two dimensional projections of experience-log wage and experience-marital status of the confidence sets for the 10\% most and least affected subpopulations. We show the results from  the additive error model for the conditional expectation. Here we use a simplified specification that excludes the two-way interactions from $W$ to get more precise estimates of all the PEs. We  obtain 90\% confidence sets for the most and least affected subpopulations by weighted bootstrap with standard exponential weights and 500 repetitions. The sets $\mathcal{CM}^{-0.1}(0.90)$ and $\mathcal{CM}^{+0.1}(0.90)$ include $23\%$ and $19\%$ of the women in the sample, respectively.\footnote{Recall that in this application the set $\mathcal{CM}^{-0.1}(0.90)$ corresponds to most affected women and  $\mathcal{CM}^{+0.1}(0.90)$ to least affected women.  We drop one woman that is included in both sets.} The projections show that there are relatively more least affected women with low experience at all wage levels, more high affected women with high wages with between 15 and 25 years of  experience, and more least affected women which are not married at all experience  levels.

\section{Detailed Large Sample Theory}\label{sec:theory2}

\subsection{Detailed Large Sample Theory for SPE}
For an open  set $\KK$, let the class $\C^1$ on $\KK$ denote the set of continuously differentiable real valued functions on  $\KK$. We make the following technical assumptions about the PE function $\Delta: \Bbb{R}^{d_x} \mapsto \Bbb{R}$ and the distribution of the covariates:

${\sf S}.1$. The part of the domain of the PE function $x \mapsto \Delta(x)$ of interest, $\mathcal{X}$, is  open and its closure $\overline{\mathcal{X}}$ is compact. The distribution $\mu$ is absolutely continuous with respect to the Lebesgue measure with density $\mu'$. There exists an open set $B(\mathcal{X})$ containing $\overline{\mathcal{X}}$ such that $x \mapsto \Delta(x)$ is $\C^1$ on $B(\mathcal{X})$, and  $x \mapsto \mu'(x)$ is continuous on  $B(\mathcal{X})$ and is zero outside the domain of interest, i.e.  $\mu'(x) = 0$ for any $x \in B(\XX) \setminus \XX$.

${\sf S}.2$. Let $\MM_{\Delta}(\delta):= \{x \in   {\XX}:  \Delta(x)=\delta\}$. For any regular value $\delta$ of $\Delta$ on $\overline{\XX}$,
we assume that the closure of $\MM_{\Delta}(\delta)$ has a finite number of connected branches.

The following property of the set $\MM_{\Delta}(\delta)$ is a useful implication of Assumptions ${\sf S}.1$ and ${\sf S}.2$ that we will exploit in the analysis.

\begin{remark}[Properties of $\MM_{\Delta}(\delta)$] By Theorem 5-1 in \citen[p. 111]{spivak-65},  ${\sf S}.1$ and ${\sf S}.2$ imply that  $\MM_{\Delta}(\delta)$ is a $(d_x-1)$-manifold without boundary in $\mathbb{R}^{d_x}$ of class $\C^1$ for any $\delta$ that is a regular value  of $x \mapsto \Delta(x)$ on  $\overline \XX$.   
\end{remark}

Assumption  ${\sf S}.1$  imposes mild smoothness conditions on the PE function $x \mapsto \Delta(x)$. It also  requires that all the components of the covariate  $X$ are continuous random variables. We defer the treatment of  the case where $X$ has both continuous and discrete  components to the SM.  As a matter of generalization, our theoretical analysis allows us to replace that $x \mapsto \mu'(x)$ vanishes on $\partial \mathcal{X}$, by the weaker condition that the intersection of $\mathcal{M}_\Delta(\delta)$ and the boundary of $\XX$  have zero volume with respect to $\mu$, namely
\begin{equation}\label{eq: Boundary}
\int_{\MM_{\Delta}(\delta)} 1\{x\in \partial \XX\} \frac{\mu'(x)}{\|\partial \Delta (x)\|} d\mathrm{Vol}=0,
\end{equation}
where $\partial \XX$ denotes the boundary of $\overline \XX$, $\partial \Delta(x)$ is the gradient of $x \mapsto \Delta(x)$, and $\int_{\MM} f(x) d\mathrm{Vol}$ denotes the integral of the function $f$ on the manifold $\MM$ with respect to volume; see Appendix \ref{sec:df} in the SM for a brief review on Differential Geometry.  This relaxation is relevant to cover the case where $X$ includes an uniformly distributed  component such as the unobserved rank in Example \ref{example:cqf}.\footnote{ In the numerical examples of Section \ref{subset:numerical} in the SM, the first  two designs only satisfy this relaxed condition.}  

Assumption ${\sf S}.2$ imposes shape restrictions on $x \mapsto \Delta(x)$ that rule out cases such as infinite cyclical oscillations or flat areas. A simple sufficient condition for  ${\sf S}.2$ is that the map $x \mapsto \Delta(x)$ does not have critical points on $\overline \XX$. This means that  $x \mapsto \Delta (x)$ is not locally flat anywhere on $\overline \XX$, which we define to mean that  the norm of the gradient, $\| \partial \Delta(x) \|$, does not vanish on $ x\in \overline \XX$. In this case, any $\delta$ 
in the image of $\overline \XX$ under $\Delta$ is regular.  This condition is probably the most relevant for practice  and can be verified in applications, at least informally.

\begin{remark}[Verification of Regularity Conditions in Practice]  The main regularity condition is that PE function $x\mapsto \Delta(x)$ be smooth and  not locally flat, namely $\| \partial \Delta(x)\|$ does not vanish. Our inferential results are developed under this assumption, and they do not apply otherwise. To verify if these results apply in practice, we strongly recommend to conduct a  Monte Carlo experiment using a data generating process that mimics the application at hand.\footnote{In fact, we recommend doing this for every econometric method.} Indeed, failure of the inference method in the simulation experiment implies failure of the regularity conditions. We provide an application of this supporting analysis to the gender wage gap example in Appendix \ref{subset:numerical}. 
Looking forward, it would be useful to develop further an inference method with good robustness properties with respect to the regularity conditions, i.e. that remains uniformly valid when the PE function is (close to being) locally flat.  We delegate this line of research to future work.\footnote{For instance, it is of interest to determine whether the use of subsampling instead of bootstrap can deliver a more robust inference method when the PE function is close to being flat; see, e.g. \citen{Romano:Shaikh:AoS}. }
\end{remark}

We make the following assumptions about the estimator of the PE.  Let $\ell^{\infty}(\mathcal{T})$ denote the set of bounded and measurable functions $g:\mathcal{T} \to \mathbb{R}$ and $\mathcal{F}$   a fixed subset of continuous  functions on $B(\mathcal{X})$. Let $\ell^{\infty} (B(\XX))$ be the set of bounded and measurable functions on $B(\XX)$ and $\rightsquigarrow$ denote weak convergence (convergence in distribution).

${\sf S}.3$. $\widehat{\Delta}$, the estimator  of $\Delta$,  belongs to $\mathcal{F}$ with probability approaching $1$ and obeys a
functional central limit theorem, namely,
$$a_n(\widehat{\Delta}-\Delta)\rightsquigarrow G_{\infty} \text{ in } \ell^{\infty}
(B(\XX)),$$ where $a_n$ is a sequence such that $a_n \to \infty$ as $n \to \infty$, and $x \mapsto G_{\infty}(x)$ is a tight process that has almost surely uniformly
continuous sample paths on $B(\XX)$.

 In the parametric and semiparametric models of Examples 1--3, ${\sf S}.3$ holds under weak conditions that guarantee  asymptotic normality of  the ML,  OLS and QR estimators. For the QR estimator in Example \ref{example:cqf} where the unobserved rank is one of the covariates, these conditions include that the density of $Y$ conditional on $X$ be bounded away from zero \cite{koenker:book}, which is facilitated by excluding tail quantile indexes.



%
%

Let $\widehat \mu$ be the estimator of the distribution $\mu$. It is convenient to identify  $\mu$  and $\widehat \mu$ with the operators: $$g \mapsto \mu (g) =\int g(x)  d\mu(x), \quad g \mapsto \widehat \mu (g) =\int g(x)  d\widehat \mu(x), $$  mapping
from the set $\mathcal{G} := \{ x \mapsto 1( f(x) \leq \delta):  f \in \mathcal{F}, \delta \in \mathcal{V}\}$   to $\Bbb{R}$, where $\mathcal{F}$  is the fixed subset of continuous  functions on $B(\mathcal{X})$ containing $\Delta$, and $\VV$ is any compact set of $\Bbb{R}$.  We  require $\mathcal{G}$ to be totally bounded under the $L^2(\mu)$ norm.
Define $\mathbb{H}$ as the set of all bounded linear operators $H$ on $\mathcal{G}$
of the form
$$
g  \mapsto H(g),
$$
which are uniformly continuous on $g \in \mathcal{G}$ under the $L^2(\mu)$ norm.
We define the boundedness of these operators
with respect to the norm:
$$\|H\|_{\mathcal{G}}=\sup_{g\in \mathcal{G}}{|H(g)|},$$
and define the corresponding distance between  two operators $H$ and $\widetilde H$ in
$\mathbb{H}$ as $\|H - \widetilde H \|_{\mathcal{G}} = \sup_{g\in \mathcal{G}}
|H(g)-\widetilde H(g)|$. Clearly,  $\mu \in \mathbb{H}$.

We make the following assumption about $\widehat \mu$.

${\sf S}.4$.
The function $x \mapsto \widehat \mu(x)$ is a distribution over $B(\XX)$ obeying in $\mathbb{H}$,
\begin{equation}
b_n(\widehat \mu - \mu) \rightsquigarrow H_{\infty},
\end{equation}
where $g\mapsto H_{\infty}(g)$ is a.s. an element of $\mathbb{H}$  (i.e. it has almost surely uniformly continuous sample paths on $\mathcal{G}$ with respect to the $L^2(\mu)$ metric)
and  $b_n$ is a sequence such that $b_n \to \infty$ as $n \to \infty$.

When $\widehat \mu$
is the  empirical distribution based on a random sample from the  population
with distribution $\mu$, then $b_n = \sqrt{n}$ and  $H_{\infty} = B_{\mu}$,  where $B_{\mu}$ is a $\mu$-Brownian Bridge, i.e. a Gaussian process with zero mean and
 covariance function $(g_1,g_2) \mapsto \mu(g_1 g_2) -  \mu(g_1) \mu(g_2)$. 
In this case condition ${\sf S}.4$ imposes that the function class  $$\mathcal{G} =  \{x \mapsto 1(f(x) \leq \delta) : f \in \mathcal{F}, \delta \in \VV \}$$ is  $\mu$-Donsker.   Note that $\mathcal{F}$ is the parameter
space that contains $\Delta(x)$ as well as $\widehat \Delta(x)$ in ${\sf S}.3$. In parametric models for the PE where $ \mathcal{F} = \{ f(x,\theta) : \theta \in \Theta \}$, $f$ is known, $\theta \subseteq \RR^{d_\theta}$ with $d_\theta < \infty$, and $x \mapsto f(x,\theta)$ is $\mathcal{C}^1$ on $\XX$ for all $\theta \in \Theta$, the class $\mathcal{G}$ is
$\mu$-Donsker under mild conditions specified for example in \citen[Chap. 19]{vdV}. Examples 1 and 2 specify the PE parametrically. Lemma  \ref{cor:donsker} in the SM gives other sufficient conditions for the Donsker property.  %


The following result is derived as a consequence of the new mathematical results on the Hadamard
differentiability of the sorting operator, stated in Lemma \ref{lemma:had3} in the Appendix (proof given in SM due to space constraints), in conjunction with the functional delta method.  It shows that the empirical SPE-function follows a FCLT over sets of quantiles corresponding to $\Delta_{\mu}^*$ pre-images of compact sets of $\mathbb{R}$.

Define $\DD$ as a compact set consisting of regular values of $x \mapsto \Delta(x)$ on $\overline \XX$,  and $\mathcal{U} :=  \{\widetilde u \in [0,1] :  \Delta^*_{\mu}(\widetilde u) \in \DD, f_{\Delta,\mu}(\Delta_{\mu}^*(\widetilde u)) > \varepsilon\},$ for a fixed $\varepsilon>0$, where $ f_{\Delta,\mu}(\Delta_{\mu}^*(\widetilde u))$ is the density of $\Delta(X)$ defined in Lemma \ref{lemma:had0}(a).  Let $r_n := a_n\wedge
b_n$, the slowest of the rates of convergence of $\widehat \Delta$ and $\widehat \mu$. Assume  $r_n/a_n \to s_{\Delta} \in [0,1]$ and $r_n/b_n \to s_{\mu} \in [0,1]$, where $s_{\Delta} = 0$ when $b_n = o(a_n)$ and $s_{\mu} = 0$ when $a_n = o(b_n)$. For example, $s_{\mu} = 0$ if $\mu$ is treated as known.

\begin{theorem}[FCLT for  $F_{\widehat \Delta,\widehat \mu}$ and $\widehat{\Delta}_{\widehat \mu}^*$]\label{thm:fclt}
Suppose that ${\sf S}.1$-${\sf S}.4$ hold, and the convergence in ${\sf S}.3$ and ${\sf S}.4$ holds jointly. Then, as $n \to \infty$,

(a) The estimator of the distribution of PE obeys a functional central limit theorem,
namely,  in $\ell^{\infty}(\DD)$,
\begin{equation*}
r_n(F_{\widehat \Delta, \widehat \mu}(\delta)-F_{\Delta,\mu}(\delta))\rightsquigarrow 
 s_{\Delta} T_{\infty}(\delta) + s_{\mu} H_{\infty}(g_{\Delta,\delta}),
\end{equation*}
as a stochastic process indexed by $\delta \in \DD$, where
$$
T_{\infty}(\delta):= -\int_{\MM_{\Delta}(\delta)}\frac{G_{\infty}(x)\mu'(x)}{\|\partial \Delta(x)\|}d\mathrm{Vol}.
$$

(b) The empirical SPE-process obeys a functional central limit
theorem, namely in $\ell^\infty (\UU)$, \begin{eqnarray}\label{eq:fclt}
r_n(\widehat\Delta^*_{\widehat \mu}(u)-\Delta^*_{\mu}(u))\rightsquigarrow  
-\frac{s_{\Delta} T_{\infty}(\Delta^*_{\mu}(u)) + s_{\mu} H_{\infty}(g_{\Delta,\Delta^*_{\mu}(u)})}{ \int \frac{\mu'(x)}{\|\partial \Delta(x)\|}d\mathrm{Vol}} 
&=:& Z_{\infty}(u),
\end{eqnarray}
as a stochastic process indexed by $u \in \UU$.
\end{theorem}

\begin{remark}[Critical values] (a) Theorem \ref{thm:fclt} shows that  $\delta \mapsto F_{\widehat \Delta,\widehat \mu}(\delta)$  ($u \mapsto \widehat \Delta^*_{\widehat \mu}(u)$) follows a FCLT over  any compact set  $\DD$ (the $\Delta_{\mu}^*$ pre-image of $\DD$), where $\DD$ excludes the critical values of $x \mapsto \Delta(x)$ on $\overline \XX$.
Thus, we can set $\DD = \Delta(\overline \XX) := \{\Delta(x) : x \in  \overline \XX\}$ when the map $x \mapsto \Delta(x)$ does not have critical points on $\overline \XX$.
This case is nice because it allows us not to worry about critical values when performing inference, and practically relevant as it occurs very naturally in many applications. For instance, it arises whenever $\Delta(x)$ is strictly locally monotonic in some direction.
(b) In numerical examples reported in the SM, we find that the bootstrap inference method proposed performs  well even in models where $x \mapsto \Delta(x)$ has critical points, without excluding the corresponding critical values from $\DD$.  This evidence suggests that the exclusion of critical values might not be necessary for inference.
\qed
\end{remark}

\subsection{Detailed Large Sample  Theory for CA} It is convenient to modify the notation for the $u$-CA separating the dependence on $\Delta^*_{\mu}(u)$ from $\Delta$ and $\mu$ and specifying the characteristic of interest as $\varphi_t$. Moreover, when $Z = (X,Y)$ we remove the dependence on  $Y$ by taking expectations conditional on $X$. Let $\Lambda_{\widehat \Delta,\widehat \mu,\widehat\Delta^*_{\widehat \mu}(u)}(\varphi_t):= \widehat{\Lambda}_{\Delta,\mu}^{u}(t)$ and  $\Lambda_{\Delta,\mu,\Delta^*_{\mu}(u)}(\varphi_t) := \Lambda_{\Delta,\mu}^{u}(t)$, where $\varphi_t \in \mathcal{F}_M \cup \mathcal{F}_I$, $t = (t_1, \ldots, t_{d_z}) \in \mathbb{R}^{d_z}$, $u \in \UU$, $\mathcal{F}_M:=\{\int z_{1}^{t_1} \cdots  z_{d_z}^{t_{d_z}} d\mu(y \mid x): t_1,...,t_{d_z} \in \{0,1,2, \ldots\}, \int |z_{1}^{t_1} \cdots  z_{d_z}^{t_{d_z}}| d\mu(z) < \infty , t_1 + \ldots + t_{d_z} \leq  M\}$,  $M$ is some fixed integer, $\mu(y \mid x)$ is the distribution of $Y$ at $y$ conditional on $X=x$,  and $\mathcal{F}_I:=\{\int 1(z_{1}\leq t_1,...,z_{d_z}\leq t_{d_z}) d\mu(y \mid x):  t_1,...,t_{d_z} \in  \mathbb{R}\}$. For example, $\varphi_t(x) = x^{t_x} \Ep[Y^{t_y} \mid X = x]$ or $\varphi_t(x) = 1(x \leq t_x) \mu(t_y \mid x)$ for $t = (t_x,t_y)$. 
 To derive the properties of $\Lambda_{\widehat \Delta,\widehat \mu,\widehat\Delta^*_{\widehat \mu}(u)}(\varphi_t)$, we use that the class of functions $\widetilde{\mathcal{G}}=\{1(f\leq \delta)\varphi : \varphi\in \mathcal{F}_M \cup  \mathcal{F}_I,\delta\in \mathcal{V}, f\in \mathcal{F}\}$ is $\mu$-Donsker. When $x \mapsto \mu(y \mid x)$ is continuous, this property holds  by assumption ${\sf S}.4$ when $\widehat \mu$ is the empirical distribution.\footnote{Lemma \ref{lemma:Donsker_supp} in the SM gives other sufficient conditions for the Donsker property.}

The following result is derived as a consequence of the new mathematical results on the Hadamard
differentiability of the classification operator, stated in Lemma \ref{lemma:Ratio-G} in the Appendix (proof given in SM due to space constraints), in conjunction with the functional delta method. 

\begin{theorem}[FCLT for $\Lambda_{\widehat \Delta,\widehat \mu,\widehat\Delta^*_{\widehat \mu}(u)}(\varphi_t)$]\label{thm:fclt_supp1}
Suppose that ${\sf S}.1$-${\sf S}.4$ hold,  the convergence in ${\sf S}.3$ and ${\sf S}.4$ holds jointly, and $u \in \UU$. If $Z = (X,Y)$,  then assume that $\mathcal{Y}$ is compact and $x \mapsto \mu(y \mid x)$ is continuous on $B(\XX)$ for all $y \in \mathcal{Y}$. Then, as $n \to \infty$,  (a) $\Lambda_{\widehat \Delta,\widehat \mu,\widehat\Delta^*_{\widehat \mu}(u)}(\varphi_t)$ obeys a  FCLT with respect to $t \mapsto \varphi_t \in \mathcal{F}_M$,
namely,  in $\ell^{\infty}(\mathbb{R}^{d_z})^2$,
 \begin{multline*}
 r_n\left(\Lambda_{\widehat \Delta,\widehat \mu,\widehat\Delta^*_{\widehat \mu}(u)}(\varphi_t) - \Lambda_{\Delta, \mu,\Delta^*_{ \mu}(u)}(\varphi_t) \right)  \\ \rightsquigarrow \int_{\mathcal{M}_\Delta(\delta)}\widetilde \varphi_t(x)\frac{Z_{\infty}(u)-s_\Delta G_\infty(x)}{\|\partial \Delta(x)\|}\mu'(x)d\mathrm{Vol}+s_\mu H_\infty(h_{\Delta,\delta,\varphi_t}) =:  Z^u_{\infty}(t),
\end{multline*}
as a stochastic process indexed by $t \in \mathbb{R}^{d_z}$, where $\widetilde \varphi_t(x) = [\varphi_t(x) -  \Lambda_{\Delta,\mu,\delta}(\varphi_t)]/F_{\Delta,\mu}(\delta)$, $\widetilde h_{\Delta,\delta,\varphi_t}:=\widetilde \varphi_t(x)1\{\Delta(x)\leq \delta\}$, and $Z_{\infty}(u)$ is the limit process of Theorem \ref{thm:fclt}; and
 (b) if in addition Assumption ${\sf AS}.1$ holds, then $\Lambda_{\widehat \Delta,\widehat \mu,\widehat\Delta^*_{\widehat \mu}(u)}(\varphi_t)$ obeys the same  FCLT with respect to $t \mapsto \varphi_t \in \mathcal{F}_I$.
%
%
%
%
%
\end{theorem}


%


Assumption ${\sf AS}.1$ is a technical condition stated in Appendix \ref{app:B}
 of the SM to deal with the discontinuity of the indicator functions when $\varphi_t \in \mathcal{F}_I$. A sufficient condition for ${\sf AS}.1$ is that  $$\int_{\MM_{\Delta}(\delta) \cap \{x : x_k=t_k\}}d \Vol=0$$ holds uniformly over all $\delta\in \mathcal{V}$, $t_k \in \mathbb{R}$ and $k=1,2,...,d_x$. In other words, the manifold $\MM_\Delta(\delta)$ and the set of points $\{x : x_k = t_k\}$ can not have an intersection with positive volume of $(d_x-1)$-dimension. %

\subsection{Bootstrap Inference for SPE and CA}\label{subsec:boot}
Corollaries  \ref{cor:cb} and \ref{cor:cb2} use critical values of statistics related to the limit processes $Z_{\infty}$ and $Z^u_{\infty}$ to construct confidence bands and  p-values. These critical values can be hard to obtain in practice.  In principle one can use simulation, but it might be difficult to numerically locate and parametrize the manifold $\MM_{\Delta}(\delta)$, and to evaluate the integrals on  $\MM_{\Delta}(\delta)$ needed to compute the realizations of $Z_{\infty}(u)$ and $Z^u_{\infty}(t)$. This creates a real challenge to implement our inference methods. To deal with this challenge we employ  (exchangeable) bootstrap to compute critical values \cite{praestgaard-wellner-93,vdV-W} instead of simulation. We show that the bootstrap law is consistent to approximate the distribution of the limit processes of Theorems \ref{thm:fclt} and \ref{thm:fclt_supp1}.

To state the bootstrap validity result formally, we follow the notation and definitions in \citen{vdV-W}. Let $\mathrm{D}_n$ denote the data vector and let $\mathrm{B}_n = (\omega_1, \dots, \omega_n)$ be the vector of bootstrap weights. Consider a random element $\widetilde{Z}_n = Z_n(\mathrm{D}_n,\mathrm{B}_n)$ in a normed space $\mathbb{D}$. We say that the bootstrap law of $\widetilde{Z}_n$ consistently estimates the law of some tight random element $Z_{\infty}$ and write $\widetilde{Z}_n\rightsquigarrow_{\Pr} Z_{\infty}$ if
$$
\sup_{h \in \mathrm{BL}_1(\mathbb{D})} |\Ep_{\mathrm{B}_n} h(\widetilde{Z}_n) - \Ep_{\Pr} h(Z_{\infty})| \to_{\Pr} 0,
$$
where $ \mathrm{BL}_1(\mathbb{D})$ denotes the space of functions with Lipschitz norm at most 1; $\Ep_{\mathrm{B}_n}$ denotes the conditional expectation with respect to $\mathrm{B}_n$ given the data $\mathrm{D}_n$; $\Ep_{\Pr}$ denotes the expectation with respect to $\Pr$, the distribution of the data $\mathrm{D}_n$; and $\to_{\Pr}$ denotes convergence in (outer) probability.


The next result is a consequence of the functional delta method for the exchangeable bootstrap. Let $\Lambda_{\widetilde \Delta,\widetilde \mu,\widetilde \Delta^*_{\widetilde \mu}(u)}(\varphi_t):= \widetilde{\Lambda}_{\Delta,\mu}^{u}(t)$, the bootstrap draw of $\widehat{\Lambda}_{\Delta,\mu}^{u}(t)$ defined in Algorithm \ref{alg:ca}.
 \begin{theorem}[Bootstrap FCLT for $\widehat{ \Delta_{ \mu}^*}$ and $\Lambda_{\widehat \Delta,\widehat \mu,\widehat\Delta^*_{\widehat \mu}(u)}(\varphi_t)$] \label{thm:bfclt} Suppose
 that the bootstrap is consistent for the law of the estimator of the PE, namely $a_n(\widetilde{\Delta} - \widehat \Delta) \rightsquigarrow_{\Pr} G_{\infty}$ in $\ell^{\infty}(B(\XX))$, and
 for the law of the estimated measure, namely $b_n (\widetilde \mu - \widehat \mu) \rightsquigarrow_{\Pr} H_{\infty}$ in $ \mathbb{H}$.
Then, (1) under the assumptions of Theorem \ref{thm:fclt}, the bootstrap is consistent for the law of the empirical SPE-process, namely
$$
r_n (\widetilde{\Delta_{\mu}^*}(u)  - \widehat{\Delta_{\mu}^*}(u) ) \rightsquigarrow_{\Pr} Z_{\infty}(u) \text{ in $\ell^{\infty}(\UU)$};
$$
and (2) under the assumptions of Theorem \ref{thm:fclt_supp1}, the bootstrap is consistent for the law of the empirical CA-process, namely
$$
 r_n\left(\Lambda_{\widetilde \Delta,\widetilde \mu,\widetilde \Delta^*_{\widetilde \mu}(u)}(\varphi_t) - \Lambda_{\widehat \Delta, \widehat \mu,\widehat \Delta^*_{ \mu}(u)}(\varphi_t) \right) \rightsquigarrow  Z^u_{\infty}(t) \text{ in $\ell^{\infty}(\mathbb{R}^{d_z})^2$}.
$$
\end{theorem}

Theorem \ref{thm:bfclt} employs the high-level condition that the bootstrap can approximate consistently the laws of $\widehat \Delta$  and $\widehat \mu$, after suitable rescaling.   In  Examples 1-3  when  $\widehat \mu$ is the empirical measure based on the random sample of size $n$,  the exchangeable bootstrap method entails randomly reweighing the sample using the weights $(\omega_{1}, \ldots, \omega_{n})$, which include empirical boostrap and i.i.d. exponential weights, for example.  In this case the high level condition holds if the weights satisfy the conditions stated in equation (3.6.8)  of \citen{vdV-W}. We refer to \citen{vdV-W} and \citen{CFM} for bootstrap FCLT for parametric and semi parametric estimators of $\Delta$ including least squares, quantile regression, and distribution regression, as well as nonparametric estimators of $\mu$ including the empirical distribution function.

\appendix

\section{Key New Mathematical Results: Hadamard Differentiability of Sorting and Classification Operators}\label{sec:theory1}

\subsection{Notation}
We denote the PE as $\Delta(x)$, the empirical PE as $\widehat{\Delta}(x)$, and
 $\partial \Delta(x) := \partial \Delta(x)/\partial x$, the gradient of $x \mapsto \Delta(x)$.  For a vector $v = (v_1, \ldots, v_{d_v}) \in \mathbb{R}^{d_v}$, $\|v\|$ denotes the Euclidian norm of $v$, that is $\|v\| = \sqrt{v\transp v}$, where the superscript  $\transp$ denotes transpose.

\subsection{Basic Analytical Properties of Sorted Functions}

The following lemma establishes the properties of the distribution function  $\delta \mapsto F_{\Delta,\mu}(\delta)$ and the SPE-function  $u \mapsto \Delta_{\mu}^*(u)$.

Define $\DD$ as a compact set consisting of regular values of $x \mapsto \Delta(x)$ on $\overline \XX$.

\begin{lemma}[Basic Properties of $F_{\Delta,\mu}$ and
$\Delta_{\mu}^*$]\label{lemma:had0}
 Under conditions ${\sf S}.1$ and ${\sf S}.2$:

1. For any  $\delta \in \DD$,
the derivative of $F_{\Delta,\mu}(\delta)$ with respect to $\delta$ is:
\begin{equation}
f_{\Delta,\mu}(\delta):= \partial_{\delta} F_{\Delta,\mu}(\delta)
=\int_{\MM_{\Delta}(\delta)}\frac{\mu'(x)}{\|\partial \Delta (x)\|}d\mathrm{Vol}.
\end{equation}
This integral is well-defined because the gradient $x \mapsto \partial \Delta(x)$ is finite,
continuous, and bounded away from $0$ on $\MM_{\Delta}(\delta)\subseteq
\overline \XX$.  The map $ \delta \mapsto f_{\Delta,\mu}( \delta)$ is uniformly continuous on $\mathcal{D}$.

2. Fix $\varepsilon>0$, then for any $u \in \mathcal{U} :=  \{\widetilde u \in [0,1] :  \Delta^*_{\mu}(\widetilde u) \in \DD, f_{\Delta,\mu}(\Delta_{\mu}^*(\widetilde u)) > \varepsilon\},$
the derivative of  $\Delta^*_\mu(u)$  respect to $u$ is:
\begin{equation}
\partial_u \Delta^*_\mu(u)= \frac{1}{f_{\Delta,\mu}(\Delta_{\mu}^*(u))}.
\end{equation}
\label{equation}
Moreover,  the derivative map $ u \mapsto \partial_u \Delta^*_\mu ( u)$ is
uniformly continuous on $\mathcal{U}$.\end{lemma}


\subsection{Functional Derivatives of Sorting-Related Operators}
We consider the properties of the distribution function and the SPE-function as functional operators $(\Delta,\mu) \mapsto F_{\Delta,\mu}$ and $(\Delta,\mu) \mapsto \Delta^*_{\mu}$. We show that these operators are Hadamard differentiable with respect to $(\Delta,\mu)$. These results are critical ingredients to deriving the large sample distributions of the empirical versions of $F_{\Delta,\mu}$ and $\Delta^*_{\mu}$  in Section \ref{sec:theory2}.

We now recall the definition of uniform Hadamard differentiability from \citen{vdV-W}.


\begin{definition}[Hadamard Derivative Uniformly in an Index]
Suppose the linear spaces $\mathbb{D}$ and $\mathbb{E}$ are equipped with
the norms $\|\cdot \|_{\mathbb{D}}$ and $\|\cdot \|_{\mathbb{E}}$, and $\Theta$
is a compact subset of a metric space. A map
$\phi_\theta: \mathbb{D}_\phi\subseteq \mathbb{D}\rightarrow \mathbb{E}$
is called Hadamard-differentiable uniformly in $\theta \in \Theta$ at $f\in \mathbb{D}_{\phi}$ tangentially to a subspace $\mathbb{D}_0 \subseteq \mathbb{D}$ if
there is a continuous linear map $\partial_f\phi_{\theta}:
\mathbb{D}_0\rightarrow \mathbb{E}$ such that uniformly in $\theta \in \Theta$:
\begin{equation}
 \frac{\phi_{\theta}(f+t_nh_n)-\phi_{\theta}(f)}{t_n} - \partial_f\phi_{\theta}[h]  \to 0, \ \
n\rightarrow \infty,
\end{equation}
for all converging real sequences $t_n\rightarrow 0$ and
$\|h_n-h\|_{\mathbb{D}}\rightarrow 0$ such that $f+t_nh_n\in \mathbb{D}_\phi$ for
every $n$, and $h \in \mathbb{D}_0$; moreover,  the map
$(\theta, h) \mapsto \partial_f\phi_{\theta}[h]$ is continuous on $\Theta \times \mathbb{D}_0$.\end{definition}




In what follows, we let $\mathbb{F}$ denote the space of continuous functions on
$B(\mathcal{X})$ equipped with the sup-norm, and  $\mathbb{F}_0$ denote
a subset of $\mathbb{F}$ that contains  uniformly continuous functions.




\begin{lemma}[Hadamard differentiability of $(\Delta,\mu) \mapsto F_{\Delta,\mu}$ and  $(\Delta,\mu) \mapsto \Delta^*_{\mu}$ ]\label{lemma:had3}  Let  $\mathbb{D}:=\mathbb{F} \times \mathbb{H}$ and
$\mathbb{D}_0:= \mathbb{F}_0\times \mathbb{H}$.
 Assume that ${\sf S}.1$-${\sf S}.2$ hold. Then,

(a) The map $(\Delta, \mu) \mapsto F_{\Delta,\mu}(\delta)$, mapping $\mathbb{D} \to  \Bbb{R}$, is Hadamard differentiable uniformly in $\delta \in \mathcal{D}$ at
$(\Delta,\mu)$ tangentially to $\mathbb{D}_0$  with the derivative map $\partial_{\Delta,\mu} F_{\Delta,\mu}(\delta):  \mathbb{D}_0  \to \Bbb{R} $ defined by
$$(G,H)\mapsto \partial_{\Delta,\mu} F_{\Delta,\mu}(\delta)[G,H]:=-\int_{\MM_{\Delta}(\delta)}\frac{G(x)\mu'(x)}{\|\partial
\Delta(x) \|}d\mathrm{Vol} + H(g_{\Delta, \delta}).$$

(b) The map $(\Delta, \mu) \mapsto \Delta_{\mu}^*(u)$, mapping $\mathbb{D} \to  \Bbb{R}$ is Hadamard differentiable uniformly in $u \in \UU$ at
$(\Delta,\mu)$ tangentially to $\mathbb{D}_0$  with
the derivative map, $\partial_{\Delta,\mu} \Delta^*_{\mu}(u): \mathbb{D}_0 \to \Bbb{R}$, defined by
$$(G,H)\mapsto \partial_{\Delta,\mu} \Delta^*_{\mu}(u)[G,H]:=
-\frac{\partial_{\Delta,\mu} F_{\Delta,\mu}(\Delta^*_\mu(u))[G,H]}{
f_{\Delta,\mu}(\Delta^*_\mu(u))}. $$
\end{lemma}

\subsection{Functional Derivatives of Classification Operators}
Let $\widetilde{\mathbb{D}} := \mathbb{F} \times \widetilde{\mathbb{H}}\times \mathbb{R}$ and $\widetilde{\mathbb{D}}_0:= \mathbb{F}_0\times \widetilde{\mathbb{H}} \times \mathbb{R}$, where $\mathbb{F}$ and $\mathbb{F}_0$ are defined as before; $\widetilde{\mathbb{H}}$ is  the set of bounded linear operators mapping from the set $\widetilde{\mathcal{G}}:=\{\varphi 1(\Delta \leq \delta) : f \in \mathcal{F}, \varphi \in \mathcal{F}_I \cup \mathcal{F}_M, \delta\in \mathcal{V}\}$ to $ \mathbb{R}$, with norm
$$\|H\|_{\widetilde{\mathcal{G}}}=\sup_{g\in \widetilde{\mathcal{G}}}|H(g)|,$$
 where the map $g \mapsto H(g)$ is uniformly continuous on $g \in \widetilde{\mathcal{G}}$ under the $L^2(\mu)$ norm. We derive the properties of the least affected classification operator $\Lambda^{-}_{\Delta, \mu,\delta}: \widetilde{\mathbb{D}} \to \mathbb{R}$ defined  by
\begin{equation*}\label{eq:G1_supp}
\Lambda^{-}_{\Delta,\mu,\delta}(\varphi_t):=\int \varphi_t(x)1\{\Delta(x)\leq \delta\}d\mu(x)/\int 1\{\Delta(x)\leq \delta\}d\mu(x),
\end{equation*}
where $\varphi_t \in \mathcal{F}_M$ for moments  and $\varphi_t \in \mathcal{F}_I$ for distributions of the components of $Z$, and $\delta = \Delta^*_{\mu}(u)$ for some $u \in \UU$. The properties of the  most affected operator $\Lambda^{+}_{\Delta, \mu,\delta}: \widetilde{\mathbb{D}} \to \mathbb{R}$ can be derived using similar arguments, which are omitted for brevity.

\begin{lemma}[Hadamard differentiability of $(\Delta, \mu,\delta) \mapsto \Lambda^{-}_{\Delta, \mu,\delta}$]\label{lemma:Ratio-G}
Assume that Assumptions S.1 and S.2 hold, $\delta \in \DD$, and $F_{\Delta,\mu}(\delta) > 0$. Then,

(a) The map $\Lambda^{-}_{\Delta,\mu,\delta}(\varphi_t): \widetilde{\mathbb{D}} \to \mathbb{R}$ is Hadamard-differentiable uniformly  in
$\varphi_t \in \mathcal{F}_M$ at $(\Delta,\mu,\delta)$ tangentially to $\widetilde{\mathbb{D}}_0$.

(b) If in addition Assumption ${\sf AS}.1$ stated in Appendix \ref{app:B}
 of the SM holds, the map $\Lambda^{-}_{\Delta,\mu,\delta}(\varphi_t): \widetilde{\mathbb{D}} \to \mathbb{R}$ is Hadamard-differentiable uniformly  in
$\varphi_t \in \mathcal{F}_I$ at $(\Delta,\mu,\delta)$ tangentially to $\widetilde{\mathbb{D}}_0$.

(c) The derivative map $\partial_{\Delta,\mu,\delta} \Lambda^{-}_{\Delta,\mu,\delta}(\varphi_t): \widetilde{\mathbb{D}} \to \mathbb{R}$ is defined by:
$$(G,H,K)\mapsto \partial_{\Delta,\mu,\delta}\Lambda^{-}_{\Delta,\mu,\delta}(\varphi_t)[G,H,K]:= \int_{\mathcal{M}_\Delta(\delta)}\widetilde \varphi_t(x)\frac{K-G(x)}{\|\partial \Delta(x)\|}d\Vol +H(\widetilde h_{\Delta,\delta,\varphi_t}),$$
where $\widetilde \varphi_t(x) = [\varphi_t(x) -  \Lambda^{-}_{\Delta,\mu,\delta}(\varphi_t)]/\int 1(\Delta(x)\leq \delta)d\mu(x)$ and $\widetilde h_{\Delta,\delta,\varphi_t}:=\widetilde \varphi_t(x)1\{\Delta(x)\leq \delta\}$.
\end{lemma}




%
%
%

\section{Proofs of Section \ref{sec:theory2}}\label{app:proofs}

We first recall Theorem 3.9.4 of \citen{vdV-W}.

\begin{lemma}[Delta-method]\label{lemma:HS}
Let $\mathbb{D}$ and $\mathbb{E}$ be metrizable topological
vector spaces, and $\Theta$
is a compact subset of a metric space. Let $\phi_{\theta}:\mathbb{D}_\phi\subseteq \mathbb{D}\rightarrow
\mathbb{E}$ be a Hadamard differentiable mapping uniformly in $\theta \in \Theta$ at $f\in
\mathbb{D}$ tangentially to $\mathbb{D}_0 \subseteq \mathbb{D}$, with derivative $\partial_f \phi_{\theta}$. Let $\widehat f_n:\Omega_n \rightarrow \mathbb{D}_\phi$ be stochastic maps taking values in $\mathbb{D}_\phi$ such that
 $r_n(\widehat f_n-f)\rightsquigarrow J_{\infty}$ for some sequence of constants $r_n\rightarrow \infty$,
where $J_{\infty}$ is separable and takes values in $\mathbb{D}_0$. Then
$r_n(\phi_{\theta}(\widehat f_n)-\phi_{\theta}(f))\rightsquigarrow \partial_f \phi_{\theta}[J_{\infty}]$, as a stochastic process indexed by $\theta \in \Theta$.
\end{lemma}

\begin{proof}[Proof of Theorem \ref{thm:fclt}]
The statements follow directly from Lemma \ref{lemma:had3}, and Lemma \ref{lemma:HS}, by setting $\phi_{\theta} = F_{\Delta,\mu}(\delta)$ with $\theta = \delta$ or $\phi_{\theta} = \Delta^*_{\mu}(u)$ with $\theta = u$,
$\mathbb{D}_\phi = \mathbb{D}=\mathbb{F}\times\mathbb{H}$, $\mathbb{E}=\mathbb{R}$,  $\mathbb{D}_0=\mathbb{F}_0\times\mathbb{H}$, $f=(\Delta,\mu)$, $\widehat f_n=(\widehat{\Delta}, \widehat \mu)$, and $J_{\infty} = (s_{\Delta}G_{\infty}, s_{\mu} H_{\infty})$. The expression of $\partial_f \phi_{\theta}$ for each statement is the Hadamard derivative in the corresponding statement of Lemma \ref{lemma:had3}.
\end{proof}

\begin{proof}[Proof of Theorem \ref{thm:fclt_supp1}]
The statements follow directly from Lemma \ref{eq:G1_supp}, and Lemma \ref{lemma:HS}, by setting $\phi_{\theta} = \Lambda^{-}_{\Delta,\mu,\delta}$, $\theta=t$,
$\mathbb{D}_\phi=\mathbb{D}=\mathbb{F} \times \widetilde{\mathbb{H}} \times \mathbb{R}$, $\mathbb{E}=\mathbb{R}$, $\mathbb{D}_0=\mathbb{F}_0 \times \widetilde{\mathbb{H}} \times  \mathbb{R}$, $f=(\Delta,\mu,\Delta_{\mu}^*(u))$, $\widehat f_n=(\widehat{\Delta},\widehat{\mu},\widehat{\Delta_{\mu}^*}(u))$, and $J_{\infty} = (s_{\Delta}G_{\infty},s_{\mu}H_{\infty},Z_{\infty})$. The expression of $\partial_f \phi_{\theta}$ for each statement is the Hadamard derivative in the corresponding statement of Lemma \ref{eq:G1_supp}.
\end{proof}

To prove Theorem \ref{thm:bfclt}, we recall Theorem 3.9.11 of \citen{vdV}. Here we use the notation for bootstrap convergence  $\rightsquigarrow_{\Pr}$ defined in Section \ref{subsec:boot}.

\begin{lemma}[Delta-method for bootstrap in probability]\label{lemma: bfdelta}
Let $\mathbb{D}$ and $\mathbb{E}$ be metrizable topological
vector spaces, and $\Theta$
is a compact subset of a metric space. Let $\phi_{\theta}:\mathbb{D}_{\phi}\subseteq
\mathbb{D}\mapsto \mathbb{E}$ be a Hadamard-differentiable
mapping uniformly in $\theta \in \Theta$ at $f$ tangentially to $\mathbb{D}_0$ with derivative $\partial_f \phi_{\theta}$. Let
$\widehat f_n$ be a random element such that
$r_n(\widehat f_n-f)\rightsquigarrow J_{\infty}$. Let
$\widetilde{f}_n$  be a stochastic map in $\mathbb{D}$, produced
by a bootstrap method, such that
$r_n(\widetilde{f}_n-\widehat f_n)\rightsquigarrow_{\Pr}
J_{\infty}$.
Then,
$r_n(\phi_{\theta}(\widetilde{f}_n)-\phi_{\theta}(\widehat f_n))\rightsquigarrow_{\Pr} \partial_f \phi_{\theta}[J_\infty]$, as a stochastic process indexed by $\theta \in \Theta$.
\end{lemma}

\begin{proof}[Proof of Theorem \ref{thm:bfclt}]
The statement (1) follows directly from Lemma \ref{lemma:had3}, and Lemma \ref{lemma: bfdelta}, by setting
$\phi_{\theta} = \Delta^*_{\mu}(u)$, $\theta=u$,
$\mathbb{D}_\phi=\mathbb{D}=\mathbb{F} \times \mathbb{H}$, $\mathbb{E}=\mathbb{R}$, $\mathbb{D}_0=\mathbb{F}_0 \times \mathbb{H}$, $f=(\Delta,\mu)$, $\widehat f_n=(\widehat{\Delta},\widehat{\mu})$, and $J_{\infty} = (s_{\Delta}G_{\infty},s_{\mu}H_{\infty})$. The expression of $\partial_f \phi_{\theta}$ is the Hadamard derivative in statement (b) of Lemma \ref{lemma:had3}. The statement (2) follows directly from Lemma \ref{lemma:Ratio-G}, and Lemma \ref{lemma: bfdelta}, by setting
$\phi_{\theta} = \Lambda^{-}_{\Delta,\mu,\delta}$, $\theta=t$,
$\mathbb{D}_\phi=\mathbb{D}=\mathbb{F} \times \widetilde{\mathbb{H}} \times \mathbb{R}$, $\mathbb{E}=\mathbb{R}$, $\mathbb{D}_0=\mathbb{F}_0 \times \widetilde{\mathbb{H}} \times  \mathbb{R}$, $f=(\Delta,\mu,\Delta_{\mu}^*(u))$, $\widehat f_n=(\widehat{\Delta},\widehat{\mu},\widehat{\Delta_{\mu}^*}(u))$, and $J_{\infty} = (s_{\Delta}G_{\infty},s_{\mu}H_{\infty},Z_{\infty})$. The expression of $\partial_f \phi_{\theta}$ is the Hadamard derivative in statement (c) of Lemma \ref{lemma:Ratio-G}.
\end{proof}

\newpage

\begin{center}
\Large{Supplement to ``The Sorted Effects Method: Discovering Heterogeneous Effects Beyond Their Averages''}

\normalsize{Victor Chernozhukov, Iv\'an Fern\'andez-Val, Ye Yuo}
\end{center}

\begin{abstract}
The supplementary material contains 7 appendices with additional results and some omitted proofs. Appendix \ref{sec:notation} introduces some notation. Appendix \ref{sec:df} includes a brief review of differential geometry. Appendix \ref{app:B} gathers the proofs of the key mathematical results in Appendix \ref{sec:theory1}. Appendix \ref{app:C} provides sufficient conditions for  the $\mu$-Donsker properties in Section \ref{sec:theory2}. Appendix \ref{sec:discrete} extends the theoretical analysis to include discrete covariates. Appendices \ref{subset:numerical} and \ref{sec:mortgage} report the results of 3 numerical simulations and an empirical application to the effect of race on mortgage denials, respectively.

%
\end{abstract}

\appendix

\setcounter{section}{2}
\setcounter{lemma}{9}

\section{Notation}\label{sec:notation}
For a possibly multivariate random variable $X$, $\XX$ denotes the interior of the support of $X$ in the part of the population of interest, $\mu$ denotes the distribution of $X$ over $\XX$, and $\widehat \mu$ denotes an estimator of $\mu$. We denote   the expectation with respect to the distribution $\widetilde \mu$  by  $\Ep_{\widetilde \mu}$.
We denote the PE as $\Delta(x)$, the empirical PE as $\widehat{\Delta}(x)$, and
 $\partial \Delta(x) := \partial \Delta(x)/\partial x$, the gradient of $x \mapsto \Delta(x)$. We also use
$a\wedge b$ to denote the minimum of $a$ and $b$. For a vector $v = (v_1, \ldots, v_{d_v}) \in \mathbb{R}^{d_v}$, $\|v\|$ denotes the Euclidian norm of $v$, that is $\|v\| = \sqrt{v\transp v}$, where the superscript  $\transp$ denotes transpose.
For a non-negative integer $r$ and an open  set $\KK$, the class $\C^r$ on $\KK$ includes the set of $r$ times continuously differentiable real valued functions on  $\KK$.
The symbol $\rightsquigarrow$ denotes  weak convergence (convergence in distribution), and $\to_{\Pr}$ denotes convergence in (outer) probability.

\section{Background on Differential Geometry}\label{sec:df}

We recall some definitions from differential geometry that
are used in the analysis. For a continuously differentiable
function $\Delta: B(\mathcal{X}) \rightarrow \mathbb{R}$ defined on an open set
$B(\mathcal{X}) \subseteq \mathbb{R}^{d_x}$ containing the set $\XX$, $x\in \mathcal{X}$ is a \textit{critical point} of $\Delta$ on $\XX$,
if
\begin{equation}
\partial \Delta(x)=0,
\end{equation}
where $\partial \Delta(x)$ is the gradient of $\Delta(x)$; otherwise $x$ is a \textit{regular point} of $\Delta$ on $\XX$. A value $\delta$ is a
\textit{critical value} of $\Delta$ on $\XX$ if the set $\{x \in \XX: \Delta(x) = \delta \} $ contains  at least one  critical point; otherwise $\delta$ is
a \textit{regular value} of $\Delta$ on $\XX$.

In the multi-dimensional space, $d_x>1,$ a function $\Delta$ can have  continuums
of critical points. For example, the function
$\Delta(x_1,x_2)=\cos(x_1^2+x_2^2)$ has continuums of critical points  on the circles $x_1^2+x_2^2=k\pi$ for each positive integer $k$.


We recall now several core concepts related to manifolds from \citen{spivak-65}  and \citen{munkres-91}.

\begin{definition}[Manifold] Let $d_k$, $d_x$ and $r$ be positive integers such that $d_x \geq d_k$. Suppose that
$\MM$ is a subspace of $\mathbb{R}^{d_x}$ that satisfies the following
property: for each point $m\in \MM$, there is a set $\VV$ containing $m$
that is open in $\MM$, a set $\KK$ that is open in $\mathbb{R}^{d_k}$,
and a continuous map $\alpha_m: \KK \rightarrow \VV$ carrying $\KK$ onto $\VV$
in a one-to-one fashion, such that: (1) $\alpha_m$ is of class $\C^r$ on $\KK$,
(2) $\alpha_m^{-1}:\VV\rightarrow \KK$ is continuous, and (3) the Jacobian matrix of $\alpha_m$, $D\alpha_m(k)$, has
rank $d_k$ for each $k\in \KK$. Then $\MM$ is called a $d_k$-\textit{manifold
without boundary} in $\mathbb{R}^{d_x}$ of class $\C^r$. The map
$\alpha_m$ is called a \textit{coordinate patch} on $\MM$ about $m$. A set of
coordinate patches that covers $\MM$ is called an \textit{atlas}.
\end{definition}

\begin{definition}[Connected Branch]
For any subset $\MM$ of a topological space, if any two points $m_1$ and
$m_2$ cannot be connected via path in $\MM$, then we say that $m_1$
and $m_2$ are not connected. Otherwise, we say that $m_1$ and $m_2$
are connected. We say that $\VV \subseteq \MM$ is a \textit{connected branch} of $\MM$ if all
points of $\VV$ are connected to each other and do not connect to any
points in $\MM \setminus \VV$.
\end{definition}

%

\begin{definition}[Volume] \label{def:volume} For a $d_x \times d_k$ matrix $A=(x_1,x_2,...,x_{d_k})$
 with $x_i \in \mathbb{R}^{d_x}$, $1\leq
i\leq d_k \leq d_x$, let $\mathrm{Vol}(A)=\sqrt{\det(A\transp A)}$, which is the \textit{volume} of the parallelepiped $P(A)$ with edges given by the columns of $A$, $P(A) = \{c_1 x_1 + \cdots + c_{d_k} x_{d_k} : 0 \leq c_i \leq 1, i = 1, \ldots, d_k \}$.
\end{definition}

The volume measures the amount of  mass in $\mathbb{R}^{d_k}$ of a $d_k$-dimensional
parallelepiped in $\mathbb{R}^{d_x}$, $d_k \leq d_x$. This concept is
essential for integration on  manifolds, which we will discuss shortly. First we
recall the concept of integration on parameterized manifolds:

\begin{definition}[Integration on a parametrized manifold] Let $\KK$ be open in $\mathbb{R}^{d_k}$, and let $\alpha: \KK \to \mathbb{R}^{d_x}$ be of class $\C^r$ on $\KK$, $r\geq 1$. The set  $\MM=\alpha (\KK)$ together with the map $\alpha$ constitute a \textit{parametrized $d_k$-manifold} in $\mathbb{R}^{d_x}$ of class $\C^r$.
Let $g$ be a real-valued continuous function defined at each point
of $\MM$. The \textit{integral of $g$ over $\MM$ with
respect to volume} is defined by
\begin{equation}
\int_{\MM}g(m) d\mathrm{Vol}:=\int_{\KK}(g \circ \alpha)(k) \mathrm{Vol}(D\alpha(k)) dk,
\end{equation}
provided that the right side integral exists. Here $D\alpha(k)$ is the Jacobian
matrix of the mapping $k \mapsto \alpha(k)$, and $\mathrm{Vol}(D\alpha(k))$ is the volume of
matrix $D\alpha(k)$ as defined in Definition \ref{def:volume}.
\end{definition}

The above definition coincides with the usual interpretation of
integration. The integral can be extended to  manifolds that do not admit a global parametrization $\alpha$ using the notion of partition of unity.  This partition is a set of smooth local functions defined in a
neighborhood of the manifold. The following Lemma  shows the  existence of the partition of unity  and is proven in Lemma 25.2 in \citen{munkres-91}.

\begin{lemma}[Partition of Unity on $\MM$ of class $\C^{\infty}$]\label{lemma:pou} Let $\MM $ be a $d_k$-manifold without boundary in $\mathbb{R}^{d_x}$ of class $\C^r$, $r \geq 1$, and let $\vartheta$ be an open cover of
$\MM$. Then, there is a collection $\mathcal{P} = \{ p_i  \in \C^\infty : i \in \II \}$, where $p_i$ is defined on an open set containing $\MM$ for all $i \in \II$, with the
following properties: (1) For each $m\in \MM$ and $i \in \II$,  $0\leq p_i(m)\leq
1$, (2) for each $m\in \MM$ there is an open set $\VV \in \vartheta$ containing $m$
such that all but finitely many $p_i \in \mathcal{P}$ are $0$ on $\VV$,
(3) for each $m\in \MM$, $\sum_{p_i \in \mathcal{P}}{p_i(m)}=1$, and (4)
for each $p_i \in \mathcal{P}$ there is an open set $\UU \in \vartheta$,
such that $supp(p_i)\subseteq \UU$.
\end{lemma}

Now we are ready to recall the definition of integration on a manifold.

\begin{definition}[Integration on a manifold with partition of
unity]\label{def:intmani} Let $\vartheta:=\{\vartheta_j : j\in \JJ \}$ be an open cover of a
$d_k$-manifold without boundary $\MM$ in $\mathbb{R}^{d_x}$ of class $\C^r$, $r \geq 1$. Suppose there is an
coordinate patch $\alpha_j: \VV_j\subseteq \mathbb{R}^{d_k} \rightarrow
\vartheta_j$, that is one-to-one and of class $\C^r$ on $\VV_j$ for each $j\in \JJ$. Denote
$\KK_j=\alpha_j^{-1}(\MM \cap \vartheta_j)$.
Then for a real-valued
continuous function $g$ defined on an open set that contains
$\MM$, the \textit{integral of $g$ over $\MM$ with respect to volume} is defined by:
\begin{equation}
\int_{\MM} g(m) d\mathrm{Vol} :=\sum_{j\in \JJ}\sum_{i\in \II} \int_{\KK_j} [(p_i g)\circ
\alpha_j](k) \mathrm{Vol}(D\alpha_j(k))dk,
\end{equation}
provided that the right side integrals exist, where $\{p_i \in \C^\infty : i \in \II\}$ is a partition of unity on $\MM$ of class $\C^{\infty}$ that satisfies the conditions of Lemma \ref{lemma:pou}.  \citen[p. 212]{munkres-91} shows that  the integral does not depend on the choice of cover and partition of unity. \end{definition}

\section{Proofs of Appendix \ref{sec:theory1}}\label{app:B}

To analyze the analytical properties of the SPE-function, it is convenient to treat the PE  as a multivariate real-valued function $$\Delta:  B(\mathcal{X}) \to \mathbb{R},$$ where  $B(\mathcal{X}) \subseteq \mathbb{R}^{d_x}$ contains the set $\XX$.
Let  $\mu$ be
a  distribution function.  The distribution of  $\Delta$ with respect to $\mu$ is the function $F_{\Delta,\mu}: \mathbb{R} \to [0,1]$ with
\begin{equation}\label{eq: dpe}
F_{\Delta,\mu}(\delta)= \int 1\{\Delta(x) \leq \delta \} d\mu(x).
\end{equation}
The SPE-function is the map
$$\Delta_{\mu}^*:\mathcal{U} \subseteq [0,1] \rightarrow \mathbb{R},$$
defined at each point as the left-inverse  function of $F_{\Delta,\mu}$, i.e.,
\begin{equation}\label{eq: qpe}
\Delta_{\mu}^*(u) := F_{\Delta,\mu}^{\leftarrow}(u):=\inf_{\delta
\in \Bbb{R} } \{F_{\Delta,\mu}(\delta)\geq u\}.
\end{equation}
 From this functional  perspective,  the map $u \mapsto \Delta_{\mu}^*(u)$ is the result of applying a sorting operator to the map  $x \mapsto \Delta(x)$ that sorts the values of $\Delta$ in increasing order weighted by $\mu$. The next subsections provide the proofs of 3 results:
\begin{itemize}
\item[1)] Lemma \ref{lemma:had0}, which characterizes some analytical properties of the distribution function $\delta \mapsto F_{\Delta,\mu}(\delta)$  and the sorted function $u \mapsto \Delta_{\mu}^*(u)$,

\item[2)] Lemma \ref{lemma:had3}, which derives the functional derivatives of $F_{\Delta,\mu}$ and $ \Delta_{\mu}^*$ with respect to $\Delta$ and $\mu$, and

\item[3)] Lemma \ref{lemma:Ratio-G}, which derives the functional derivatives of the related classification operator $\Lambda^{-}_{\Delta, \mu,\delta}$ with respect to $\Delta$, $\mu$ and $\delta$.

\end{itemize}

\subsection{Proof of Lemma \ref{lemma:had0}}

We use the following results in the proof of Lemma \ref{lemma:had0}.

\begin{lemma}\label{manifold}
If $\Delta: B(\XX) \rightarrow \mathbb{R}$ is $\C^1$ on an open set $B(\XX) \subseteq
\mathbb{R}^{d_x}$, then for any compact subset $\overline \XX$ of $B(\XX)$, the sets of
critical points and critical values of $x \mapsto \Delta(x)$ on $\overline \XX$ are closed.
\end{lemma}

\begin{proof}

(1) Critical points: since $x \mapsto \partial \Delta(x)$ is continuous on $\overline{\XX}$ and $\overline \XX$ is compact, the set of
points $x \in \overline \XX$ such that $\partial \Delta(x)=0$
is closed.

(2) Critical values: since $x \mapsto \Delta(x)$ is continuous and $\overline \XX$ is compact, the image set $\Delta(\overline \XX)$ is a compact set in $\mathbb{R}$. For any sequence of critical values $\{\delta_i\}_{i\geq
1}^{\infty}$ in $\Delta(\overline \XX)$, there is a corresponding sequence
$\{x_i\}_{i\geq 1}$ in $ \overline \XX$ such that $\Delta(x_i)=\delta_i$. Suppose $\{\delta_i\}_{i\geq
1}^{\infty}$ converges to $\delta_0\in \Delta(\overline \XX)$. By compactness of $\overline \XX$, we can find a converging subsequence of $\{x_i\}_{i\geq 1}$ with limit $x_0\in \overline \XX$ such that $\Delta(x_i)=\delta_i$. Then by continuity of
$x \mapsto \partial \Delta(x)$, $\partial \Delta(x_0)=0$. By continuity of $x \mapsto \Delta(x)$, $\Delta(x_0)=\delta_0$, and therefore $\delta_0=\Delta(x_0)$ is a critical value of $\Delta(x)$.
Hence the set of critical values is closed.
\end{proof}

\begin{lemma}\label{lemma:cover}
For a compact set $\mathcal{V}$ in a metric space $\mathbb{D}$, suppose
there is an open cover $\{\theta_i : i\in I\}$ of $\mathcal{V}$. Then there exists a
finite open sub-cover of $\mathcal{V}$ and  $\eta>0$, such that
for every point  $x \in \mathcal{V}$, the $\eta$-ball around $x$ is contained
in the finite sub-cover.
\end{lemma}

\begin{proof}[Proof of Lemma \ref{lemma:cover}]

Since $\mathcal{V}$ is a compact set in the metric space $\mathbb{D}$ (with metric
$\|\cdot\|_{\mathbb{D}}$), then any open cover $\{\theta_i : i\in I\}$ of $\mathcal{V}$ has a
finite open subcover $\{\widetilde \theta_i : i=1,2,...,m\}$ which covers $\mathcal{V}$.

Let $\Theta=\cup_{i=1}^m \widetilde \theta_i$. We prove the statement of the lemma by contradiction.
Suppose for any $i>0$, there exists some point $x_i \in \mathbb{D}$ such that $d(x_i,\mathcal{V}):=\inf_{v \in
\mathcal{V}}\|x_i-v\|_{\mathbb{D}} < i^{-1}$ and $x_i \notin \Theta$. Then, by compactness of $\mathcal{V}$ there
exists $v_i \in \mathcal{V}$ such that
$d(x_i,\mathcal{V})=d(x_i,v_i) < i^{-1}$. Let $v_0$ be the limit of  $\{v_i : i\geq 1\}$.  By compactness of $\mathcal{V}$, $v_0\in
\mathcal{V}$. Since $d(x_i,v_0) \to 0$ as $i \to \infty$ and $\Theta$
is an open cover of $\mathcal{V}$, there must be a open ball
$B(v_0)$ around $v_0$ such that $B(v_0) \subseteq \Theta$, which
contradicts with  $x_i\notin
\Theta$, for $i$ large enough. Therefore there must be an $\eta$ such that the $\eta$-ball around any $x \in \mathcal{V}$ is covered by $\Theta$.
\end{proof}

\begin{proof}[Proof of Lemma \ref{lemma:had0}]
The proof of statement (2) follows directly from the inverse
function theorem.

The proof of statement (1) is divided in two steps. Step 1 constructs a finite set of open rectangles that covers the set $\MM_{\Delta}(\delta)$ and has certain properties that allow us to apply a change of variable to the derivative of $\delta \mapsto F_{\Delta,\mu}(\delta)$. Step 2 expresses the derivative as an integral on a manifold.

For a subset $\Ss \subseteq \mathbb{R}^{d_x}$ and $\eta>0$, define $B_\eta(\Ss):=\{x \in \mathbb{R}^{d_x}: d(x,\Ss) = \inf_{s \in \Ss} \|x - s \| < \eta\}$. Similarly, for any $\delta\in \mathbb{R}$ and $\eta>0$, define $B_{\eta}(\delta):=(\delta-\eta,\delta+\eta)$. Without loss of
generality, we assume that $\MM_{\Delta}(\delta)$ only has one connected branch. We will discuss the case
where $\MM_{\Delta}(\delta)$ has multiple connected branches at the end of the proof of this lemma.

\begin{step}

For any regular value $\delta\in \mathcal{D}$,  the set $\MM_{\Delta}(\delta)$ is a $(d_x-1)$-manifold in $\mathbb{R}^{d_x}$ of class $\C^1$ by Theorem 5-1 in \citen[p. 111]{spivak-65}. Denote $\widetilde{\MM}_{\Delta}(\delta):=\{x\in B(\XX) : \Delta(x)=\delta\}$ and $\widetilde{\MM}_{\Delta}(B_{\eta}(\delta)):=\cup_{\delta'\in B_{\eta}(\delta)}\widetilde{\MM}_{\Delta}(\delta')$ for $\eta>0$. These enlargements of the set $\MM_{\Delta}(\delta)$ are used to apply a change of variable technique to integrals on $\MM_{\Delta}(\delta)$.

By  assumptions ${\sf S}.1$-${\sf S}.2$, there exists  $\eta_1>0$ small enough and $C>c>0$ such that:

(1)
$\overline{B_{\eta_1}(\delta)}:=[\delta-\eta_1,\delta+\eta_1]\subseteq \Delta(\overline{\mathcal{X}}):=\{\Delta(x):x\in \overline{\mathcal{X}}\}$ and  contains no critical values of
$\Delta$ on $\overline{\mathcal{X}}$, and $B_{\eta_1}(\overline{\XX}) \subseteq B(\XX)$.

(2) $\inf_{x\in \overline{\widetilde{\MM}_{\Delta}(B_{\eta_1}(\delta))\cap B_{\eta_1}(\overline{\XX})}}\|\partial \Delta(x)\|>c$.

(3) $\sup_{x\in \overline{\widetilde{\MM}_{\Delta}(B_{\eta_1}(\delta))\cap B_{\eta_1}(\overline{\XX})}}\|\partial \Delta(x)\|<C$.

(4) For any $\eta<\eta_1$, $\widetilde{\MM}_{\Delta}(\delta)\cap B_{\eta}(\overline{\XX})$ is a $(d_x-1)$-manifold in $\mathbb{R}^{d_x}$ of class $\C^1$.

Indeed, by Lemma \ref{manifold}, the set of regular values is open. Therefore, there exists a small neighborhood $B_{\eta}(\delta)$ with $\eta>0$ such that there exists no critical value of $\Delta$ on $\overline{\XX}$ in $B_{\eta}(\delta)$. Then any $\eta_1<\eta$ satisfies statement (1). Statements (2) and (3) follow by the compactness of $\overline{\mathcal{X}}$, the continuity of mapping $x\mapsto \partial \Delta(x)$, and assumptions ${\sf S}.1$ and ${\sf S}.2$. Statement (4) is implied by Theorem 5-1 in \citen[p. 111]{spivak-65}.

Next, we establish a finite cover of $\widetilde{\MM}_{\Delta}(B_{\eta_2}(\delta))\cap B_{\eta_2}(\overline{\XX})$ with certain good properties, for some $\eta_2<\eta_1$.

For any $\eta_3<\eta_1$, $\widetilde{\MM}_{\Delta}(B_{\eta_3}(\delta))\cap B_{\eta_3}(\overline{\XX})$ satisfies the properties (2)--(4) stated above. Consider the rectangles $\theta(x):=X_1(x)\times...\times X_{d_x}(x)$ centered at $x=(x_1,...,x_{d_x})$ where $X_k(x):=(x_k-a_k(x),x_k+a_k(x))$, with $a_k(x)>0$, $k=1,2,...,d_x$. Let $A(x):=\sup_{1\leq k\leq d_x}a_k(x)$ be such that:
$$\overline{\widetilde{\MM}_{\Delta}(B_{\eta_3}(\delta))\cap B_{\eta_3}({\XX})}\subseteq \cup_{x\in \widetilde{\MM}_{\Delta}(\delta)\cap B_{\eta_3}{(\mathcal{X})}}\theta(x)\subseteq \widetilde{\MM}_{\Delta}(B_{\eta_1}(\delta))\cap B_{\eta_1}(\overline{\XX}),$$
which can be fulfilled by using small enough $\eta_3$.

By continuity of $x \mapsto \partial \Delta(x)$,  for small enough $A(x)$ and any $x'\in \theta(x)$, there always exists an index $\text{i}(x)\in\{1,2,...,d_x\}$ such that $|\partial_{x_{\text{i}(x)}}\Delta(x')|\geq \frac{c}{2\sqrt{d_x}}$ since $\|\partial \Delta(x') \| \geq c$ for all $x' \in \theta(x)$ by the property (2) above, where $\partial_{x} := \partial / \partial_{x}$. Also we can find a finite set of $\theta(x)$'s, denoted as $\Theta:=\{\theta(x^i)\}_{i=1}^m$, such that
 $\Theta$ forms a finite open cover of $\overline{\widetilde{\MM}_{\Delta}(B_{\eta_3}(\delta))\cap B_{\eta_3}(\overline{\XX})}$. We rename these open rectangles as $\theta_i:=\theta(x^i)$, $i\in \{1,2,...,m\}$, where $\theta_i = X_{i1}\times...\times X_{id_x}$ and $X_{ik} := X_k(x^i)$, $k \in \{1, \ldots, d_x\}$.


For a given $i\in \{1,2,...,m\}$, consider the center of $\theta_i$, denoted as $x^i$. Without loss of generality, we can assume that $\text{i}(x^i)=d_x$. Then, for all $x'\in \theta(x^i)$, $|\partial_{x_{d_x}} \Delta(x')|\geq c/2\sqrt{d_x}$. This means that $\Delta(x)$ is partially monotonic in $x_{d_x}$ on $\theta(x^i)$. By the implicit function theorem, there exists $g$ such that $g(x_1',x_2',...,x_{d_x-1}',\delta')=x_{d_x}'$,
for any $x'=(x_1',x_2',...,x_{d_x}')\in \widetilde{\MM}_{\Delta}(B_{\eta_3}(\delta))\cap \theta(x^i)$ and $\delta'=\Delta(x')$.  Also by the implicit function theorem, $$\partial g(x_1',...,x_{d_x-1}',\delta')= \frac{-(\partial_{x_1} \Delta(x'),\partial_{x_2} \Delta(x'),..., \partial_{x_{d_x-1}} \Delta(x'),-1)}{\partial_{d_x} \Delta(x')}.$$ So $\|\partial g(x_1',...,x_{d_x-1}',\delta')\|\leq \frac{\|\partial \Delta(x')\|}{|\partial_{x_{d_x}}\Delta(x')|}\leq \frac{2(C+1)\sqrt{d_x}}{c}:=\Lambda$ because $|\partial_{x_{d_x}}\Delta(x')| \geq c/2\sqrt{d_x}$ and $\|\partial \Delta(x')\| \leq C$.
Therefore,
\begin{multline*}|g(x_1',x_2',...,x_{d_x-1}',\delta')-x^i_{d_x}|=|g(x_1',x_2',...,x_{d_x-1}',\delta')-g(x^i_1,x^i_2,...,x^i_{d_x-1},\delta)| \\
\leq \sup_{x'\in \theta(x),\delta'=\Delta(x')}\|\partial g(x_1,x_2,...,x_{d_x-1},\delta')\| \cdot \|(x_{1}-x^{i}_{1},x_2-x^{i}_{2}...,x_{d_x-1}-x^{i}_{d_x-1},\delta'-\delta)\|\\
\leq \Lambda (\sqrt{a_1^2(x^i)+...+a_{d_x-1}^2(x^i)}+\eta_3),
\end{multline*}

since $\|(x_1-x^i_1,...,x_{d_x-1}-x^{i}_{d_x-1},\delta'-\delta)\|\leq \|(x_1-x^i_1,...,x_{d_x-1}-x^{i}_{d_x-1})\|+|\delta'-\delta|$, with $\|(x_1-x^i_1,...,x_{d_x-1}-x^{i}_{d_x-1})\|\leq \sqrt{a_1^2(x^i)+...+a_{d_x-1}^2(x^i)}$ and $|\delta'-\delta|<\eta_3$.

We can choose  $a_1(x^i)=a_2(x^i)=...=a_{d_x-1}(x^i)=\eta_4$ and $a_{d_x}(x^i)=2(1+\eta_3)\Lambda(\sqrt{d_x-1}\eta_4+\eta_3)$, using $\eta_4$ small enough in order to fulfill the following property of $\theta_i$: with $\eta_4$ small enough,  $$\widetilde{\MM}_{\Delta}(B_{\eta_3}(\delta))\cap \theta_i\subseteq  X_{i1}\times...\times X_{i,d_x-1}\times \left(x^i_{d_x}-\frac{a_{d_x}(x^i)}{2(1+\eta_3)},x^i_{d_x}+\frac{a_{d_x}(x^i)}{2(1+\eta_3)}\right),$$ or geometrically, the tube $\widetilde{\MM}_{\Delta}(B_{\eta_3}(\delta))$ does not intersect $\theta_i$'s faces except at the ones which are \textit{parallel} to the vector $(0,...,0,1)\in \RR^{d_x}$. In such a case, we say that
$\widetilde{\MM}_{\Delta}(B_{\eta_3}(\delta))$ intersects $\theta_i$ at the axis $x_{d_x}$. More generally, for all $i\in \{1,2,...,m\}$, $\widetilde{\MM}_{\Delta}(B_{\eta_3}(\delta))$ intersects $\theta_i$ at axis $\text{i}(x^i)$, where $x^i$ is the center of
$\theta_i$. This property implies that $g$ is a well-defined injection from $X_{i1}\times...\times X_{i,d_x-1}\times B_{\eta_3}(\delta)$  to
$X_{i1}\times ...\times X_{i,d_x}$, for $i \in \{1,\ldots,m\}$, which will allow us to perform a change of variable in the equation (\ref{eq: maindec}). Such a property holds for any
$\eta_2 <\eta_3$. 
\end{step}

\begin{step}  Let $\eta_2$ be such that $0 < \eta_2 < \eta_3$.
We first apply partition of unity to the open cover $\Theta = \{\theta_i\}_{i=1}^m$ of $\overline{\widetilde{\MM}_{\Delta}(B_{\eta_2}(\delta))\cap B_{\eta_2}(\overline{\XX})}$ of Step 1.

By Lemma \ref{lemma:pou}, for the finite open cover $\Theta$ of the
manifold $\widetilde{\MM}_{\Delta}(B_{\eta_2}(\delta))\cap B_{\eta_2}(\overline{\XX})$, we can find a set of
$\C^\infty$ partition of unity $p_j$, $1\leq j\leq J$ on $\Theta$ with the properties
given in the lemma.

Our main goal is to compute
$$
\partial_{\delta} F_{\Delta,\mu}(\delta) = \lim_{h \to 0} \frac{F_{\Delta,\mu}(\delta + h) - F_{\Delta,\mu}(\delta)}{h}.
$$

Denote $B_{\eta}^+(\delta)=[\delta,\delta+\eta]$, for any $\delta\in \mathbb{R}$ and $\eta>0$. Denote ${\MM}_{\Delta}(B_{\eta}^+(\delta))=\cup_{\delta'\in B_{\eta}^+(\delta)}{\MM}_{\Delta}(\delta')$, and $\widetilde{\MM}_{\Delta}(B_{\eta}^+(\delta))
=\cup_{\delta'\in B_{\eta}^+(\delta)}\widetilde{\MM}_{\Delta}(\delta')$.

For any $0<\eta<\eta_2$, $\widetilde{\MM}_{\Delta}(B_{\eta}^+(\delta))\subseteq \widetilde{\MM}_{\Delta}(B_{\eta}(\delta))$. Therefore, the properties (1) to (4) stated in Step 1 are satisfied when we replace $\widetilde{\MM}_{\Delta}(B_{\eta}(\delta))$ by $\widetilde{\MM}_{\Delta}(B_{\eta}^+(\delta))$. Note that,
\begin{multline}\label{eq: Decp}
F_{\Delta,\mu}(\delta+\eta)-F_{\Delta,\mu}(\delta)=\int_{x\in \mathcal{X}}1(\delta\leq \Delta(x)\leq \delta+\eta)\mu'(x)dx
\\=\int_{{\MM}_{\Delta}(B_{\eta}^+(\delta))}\mu'(x)dx
=\int_{\widetilde{\MM}_{\Delta}(B_{\eta}^+(\delta))}\mu'(x)dx
=\int_{\widetilde{\MM}_{\Delta}(B_{\eta}^+(\delta))\cap \Theta}\mu'(x)dx
\\=\int_{\widetilde{\MM}_{\Delta}(B_{\eta}^+(\delta))\cap (\cup_{i=1}^m\theta_i)}\mu'(x){\sum_{j=1}^J p_j(x)dx}
=\sum_{1\leq i\leq m,1\leq
j\leq J}{\int_{\widetilde{\MM}_{\Delta}(B_{\eta}^+(\delta))\cap\theta_i}p_j(x)\mu'(x) dx}.
\end{multline}
This third and fourth equalities hold because $\mu'(x)=0$ for any $x \in \widetilde{\MM}_{\Delta}(B_{\eta}^+(\delta))\setminus \MM_{\Delta}(B_{\eta}^+(\delta))$ and $x\in \widetilde{\MM}_{\Delta}(B_{\eta}^+(\delta)) \setminus \Theta$, respectively.

For any $i\in \{1,2,...,m\}$, without loss of generality, suppose that $\MM_\Delta(B_{\eta}^+(\delta))$ intersects
$\theta_i=X_{i1}\times...\times X_{id_x}$ at the $x_{d_x}$ axis. Then, $|\partial_{x_{d_x}}\Delta(x)|\geq c/\sqrt{d_x}$ on $\theta_i$, and we can apply the implicit function theorem to
show existence of the $\mathcal{C}^1$ implicit function
$
g:X_{i1}\times...\times X_{i (d_x-1)}\times
B_\eta^+(\delta)\rightarrow X_{id_x},
$
such that $\Delta(x_1,...,x_{d_x-1},g(x_1,...,x_{d_x-1},\delta'))=\delta'$ for all
$(x_1,...,x_{d_x-1},\delta')\in X_{i1}\times ...\times X_{i
(d_x-1)}\times B_\eta^+(\delta)$. Define the injective mapping $\psi_{d_x}$ as:
\begin{eqnarray*}
&& \psi_{d_x}:X_{i1}\times ...\times X_{i(d_x-1)}\times
 B_\eta^+(\delta)\rightarrow X_{i1}\times...\times
X_{i(d_x-1)}\times X_{i(d_x)}, \\
& & \psi_{d_x}(x_{-d_x},\delta')=(x_{-d_x},g(x_{-d_x},\delta')) \text{ for } x_{-d_x} := (x_1,x_2,...,x_{d_x-1}).
\end{eqnarray*}

In equation (\ref{eq: Decp}),  we apply  a change of variable defined by the map $\psi_{d_x}$ to the $(i,j)$-th element of the sum:
$$\int_{\theta_i\cap \widetilde{\MM}_\Delta(B_{\eta}^+(\delta))}p_j(x)\mu'(x) dx=\int_{X_{i1}\times ...\times X_{i(d_x-1)}\times B_\delta^+(\eta)}{(p_j\circ \psi_{d_x})\cdot(\mu'\circ \psi_{d_x}) |\text{det}( D \psi_{d_x})}|d\delta' dx_{-d_x}$$
$$=\int_{X_{i1}\times ...\times X_{i(d_x-1)}}{\int_{B_\delta^+(\eta)}\frac{(p_j\circ\psi_{d_x})\cdot (\mu'\circ \psi_{d_x}) }{|\partial_{x_{d_x}} \Delta \circ \psi_{d_x}|}}d\delta' dx_{-d_x}$$
\begin{equation}\label{eq: maindec}=\eta \int_{X_{i1}\times ...\times X_{i(d_x-1)}}{\frac{(p_j\circ \psi_{d_x})\cdot (\mu'\circ \psi_{d_x})}{|\partial_{x_{d_x}} \Delta \circ \psi_{d_x}|}}dx_{-d_x}+o(\eta).\end{equation}
The second equality follows because
$$ D {\psi_{d_x}} (x_{-d}, \delta) = \[
\begin{array}{ccccc}
1&0&...&0\\
0&1&...&0\\
...&...&...&...\\
0&...&...&\partial_{\delta} g(x_{-d_x}, \delta)
\end{array}
\] =  \[
\begin{array}{ccccc}
1&0&...&0\\
0&1&...&0\\
...&...&...&...\\
0&...&...&1/\partial_{x_{d_x}} \Delta (\widetilde x)
\end{array}
\],$$
where $\widetilde x = \psi_{d_x}(x_{-d}, \delta)$.

The last equality follows as $\eta \to 0$, because by the uniform continuity of $$(x_{-d_x},\delta') \mapsto (p_j\circ\psi_{d_x})\cdot (\mu'\circ \psi_{d_x}) / |\partial_{x_{d_x}} \Delta \circ \psi_{d_x}| \Big|_{(x_{-d_x},\delta')}$$
over $(x_{-d_x},\delta')\in X_{i1}\times ...\times X_{i(d_x-1)}\times B_{\eta}^+(\delta)$.
In (\ref{eq: maindec}), the last component of $\psi_{d_x}$
is fixed to be $\delta$ without being specified for simplicity. We will maintain this convention in the rest of the proof whenever the variable of integration is $x_{-d_x}$ (excluding $x_{d_x}$).

Next, we write the first term of (\ref{eq: maindec}) as an  integral
on a manifold, which is
\begin{equation}\label{eq: md}
\eta \int_{X_{i1}\times ...\times X_{i(d_x-1)}}{\frac{(p_j\circ \psi_{d_x}) \cdot (\mu'\circ
\psi_{d_x})}{|\partial_{x_{d_x}} \Delta \circ
\psi_{d_x}|}}dx_{-d_x} = \eta \int_{\widetilde{\MM}_{\Delta}(\delta) \cap \theta_i}
 \frac{p_j(x)\mu'(x)}{\parallel \partial \Delta(x)\parallel}d\mathrm{Vol}.
\end{equation}
Summing up over $i$ and $j$ in  (\ref{eq: maindec}) and using  Definition 5.5,
\begin{equation}\label{eq: maineq}
\sum_{1\leq i\leq m,1\leq
j\leq J}{\int_{\widetilde{\MM}_{\Delta}(B_{\eta}^+(\delta))\cap\theta_i}p_j(x)\mu'(x) dx} =\eta\int_{\widetilde{\MM}_{\Delta}(\delta)\cap \Theta} \frac{\mu'(x)}{\parallel \partial
\Delta(x)\parallel}d\mathrm{Vol}+o(\eta).
\end{equation}

Let us explain(\ref{eq: md}). Equation (\ref{eq: md}) is calculated using the following fact: The mapping $\alpha:
X_{i1}\times ...\times X_{id_x-1}\rightarrow X_{i1}\times ...
\times X_{id_x}$ such that
$\alpha(x_1,...,x_{d_x-1})=(x_1,...,x_{d_x-1},g(x_1,...,x_{d_x-1}, \delta))$ has
Jacobian matrix
$$D\alpha\transp(x_{-d_x})= \[
\begin{array}{ccccc}
1&0&...&0&\partial_{x_1} g (x_{-d_x})\\
0&1&...&0&\partial_{x_2} g (x_{-d_x})\\
...&...&...&...&...\\
0&...&...&1&\partial_{x_{d_x-1}} g(x_{-d_x})
\end{array}
\] =
\[
\begin{array}{ccccc}
1&0&...&0&(\partial_{x_1} \Delta/\partial_{x_{d_x}}
\Delta)(\widetilde x) \\
0&1&...&0&(\partial_{x_2} \Delta/ \partial_{x_{d_x}}
\Delta) (\widetilde x)\\
...&...&...&...&...\\
0&...&...&1&(\partial_{x_{d_x-1}} \Delta/ \partial_{x_{d_x}}\Delta)
(\widetilde x)
\end{array}
\],
$$
where $\widetilde x =  (x_1,...,x_{d_x-1},g(x_1,...,x_{d_x-1}, \delta))$.
The volume of $D\alpha$ is $\mathrm{Vol}(D\alpha)=\sqrt{\det(D\alpha\transp D\alpha)}$, where $D\alpha\transp D\alpha=I_{d_x-1}+\partial g \partial g\transp$. By the Matrix Determinant Lemma,
$$\mathrm{Vol}(D\alpha)(x_{-d_x}) = \sqrt{1 + \partial g \transp \partial g}=\|\partial \Delta \|/|\partial_{x_{d_x}}
\Delta| \Big |_{x = \widetilde x}.$$
Hence, the left hand side of equation (\ref{eq: md}) is:
\begin{equation*}
\eta \int_{X_{i1}\times ...\times
X_{i(d_x-1)}}{ \frac{(p_j\circ \psi_{d_x}) \cdot (\mu'\circ \psi_{d_x})}{\|\partial
\Delta\circ \psi_{d_x}\|}}\mathrm{Vol}(D\alpha)d{x_{-d_x}},
\end{equation*}
and it can be further re-expressed as the right side of (\ref{eq: md}) using Definition 5.4.

By equations (\ref{eq: Decp}) and  (\ref{eq: maineq}),
\begin{equation}\label{eq: asy}
\frac{F_{\Delta,\mu}(\delta+\eta)-F_{\Delta,\mu}(\delta)}{\eta}=\int_{\MM_{\Delta}(\delta)}
\frac{\mu'(x)}{\parallel \partial \Delta(x)\parallel}d\mathrm{Vol}+o(1),
\end{equation}
where we use that $\mu'(x) = 0$ for all $x \in \widetilde{\MM}_{\Delta}(\delta) \setminus \MM_{\Delta}(\delta)$. Similarly, we can show that
\begin{equation*}
\frac{F_{\Delta,\mu}(\delta)-F_{\Delta,\mu}(\delta-\eta)}{\eta}=\int_{\MM_{\Delta}(\delta)}
\frac{\mu'(x)}{\parallel \partial \Delta(x)\parallel}d\mathrm{Vol}+o(1).
\end{equation*}
Thus, we conclude that $F_{\Delta,\mu}(\delta)$ is
differentiable at $\delta\in \mathcal{D}$ with derivative
\begin{equation*}
f_{\Delta,\mu}(\delta) := \partial_{\delta} F_{\Delta,\mu}(\delta) =\int_{\MM_{\Delta}(\delta)} \frac{\mu'(x)}{\parallel
\partial \Delta(x)\parallel}d\mathrm{Vol}.
\end{equation*}
\end{step}

Finally, if $\MM_{\Delta}(\delta)$ has multiple branches but a finite number of them, we can repeat Step 1 and 2 in the proof above for each individual branch. Since the number of connected branches is finite,  the remainders in equation (\ref{eq: asy}) converge to $0$ uniformly. Thus, adding up the results for all connected branches in equation (\ref{eq: asy}), the statements of Lemma \ref{lemma:had0} hold.
\end{proof}

\subsection{Proof of Lemma \ref{lemma:had3}}

We use the following results in the proof of Lemma \ref{lemma:had3}.

\begin{lemma}[Continuity]\label{lemma:continuity}
Let $f$ be a measurable function defined on $B_\eta(\mathcal{X})\subset B(\mathcal{X})$ which vanishes outside $\mathcal{X}$, where $\eta>0$ is a constant. Let $\delta$ be a regular value of $\Delta$ on $\overline{\mathcal{X}}$. Suppose $f$ is continuous on ${\widetilde{\MM}_{\Delta}(B_{\eta_1}(\delta))\cap B_{\eta_1}(\overline{\XX})}$ for any $\delta \in \DD$ and some small $\eta_1$ such that $0<\eta_1<\eta$. Then, $\delta \mapsto \int_{\MM_\Delta(\delta)}fd \mathrm{Vol}$ is continuous on $\DD$.
\end{lemma}

\begin{proof}
First, we follow Step 1 in the Proof of Lemma \ref{lemma:had0}. Suppose we have a set of open rectangles $\Theta=\{\theta_1,...,\theta_m\}$ such that $\overline{\widetilde{\MM}_{\Delta}(B_{\eta_2}(\delta))\cap B_{\eta_2}(\overline{\XX})}\subset\cup_{i=1}^m \theta_i\subset\overline{\cup_{i=1}^m \theta_i}\subset \widetilde{\MM}_{\Delta}(B_{\eta_1}(\delta))\cap B_{\eta_1}(\overline{\XX})$ for any $\eta_2<\eta_3$, where $\eta_3$ is a small enough positive number, $\eta_3<\eta_1$. Moreover, let $\eta_3$ be small enough such that all $\delta'\in B_{\eta_3}(\delta)$ are regular values. By compactness of $\overline{\cup_{i=1}^m \theta_i}$, $f$ is bounded and uniformly continuous on $\cup_{i=1}^m\theta_i$.

By construction, $\theta_i$, $i=1,2,...,m$, satisfies that $\widetilde{\MM}_{\Delta}(B_{\eta_3})$ intersects $\theta_i$ at axis $\text{i}(\theta_i)$, for any $\eta_2<\eta_3$.

Then, following Step 2 in the Proof of Lemma \ref{lemma:had0}, there exists a set of $\C^{\infty}$ partition of unity functions $x \mapsto p_j(x)$ of $\Theta$, $j=1,2,...,J$.

Then, for any $\delta' \in B_{\eta_3}(\delta)$, by the definition of partition of unity,

\begin{equation}\label{eq: intdcp}
\int_{{\MM_\Delta(\delta')}}f d\mathrm{Vol}=\sum_{1\leq i\leq m,1\leq
j\leq J}{\int_{\widetilde{\MM}_{\Delta}(\delta')\cap\theta_i}p_j(x)f(x) d\mathrm{Vol}}.
\end{equation}

The equation (\ref{eq: intdcp}) holds since $f(x)=0$ for all $x \notin \mathcal{X}$.

To show that $\int_{\MM_\Delta(\delta')}f d\mathrm{Vol}$ converges to $\int_{\MM_\Delta(\delta)}f d\mathrm{Vol}$ as $\delta'$ converges to $\delta$, it suffices to show that $\int_{\widetilde{\MM}_{\Delta}(\delta')\cap\theta_i}p_j(x)f(x) d\mathrm{Vol}$ converges to $\int_{\widetilde{\MM}_{\Delta}(\delta)\cap\theta_i}p_j(x)f(x) d\mathrm{Vol}$ as $\delta'$ converges to $\delta$, for all $i=1,2,...,m$ and $j=1,2,...,J$.

Without loss of generality, assume that $\widetilde{\MM}_\Delta(B_{\eta_3}(\delta))$ intersects $\theta_i$ at axis $\text{i}(\theta_i)=d_x$. Then, there exists constants $c>0$ and $C>0$ such that $\partial_{x_{d_x}} \Delta(x)>c$ and $\|\partial \Delta(x)\|<C$ for all $x\in \theta_i$, $i=1,2,...,m$.

We can apply the implicit function theorem to establish existence of the $\mathcal{C}^1$ function
$
g:X_{i1}\times...\times X_{i (d_x-1)}\times
B_\eta^+(\delta)\rightarrow X_{id_x},
$
such that $\Delta(x_1,...,x_{d_x-1},g(x_1,...,x_{d_x-1},\delta'))=\delta'$ for all
$(x_1,...,x_{d_x-1},\delta')\in X_{i1}\times ...\times X_{i
(d_x-1)}\times B_\eta(\delta)$. Define the one-to-one mapping $\psi_{d_x}$ as:
\begin{equation*}
\psi_{d_x}:X_{i1}\times ...\times X_{i(d_x-1)}\times
B_\eta^+(\delta)\rightarrow X_{i1}\times...\times
X_{i(d_x-1)}\times X_{i(d_x)},
\end{equation*}
where
$\psi_{d_x}(x_{-d_x},\delta')=(x_{-d_x},g(x_{-d_x},\delta'))$ for $x_{-d_x} := (x_1,x_2,...,x_{d_x-1})$. Note that $\psi_{d_x}$ and $g$ are both $\C^1$ functions.

For any $\delta'$ such that $|\delta'-\delta|<\eta_3$, by the change of variables we have:
\begin{equation}\label{eq: continuity}
\int_{\widetilde{\MM}_{\Delta}(\delta')\cap\theta_i}p_j(x)f(x) d\mathrm{Vol} = \int_{X_1\times X_2\times ...\times X_{d_x-1}} (p_j f)\circ \psi_{d_x}(x_{-d_x},\delta')\frac{\|\partial \Delta \circ \psi_{d_x}(x_{-d_x},\delta')\|}{|\partial_{x_{d_x}} \Delta \circ \psi_{d_x}(x_{-d_x},\delta')|}dx_{-d_x}.
\end{equation}

Since $|\partial_{x_{d_x}} \Delta \circ \psi_{d_x}(x_{-d_x},\delta')| = |\partial_{x_{d_x}} \Delta|_{x=\psi_{d_x}(x_1,...,x_{d_x-1},\delta')}>c$ for all $\delta'\in B_{\eta_3}(\delta)$ and $x_{-d_x}\in X_1\times X_2\times...\times X_{d_x-1}$ and $p_j$, $f$, $\partial \Delta$ and $\partial_{x_{d_x}} \Delta$ are uniformly continuous functions on $\widetilde{\MM}_{\Delta}(B_{\eta_3}(\delta))\cap B_i$, conclude that the map
\begin{quote}
$ (p_j f)\circ \psi_{d_x}\frac{\|\partial \Delta \circ \psi_{d_x}\|}{|\partial_{x_{d_x}} \Delta \circ \psi_{d_x}|}$ is uniformly continous on $X_1\times ...\times X_{d_x-1}\times B_{\eta_3}(\delta)$.
\end{quote}

Since $X_1\times...\times X_{d_x-1}$ and is bounded, it  immediately follows that
$\delta' \mapsto \int_{\widetilde{\MM}_{\Delta}(\delta')\cap\theta_i}p_j(x)f(x) d\mathrm{Vol}$ is continuous at $\delta'=\delta$, and hence $$ \delta' \mapsto \int_{\MM_\Delta(\delta')}f d\mathrm{Vol}=\sum_{1\leq i\leq m, 1\leq j\leq J}\int_{\widetilde{\MM}_{\Delta}(\delta')\cap\theta_i}p_j(x)f(x) d\mathrm{Vol}$$ is continuous at $\delta'=\delta$.

This argument applies to every $\delta \in \DD$, and by compactness
of $\DD$ the continuity claim extends to the entire $\DD$.
 \end{proof}

\begin{lemma}[Hadamard differentiability of $\Delta \mapsto F_{\Delta,\mu}$ and $\Delta \mapsto \Delta^*_{\mu}$ ] \label{lemma:had1}
 Suppose that ${\sf S}.1$-${\sf S}.2$ hold.
Then:

(a) The map  $F_{\Delta,\mu}(\delta):
\mathbb{F}\rightarrow \mathbb{R}$ is Hadamard-differentiable  uniformly in $\delta \in \DD$ at
$\Delta$ tangentially to $\mathbb{F}_0$, with the derivative
map  $\partial_\Delta F_{\Delta,\mu}(\delta): \mathbb{F}_0 \to \Bbb{R}$ defined by
$$
G \mapsto \partial_\Delta F_{\Delta,\mu}(\delta)[G] :=
-\int_{\MM_{\Delta}(\delta)}\frac{G(x) \mu'(x)}{\|\partial \Delta(x) \|}d\mathrm{Vol}.
$$

(b) The map $\Delta^*_{\mu}(u): \mathbb{F}\rightarrow \mathbb{R}$ is Hadamard-differentiable uniformly in $u \in \UU$ at
$\Delta$  tangentially to $\mathbb{F}_0$, with the
derivative map $\partial_\Delta \Delta^*_{\mu}(u):  \mathbb{F}_0 \to \Bbb{R}$ defined by:
$$
G \mapsto \partial_\Delta \Delta^*_{\mu}(u)[G] : =
-\frac{\partial_\Delta F_{\Delta,\mu}(\Delta^*_{\mu}(u))[G]}{
f_{\Delta,\mu}(\Delta^*_{\mu}(u))}.
$$

\end{lemma}

\begin{proof}[Proof of Lemma \ref{lemma:had1}] To shows statement (a), for any $G_n\rightarrow G\in \mathbb{F}_0$ under
sup-norm such that $\Delta+t_nG_n \in \mathbb{F}$, and $t_n\rightarrow 0$,
we consider
$$\frac{F_{\Delta+t_nG_n,\mu}(\delta)-F_{\Delta,\mu}(\delta)}{t_n}.$$
By assumption, any function $G\in \mathbb{F}_0$ is bounded  and uniformly continuous on
$B(\mathcal{X})$. Hence, $G_n$ is uniformly
bounded for $n\geq N$ , since $G_n\rightarrow G$ in sup-norm.

For any $\delta\in \mathcal{D}$ we consider a procedure similar to Lemma \ref{lemma:had0}. We use the same notation as in Step 1 of the proof of Lemma \ref{lemma:had0}. Suppose for $\eta_1 > 0$ small enough,
we have a rectangle cover $\Theta = \cup_{i=1}^m \theta_i\subseteq B(\mathcal{X})$ of
$\overline{\widetilde{\MM}_\Delta(B_{\eta_1}(\delta))\cap B_{\eta_1}(\overline\XX)}$ such that for all $\eta<\eta_1$,  $\widetilde{\MM}_\Delta(B_{\eta}(\delta))$ intersects each $\theta_i$ at some axis $\text{i}(\theta_i)$, $1\leq i\leq m$. As before, there is a partition of unity $\{p_j\}_{j=1}^J$ on the cover sets $\Theta=\{\theta_i\}_{i=1}^m$. As in the proof of Lemma \ref{lemma:had0}, we can rewrite
\begin{eqnarray*}
& &  \frac{\int_\mathcal{X} \left[1\{\Delta(x)+t_nG_n(x)\leq \delta\}-1\{\Delta(x)\leq \delta\}\right] \mu'(x)dx}{t_n}  \\
&& = \sum_{1\leq i\leq m,1\leq
j\leq J} \int_{\widetilde{\MM}_{\Delta}(B_{\eta}^+(\delta))\cap\theta_i}p_j(x)\frac{\left[1\{\Delta(x)+t_nG_n(x)\leq \delta\}-1\{\Delta(x)\leq \delta\}\right] \mu'(x)}{t_n}  dx.
\end{eqnarray*}

Then, for any fixed positive number $|\zeta|$, there exist $N$ large enough
such that $\sup_{x\in B(\mathcal{X}),n\geq N}
|G_n-G|<|\zeta|$. Moreover,  for any $x\in B(\XX)$, and large enough $n$,
\begin{equation*}
1\{\Delta(x)+t_nG_n(x)\leq \delta\}\leq
1\{\Delta(x)+t_n(G(x)-\zeta)\leq \delta\}.
\end{equation*}

As in  Step 2 of the proof of Lemma \ref{lemma:had0}, suppose $\theta_i=X_{i1} \times...\times X_{id_x}$ intersects
$\widetilde{\MM}_\Delta(B_{\eta}(\delta))$ at  $\text{i}(\theta_i) = x_{d_x}$.
Define the  parametrization $$\psi_{d_x}: X_{i1} \times...\times X_{i,d_x-1}\times B_\eta(\delta)\mapsto \theta_i, $$
$$\psi_{d_x}(x_{-d_x},\delta')=(x_{-d_x},g(x_{-d_x},\delta')),$$
where $g(x_{-d_x},\delta')$ is the implicit function derived
from equation $\Delta(x)=\delta'$, for any $\delta'\in B_\eta(\delta)$. Therefore, for large enough $n$,
\begin{eqnarray*}
& &  \int_{\widetilde{\MM}_{\Delta}(B_{\eta}^+(\delta))\cap\theta_i}p_j(x)\frac{\left[1\{\Delta(x)+t_nG_n(x)\leq \delta\}-1\{\Delta(x)\leq \delta\}\right] \mu'(x)}{t_n}  dx \\
&\leq&
\frac{\int_{\widetilde{\MM}_\Delta(B_{\eta}(\delta))\cap\theta_i} \left[1\{\Delta(x)+t_n(G(x)-\zeta)\leq \delta\}-1\{\Delta(x)\leq \delta\} \right] \mu'(x)dx}{t_n}. \notag \\
\end{eqnarray*}

Next, by a change of variables $\psi_{d_x}^{-1}$ from $\theta_i$ to $X_{i1}\times ...\times X_{i,d_x-1}\times B_\eta(\delta)$,
$$\int_{\widetilde{\MM}_\Delta(B_{\eta}(\delta)) \cap \theta_i}p_j(x) {\frac{1\{ \delta\leq \Delta(x)\leq \delta-t_n(G(x)-\zeta)\}\mu'(x)}{t_n}} dx$$
$$= \int_{X_{i1}\times ...\times X_{i,d_x-1}}\int_{B_\eta(\delta)}\frac{(p_j \cdot \mu')\circ \psi_{d_x}(x_{-d_x},\delta')}{|\partial_{x_{d_x}} \Delta \circ \psi_{d_x}(x_{-d_x},\delta')|} \frac{1\{ \delta\leq \delta' \leq \delta-t_n(G\circ \psi_{d_x}(x_{-d_x},\delta)-\zeta)\} }{t_n}d{\delta'}dx_{-d_x} $$
$$= \int_{X_{i1}\times ...\times X_{i,d_x-1}}\int_{B_\eta(\delta)\cap[\delta, \delta-t_n(G\circ \psi_{d_x}(x_{-d_x},\delta)-\zeta)]}\frac{(p_j \cdot \mu')\circ \psi_{d_x}(x_{-d_x},\delta')}{|\partial_{x_{d_x}} \Delta \circ \psi_{d_x}(x_{-d_x},\delta')|t_n} d{\delta'}dx_{-d_x} $$
$$ \leq - \int_{X_{i1} \times...\times X_{i,d_x-1}}\frac{(p_j \cdot \mu') \circ \psi_{d_x}(x_{-d_x},\delta)}{|\partial_{x_{d_x}} \Delta \circ \psi_{d_x}(x_{-d_x},\delta)|} (G\circ \psi_{d_x}(x_{-d_x},\delta)-\zeta) dx_{-d_x}+o(\eta)$$
\begin{eqnarray*}\label{eq: hadamard2}
&& =-\int_{\theta_i\cap \widetilde{M}_{\Delta}(\delta)}p_j(x)\mu'(x)\frac{G(x)-\zeta}{\parallel \partial
\Delta(x)\parallel}d\mathrm{Vol}+o(\eta)\\
&&  = -\int_{\theta_i\cap {M}_{\Delta}(\delta)}p_j(x)\mu'(x)\frac{G(x)-\zeta}{\parallel \partial
\Delta(x)\parallel}d\mathrm{Vol}+o(\eta),
\end{eqnarray*}
where the inequality in the above equation holds by continuity of $(x_{-d_x},\delta') \mapsto (p_j \cdot \mu') \circ \psi_{d_x}(x_{-d_x},\delta')/|\partial_{x_{d_x}} \Delta \circ \psi_{d_x}(x_{-d_x},\delta')| $. More specifically, fixing $\eta>0$ and $x_{-d_x}$, for $t_n \to 0$,
$$B_\eta(\delta)\cap[\delta, \delta-t_n(G\circ \psi_{d_x}(x_{-d_x},\delta)-\zeta)] = [\delta, \delta-t_n(G\circ \psi_{d_x}(x_{-d_x},\delta)-\zeta)] $$
and
$$\frac{(p_j \cdot \mu')\circ \psi_{d_x}(x_{-d_x},\delta')}{|\partial_{x_{d_x}} \Delta \circ \psi_{d_x}(x_{-d_x},\delta')|} \to \frac{(p_j \cdot \mu')\circ \psi_{d_x}(x_{-d_x},\delta)}{|\partial_{x_{d_x}} \Delta \circ \psi_{d_x}(x_{-d_x},\delta)|}$$
 as $\delta' \to \delta$.
 The last equality above holds because $\mu'(x)=0$ for all $x \in \widetilde{\MM}_{\Delta}(\delta) \setminus \MM_{\Delta}(\delta)$.

Since  $m$ and $J$ are fixed for any $n\geq N$, and $|G\circ \psi_{d_x}(x_{-d_x},\delta)-\zeta|$ is bounded by some absolute constant, $\sum_j p_j(x)=1$ and
$p_j(x)\geq 0$, we can let $\zeta \to 0$ to conclude that:
$$\lim_{n\rightarrow
\infty}\frac{F_{\Delta+t_nG_n,\mu}(\delta)-F_{\Delta,\mu}(\delta)}{t_n}\leq
\sum_{i=1}^m\sum_{j=1}^J -\int_{\theta_i\cap {M}_{\Delta}(\delta)}p_j(x)\mu'(x)\frac{G(x)}{\parallel \partial
\Delta(x)\parallel}d\mathrm{Vol}.$$
The right side is given by:
$$
-\int_{\MM_{\Delta}(\delta)}\frac{\mu'(x)G(x)}{\|\partial \Delta(x)\|}d\mathrm{Vol}.
$$

On the other hand, $$1(\Delta(x)+t_nG_n(x)\leq \delta)\geq 1(\Delta(x) + t_n(G(x)+\zeta) \leq \delta)$$ for some $\zeta>0$. So,
\begin{eqnarray*}
& &  \int_{\widetilde{\MM}_{\Delta}(B_{\eta}^+(\delta))\cap\theta_i}p_j(x)\frac{\left[1\{\Delta(x)+t_nG_n(x)\leq \delta\}-1\{\Delta(x)\leq \delta\}\right] \mu'(x)}{t_n}  dx \\
&\geq&
\frac{\int_{\widetilde{\MM}_\Delta(B_{\eta}(\delta))\cap\theta_i} \left[1\{\Delta(x)+t_n(G(x)+\zeta)\leq \delta\}-1\{\Delta(x)\leq \delta\} \right] \mu'(x)dx}{t_n}. \notag \\
\end{eqnarray*}
And,  by a change of variables $\psi_{d_x}^{-1}$ from $\theta_i$ to $X_{i1}\times ...\times X_{i,d_x-1}\times B_\eta(\delta)$,
$$\int_{\widetilde{\MM}_\Delta(B_{\eta}(\delta)) \cap \theta_i}p_j(x) {\frac{1\{ \delta\leq \Delta(x)\leq \delta-t_n(G(x)+\zeta)\}\mu'(x)}{t_n}} dx$$
$$= \int_{X_{i1}\times ...\times X_{i,d_x-1}}\int_{B_\eta(\delta)}\frac{(p_j \cdot \mu')\circ \psi_{d_x}(x_{-d_x},\delta')}{|\partial_{x_{d_x}} \Delta \circ \psi_{d_x}(x_{-d_x},\delta')|} \frac{1\{ \delta\leq \delta' \leq \delta-t_n(G\circ \psi_{d_x}(x_{-d_x},\delta)+\zeta)\} }{t_n}d{\delta'}dx_{-d_x} $$
$$= \int_{X_{i1}\times ...\times X_{i,d_x-1}}\int_{B_\eta(\delta)\cap[\delta, \delta-t_n(G\circ \psi_{d_x}(x_{-d_x},\delta)+\zeta)]}\frac{(p_j \cdot \mu')\circ \psi_{d_x}(x_{-d_x},\delta')}{|\partial_{x_{d_x}} \Delta \circ \psi_{d_x}(x_{-d_x},\delta')|t_n} d{\delta'}dx_{-d_x} $$
$$ \geq - \int_{X_{i1} \times...\times X_{i,d_x-1}}\frac{(p_j \cdot \mu') \circ \psi_{d_x}(x_{-d_x},\delta)}{|\partial_{x_{d_x}} \Delta \circ \psi_{d_x}(x_{-d_x},\delta)|} (G\circ \psi_{d_x}(x_{-d_x},\delta)+\zeta) dx_{-d_x} - o(\eta)$$
\begin{eqnarray*}
&& =-\int_{\theta_i\cap \widetilde{M}_{\Delta}(\delta)}p_j(x)\mu'(x)\frac{G(x)+\zeta}{\parallel \partial
\Delta(x)\parallel}d\mathrm{Vol}-o(\eta)\\
&&  = -\int_{\theta_i\cap {M}_{\Delta}(\delta)}p_j(x)\mu'(x)\frac{G(x)+\zeta}{\parallel \partial
\Delta(x)\parallel}d\mathrm{Vol}-o(\eta).
\end{eqnarray*}
Let $\zeta\rightarrow 0$ and $\eta\rightarrow 0$, it follows that
$$\lim_{n\rightarrow
\infty}\frac{F_{\Delta+t_nG_n,\mu}(\delta)-F_{\Delta,\mu}(\delta)}{t_n}\geq
-\int_{\MM_{\Delta}(\delta)}\frac{\mu'(x)G(x)}{\|\partial \Delta(x)\|}d\mathrm{Vol}.$$

Combining the two inequalities, we conclude that
$F_{\Delta,\mu}(\delta)$ is Hadamard-differentiable at $\Delta$
tangentially to $\mathbb{F}_0$ with derivative
$$
\partial_{\Delta} F_{\Delta,\mu}(\delta)[G] = -\int_{\MM_{\Delta}(\delta)} \frac{\mu'(x)G(x)}{\|\partial \Delta(x)\|}d\mathrm{Vol}.
$$

To show that the result holds uniformly in $\delta \in \DD$, we use the equivalence between uniform convergence and continuous convergence (e.g., \citen[p.2]{Resnick87}). Take a sequence $\delta_n$ in $\DD$ that converges to $\delta \in \DD$. Then, the preceding argument applies to this sequence and $\partial_{\Delta} F_{\Delta,\mu}(\delta_n)[G] \to  \partial_{\Delta} F_{\Delta,\mu}(\delta)[G]$ by uniform continuity of
$\delta \mapsto \partial_{\Delta} F_{\Delta,\mu}(\delta)[G]$ on $\DD$,
which holds by Lemma \ref{lemma:continuity}
because $G$, $\mu'$, and $\|\partial \Delta \|$ are continuous on $\overline \XX$ and  $\DD$ excludes neighborhoods of the critical values of $\Delta$ in $\overline \XX$. 


Tho show statement (b), note that  by statement (a),  Hadamard differentiability of the quantile map, see e.g., Lemma 3.9.20 in \citen{vdV-W},  and the chain rule for Hadamard differentiation, the inverse map $\Delta^*_{\mu}(u)$ is Hadamard differentiable at $\Delta$
tangentially to $\mathbb{F}_0$ with the derivative map
$$\partial_{\Delta} \Delta^*_{\mu}(u)[G]=  -\left. \frac{\partial_{\Delta} F_{\Delta,\mu}(\delta)[G]}{ \partial_{\delta} F_{\Delta,\mu}(\delta)}\right|_{\delta=\Delta^*_\mu(u)} = \frac{\partial_{\Delta} F_{\Delta,\mu}(\Delta^*_\mu(u))[G]}{ f_{\Delta,\mu}(\Delta^*_\mu(u))},$$
uniformly in the index $u \in \mathcal{U}= \{ u \in (0,1): \Delta^*_{\mu}(u)\in \mathcal{D},  \ f_{\Delta,\mu}(\Delta^*_\mu(u)) \geq \eps\}$. \end{proof}

\begin{proof}[Proof of Lemma \ref{lemma:had3} ]
To show Statement (a), Consider $t_n\rightarrow 0$ and $(G_n,H_n)\rightarrow (G,H)\in \mathbb{D}_0:=\mathbb{F}_0\times \mathbb{H}$
 as $n \to \infty$, such that $(\Delta+t_n G_n, \mu+t_nH_n) \in  \mathbb{D}$.
Let $\Delta_n:=\Delta+t_n G_n$ and $\mu_n:=\mu+t_nH_n$. Then, we can decompose
$$F_{\Delta_n,{\mu}_n}(\delta)-F_{\Delta,\mu}(\delta)=
[F_{\Delta_n,{\mu}_n}(\delta)-F_{\Delta_n,{\mu}}(\delta)]+[F_{\Delta_n,\mu}(\delta)-F_{\Delta,\mu}(\delta)].$$

By Lemma \ref{lemma:had1}, $$\frac{F_{\Delta_n,\mu}(\delta)-F_{\Delta,\mu}(\delta)}{t_n}=-\int_{\MM_{\Delta}(\delta)}
\frac{G(x)\mu'(x)}{\|\partial \Delta(x)\|}d\Vol+o(1).$$

Let $g_{\Delta,\delta} := 1(\Delta(x) \leq \delta)$. By definition of $F_{\Delta_n,{\mu}_n}(\delta)$,
$$\frac{F_{\Delta_n,\mu_n}(\delta)-F_{\Delta_n,\mu}(\delta)}{t_n}=H_{n}(g_{\Delta_n,\delta}).$$
Note that
$$H_{n}(g_{\Delta_n,\delta})-H(g_{\Delta,\delta})= [H_{n}(g_{\Delta_n,\delta})-H_{n}(g_{\Delta,\delta})] +[H_n-H](g_{\Delta,\delta}).$$
The second term goes to 0 by the assumption $H_n \to H$ in $\mathbb{H}$.
For the first term, we further decompose
$$|H_{n}(g_{\Delta_n,\delta})-H_{n}(g_{\Delta,\delta})|\leq |H_n(g_{\Delta_n,\delta})-H(g_{\Delta_n,\delta})|+
|H_n(g_{\Delta,\delta})-H(g_{\Delta,\delta})| + |H(g_{\Delta_n,\delta})-H(g_{\Delta,\delta)}|.$$
The first two terms go to 0 by  $\|H_n-H\|_{\mathcal{G}}\rightarrow 0$.
Moreover, $H(g_{\Delta_n,\delta}) \to H(g_{\Delta,\delta})$ because $g_{\Delta_n, \delta}(X)
= 1(\Delta_n(X) \leq \delta) \to
g_{\Delta, \delta}(X) = 1(\Delta(X) \leq \delta)$ in the $L^2 (\mu)$ norm, since $\Delta_n \to \Delta$ in the sup norm and $\Delta(X)$ has an absolutely continuous distribution, and since we require the operator $H$ to be continuous under the $L^2(\mu)$ norm.

We conclude that for any $\delta \in \mathcal{D}$,
$$\frac{F_{\Delta_n,\mu_n}(\delta)-F_{\Delta,\mu}(\delta)}{t_n}\to
-\int_{\MM_{\Delta}(\delta)} \frac{G(x)\mu'(x)}{\|\partial
\Delta(x)\|}d\mathrm{Vol}+H(g_{\Delta, \delta}) = \partial_{\Delta,\mu} F_{\Delta,\mu}(\delta)[G,H].$$
By an argument similar to the proof of Lemma \ref{lemma:had1}, it can be shown that the convergence is uniform in $\delta \in \DD$.

Statement (b) follows by statement (a) and the Hadamard differentiability of the quantile map uniformly in the quantile index, see, e.g., Lemma 3.9.20 in \citen{vdV-W}.\end{proof}

\subsection{Proof of Lemma \ref{lemma:Ratio-G}}

We will denote the functions in the classes $\mathcal{F}_M$ and $\mathcal{F}_I$ by $\varphi_t(x)$ whenever we want  to distinguish $x=(x_1,\ldots,x_{d_x})$, the argument of the function, from $t := (t_1, \ldots, t_{d_z})$, the index of the function in the class. Otherwise, we will use $\varphi(x)$.
To analyze $\Lambda^{-}_{\Delta, \mu,\delta}$ it is convenient to introduce  the operator  $\Upsilon_{\Delta, \mu,\delta}: \widetilde{\mathbb{D}} \to \mathbb{R}$ defined by
\begin{equation*}\label{eq:G_supp}
\Upsilon_{\Delta,\mu,\delta}(\varphi)  :=\int \varphi(x)1\{\Delta(x)\leq \delta\}d\mu(x),\end{equation*}
since $\Lambda^{-}_{\Delta, \mu,\delta}(\varphi) = \Upsilon_{\Delta,\mu,\delta}(\varphi) /\Upsilon_{\Delta,\mu,\delta}(1)$.

%
%


Let $\widetilde{\MM}_{\Delta}(B_{\eta}(\delta)):=\cup_{\delta'\in B_{\eta}(\delta)}\widetilde{\MM}_{\Delta}(\delta')$, where  $\widetilde{\MM}_{\Delta}(\delta):=\{x\in B(\XX) : \Delta(x)=\delta\}$ and  $B_{\eta}(\delta):=(\delta-\eta,\delta+\eta)$ for any $\delta\in \mathcal{V}$ and $\eta>0$. When $\varphi_t \in \mathcal{F}_I$ we make the following technical assumption to deal with the discontinuity of the indicator functions:

${\sf AS}.1$.
Define the set $\widetilde{\mathcal{Z}}_{k,\eta}(\delta,t_k):=\{x_{-k} : (x_k,x_{-k})\in\widetilde{\MM}_\Delta(B_\eta(\delta)), x_{k}=t_k \}$ for any $\eta>0$, $\delta\in \mathcal{V}$, $k=1,2,...,d_x$, and  $t_k \in \mathbb{R}$. Then, for any $\epsilon>0$, there exist $\eta_0>0$ such that for any $\eta<\eta_0$, $\int_{\widetilde{\mathcal{Z}}_{k,\eta}(\delta,t_k)}  d \mu(x_{-k}) \leq \epsilon$ holds uniformly over all $\delta\in \mathcal{V}$, $t_k \in \mathbb{R}$ and $k=1,2,...,d_x$.

The next result shows that $(\Delta, \mu,\delta) \mapsto \Upsilon_{\Delta, \mu,\delta}$ is Hadamard differentiable.

\begin{lemma}[Hadamard differentiability of $(\Delta, \mu,\delta) \mapsto \Upsilon_{\Delta, \mu,\delta}$]\label{lemma:had4}
Assume that Assumptions S.1 and S.2 hold and $\delta \in \DD$. Then,

(a) The map $\Upsilon_{\Delta,\mu,\delta}(\varphi): \widetilde{\mathbb{D}} \to \mathbb{R}$ is Hadamard-differentiable uniformly  in
$\varphi \in \mathcal{F}_M$ at $(\Delta,\mu,\delta)$ tangentially to $\widetilde{\mathbb{D}}_0$.

(b) If in addition Assumption ${\sf AS}.1$ holds, the map $\Upsilon_{\Delta,\mu,\delta}(\varphi): \widetilde{\mathbb{D}} \to \mathbb{R}$ is Hadamard-differentiable uniformly  in
$\varphi \in \mathcal{F}_I$ at $(\Delta,\mu,\delta)$ tangentially to $\widetilde{\mathbb{D}}_0$.

(c) The derivative map $\partial_{\Delta,\mu,\delta} \Upsilon_{\Delta,\mu,\delta}(\varphi): \widetilde{\mathbb{D}} \to \mathbb{R}$ is defined by:
%
%
$$(G,H,K)\mapsto \partial_{\Delta,\mu,\delta}\Upsilon_{\Delta,\mu,\delta}(\varphi)[G,H,K]:= \int_{\mathcal{M}_\Delta(\delta)}\varphi(x)\frac{K-G(x)}{\|\partial \Delta(x)\|}d\Vol +H(h_{\Delta,\delta,\varphi}),$$
where $h_{\Delta,\delta,\varphi}:=\varphi(x)1\{\Delta(x)\leq \delta\}$.
\end{lemma}


%
%
%

\begin{proof}[Proof of Lemma \ref{lemma:had4}]

Statements (a) and (b) follow by similar arguments. For brevity, we focus on the proof of Statement (b) and mention the changes needed for the proof of Statement (a), which is simpler. 


To show Statement (b), consider  $s_n \to 0$ and $(G_n,H_n,K_n) \to (G,H,K) \in \widetilde{\mathbb{D}}_0$ as $n\to \infty$, such that $(\Delta + s_n G_n, \mu + s_n H_n, \delta + s_n K_n) \in \widetilde{\mathbb{D}}$. Let $\Delta_n := \Delta + s_n G_n$, $\mu_n := \mu + s_n H_n$, and $\delta_n := \delta + s_n K_n$. Then, we can decompose
\begin{equation}\label{eq:gamma}
\Upsilon_{\Delta_n, \mu_n, \delta_n}(\varphi) -  \Upsilon_{\Delta, \mu, \delta}(\varphi) = [\Upsilon_{\Delta_n, \mu_n, \delta_n}(\varphi) -  \Upsilon_{\Delta_n, \mu, \delta_n}(\varphi)] + [\Upsilon_{\Delta_n, \mu, \delta_n}(\varphi) -  \Upsilon_{\Delta, \mu, \delta}(\varphi)].
\end{equation}

%

The first term of (\ref{eq:gamma}) satisfies
$$
\frac{\Upsilon_{\Delta_n, \mu_n, \delta_n}(\varphi) -  \Upsilon_{\Delta_n, \mu, \delta_n}(\varphi)}{s_n} = H_n(h_{\Delta_n,\delta_n,\varphi}) = H(h_{\Delta,\delta,\varphi}) + o(1).
$$
The first equality follows from  linearity of $\mu \mapsto \Upsilon_{\Delta_n, \mu, \delta_n}(\varphi)$ and $h_{\Delta_n,\delta_n,\varphi} = \varphi(x) 1\{\Delta_n(x) \leq \delta_n \}$.  To show the second equality note that
$$
H_n(h_{\Delta_n,\delta_n,\varphi}) - H(h_{\Delta,\delta,\varphi}) = H_n(h_{\Delta_n,\delta_n,\varphi}) - H_n(h_{\Delta,\delta,\varphi}) +  [H_n - H](h_{\Delta,\delta,\varphi}),
$$
where the second term goes to zero by the assumption $H_n \to H$ in $\widetilde{\mathbb{H}}$.
For the first term, we further decompose
\begin{multline*}
|H_{n}(h_{\Delta_n,\delta_n,\varphi})-H_{n}(h_{\Delta,\delta,\varphi})|\leq |H_n(h_{\Delta_n,\delta_n,\varphi})-H(h_{\Delta_n,\delta_n,\varphi})| \\ +
|H_n(h_{\Delta,\delta,\varphi})-H(h_{\Delta,\delta,\varphi})| + |H(h_{\Delta_n,\delta_n,\varphi})-H(h_{\Delta,\delta,\varphi})|.
\end{multline*}
By definition of the space $\widetilde{\mathbb{H}}$, the first two terms go to 0 by  $\|H_n-H\|_{\widetilde{\mathcal{G}}}\rightarrow 0$.


Moreover, $H(h_{\Delta_n,\delta_n,\varphi}) \to H(h_{\Delta,\delta,\varphi})$ because
$$h_{\Delta_n, \delta_n,\varphi}(X)
= \varphi(X)1(\Delta_n(X) \leq \delta_n) \to
h_{\Delta, \delta, \varphi}(X) = \varphi(X) 1(\Delta(X) \leq \delta)$$ in the $L^2 (\mu)$ norm, since $\Delta_n \to \Delta$ in the sup norm and $\Delta(X)$ has an absolutely continuous distribution, and since we require the operator $H$ to be continuous under the $L^2(\mu)$ norm.


Next we show that the second term of (\ref{eq:gamma}) satisfies
$$
\frac{\Upsilon_{\Delta_n, \mu, \delta_n}(\varphi) -  \Upsilon_{\Delta, \mu, \delta}(\varphi)}{s_n} = \int_{\mathcal{M}_\Delta(\delta)} \varphi(x)\frac{K-G(x)}{\|\partial \Delta(x)\|}\mu'(x)d\Vol + o(1).
$$
The proof follows the same steps as the proof of Lemma \ref{lemma:had1} after noticing that we can write
$$
\Upsilon_{\Delta_n, \mu, \delta_n}(\varphi) = \Upsilon_{\widetilde \Delta_n, \mu, \delta}(\varphi),
$$
where $\widetilde \Delta_n  = \Delta + s_n \widetilde G_n$ with $\widetilde G_n = G_n -K_n$, and replacing $\mu'(x)$ by $\widetilde \mu'(x) = \varphi(x) \mu'(x)$.


Specifically, following the notation in the proof of Lemma  \ref{lemma:had0},
$$
\Upsilon_{\widetilde \Delta_n, \mu, \delta}(\varphi)
=\sum_{i=1}^m\sum_{j=1}^J\int_{\widetilde{\MM}_\Delta(B_{\eta}(\delta)) \cap \theta_i}p_j(x)\varphi(x) {\frac{1\{ \delta\leq \Delta(x)\leq \delta-s_n \widetilde G_n(x) \}}{s_n}} \widetilde \mu'(x) dx.$$

Without loss of generality, assume that  $\theta_i$  intersects with $\widetilde{\MM}_\Delta(B_{\eta}(\delta))$ at the axis $x_{k_i} = x_{d_x}$.  When $\varphi(x) \in \mathcal{F}_I$, each component in the above summation satisfies:
\begin{multline*}\label{eq:partition}
\int_{\widetilde{\MM}_\Delta(B_{\eta}(\delta)) \cap \theta_i}p_j(x)\varphi(x) {\frac{1\{ \delta\leq \Delta(x)\leq \delta-s_n \widetilde G_n(x)\} \widetilde \mu'(x) }{s_n}} dx
\\=\int_{X_{i1}\times...\times X_{i,d_x-1}}{\int_{B_\eta(\delta)} \frac{(p_j\cdot \varphi \cdot \widetilde \mu') \circ \psi_{d_x}(x_{-d_x}, \delta')}{|\partial_{x_{d_x}} \Delta \circ \psi_{d_x}(x_{-d_x}, \delta')|} \times \frac{1\{ \delta\leq \delta' \leq \delta-s_n \widetilde G_n \circ \psi_{d_x}(x_{-d_x}, \delta')\}}{s_n}}d{\delta'}dx_{-d_x}
\\ =\int_{\widetilde{\mathcal{X}}^c_{d_x,\eta}(\delta,t_{d_x}) }{\int_{B_\eta(\delta)} \frac{(p_j \cdot \widetilde \mu') \circ \psi_{d_x}(x_{-d_x}, \delta')}{|\partial_{x_{d_x}} \Delta \circ \psi_{d_x}(x_{-d_x}, \delta')|} \times \frac{1\{ \delta\leq \delta' \leq \delta-s_n \widetilde G_n \circ \psi_{d_x}(x_{-d_x}, \delta')\}}{s_n}}d{\delta'}dx_{-d_x} \\ +\int_{\widetilde{\mathcal{X}}_{d_x,\eta}(\delta,t_{d_x})\}}{\int_{B_\eta(\delta)} \frac{(p_j \cdot \widetilde \mu') \circ \psi_{d_x}(x_{-d_x}, \delta')}{|\partial_{x_{d_x}} \Delta \circ \psi_{d_x}(x_{-d_x}, \delta')|} \times \frac{1\{ \delta\leq \delta' \leq \delta-s_n \widetilde G_n \circ \psi_{d_x}(x_{-d_x}, \delta')\}}{s_n}}d{\delta'}dx_{-d_x},
\end{multline*}
where $\widetilde{\mathcal{X}}_{d_x,\eta}(\delta,t_{d_x}) := [X_{i1}\times...\times X_{i,d_x-1}] \cap \widetilde{\mathcal{Z}}_{d_x,\eta}(\delta,t_{d_x})$ and $\widetilde{\mathcal{X}}^c_{d_x,\eta}(\delta,t_{d_x}) := X_{i1}\times...\times X_{i,d_x-1}\setminus \widetilde{\mathcal{X}}_{d_x,\eta}(\delta,t_{d_x})$. When $\varphi(x)\in \mathcal{F}_M$, then we could simply let $\widetilde{\mathcal{X}}_{d_x,\eta}(\delta,t_{d_x})=\emptyset$ in the rest of the proof.





Partition $t = (t_x,t_y)$ corresponding to $Z=(X,Y)$. Although $x \mapsto \varphi(x)=1(x \leq t_{x}) \mu(t_y \mid x)$ is a discontinuous function, $\delta \mapsto \varphi(x)\circ \psi_{d_x}(x_{-d_x},\delta)$ is continuous for those $x$ such that $x_{-d_x} \in\widetilde{\mathcal{X}}^c_{d_x,\eta}(\delta,t_{d_x})$ and $\delta = \Delta(x)$.
Accordingly, we partition the integral in two regions because the integrand is not necessarily continuous on $\widetilde{\X}_{d_x,\eta}(\delta,t_{d_x})\times B_\eta(\delta)$. We use Assumption {\sf AS}.1 to bound the integral in this region.
Thus, for any $\epsilon>0$, for $\eta$ being small enough, the area of $\widetilde{\mathcal{X}}_{d_x,\eta}(\delta,t_{d_x})$, defined as $\int_{\widetilde{\mathcal{X}}_{d_x,\eta}(\delta,t_{d_x})}\mu'(x_{-d_x}) d x_{-d_x}$, is less than or equal to $\epsilon$ by {\sf AS}.1 uniformly over $\delta$ and $t_{d_x}$. Then,  for large enough $n$, $k=1,2,...,d_x$ and some arbitrarily small $\zeta>0$, by continuity of the integrand,
\begin{multline*}
\int_{ \widetilde{\mathcal{X}}^c_{d_x,\eta}(\delta,t_{d_x})}{\int_{B_\eta(\delta)} \frac{(p_j\cdot  \widetilde \mu') \circ \psi_{d_x}(x_{-d_x}, \delta')}{|\partial_{x_{d_x}} \Delta \circ \psi_{d_x}(x_{-d_x}, \delta')|} \times \frac{1\{ \delta\leq \delta' \leq \delta-s_n \widetilde G_n \circ \psi_{d_x}(x_{-d_x}, \delta')\}}{s_n}}d{\delta'}dx_{-d_x} \\
 \leq - \int_{ \widetilde{\mathcal{X}}^c_{d_x,\eta}(\delta,t_{d_x})}\frac{(p_j \cdot \widetilde \mu') \circ \psi_{d_x}(x_{-d_x},\delta)}{|\partial_{x_{d_x}} \Delta \circ \psi_{d_x}(x_{-d_x},\delta)|} (\widetilde G\circ \psi_{d_x}(x_{-d_x},\delta)-\zeta) dx_{-d_x}+o(\eta)\\
 =- \int_{ X_{i1} \times...\times X_{i,d_x-1} }\frac{(p_j \cdot \widetilde \mu') \circ \psi_{d_x}(x_{-d _x},\delta)}{|\partial_{x_{d_x}} \Delta \circ \psi_{d_x}(x_{-d_x},\delta)|} (\widetilde G\circ \psi_{d_x}(x_{-d_x},\delta)-\zeta) dx_{-d_x} \\ +\int_{ \widetilde{\mathcal{X}}_{d_x,\eta}(\delta,t_{d_x})}\frac{(p_j \cdot \widetilde \mu') \circ \psi_{d_x}(x_{-d_x},\delta)}{|\partial_{x_{d_x}} \Delta \circ \psi_{d_x}(x_{-d_x},\delta)|} (\widetilde G\circ \psi_{d_x}(x_{-d_x},\delta)-\zeta) dx_{-d_x}+o(\eta),
\end{multline*}
where $\widetilde G = G - K$.

The inequality above holds by continuity of the integrand $(x_{-d_x},\delta') \mapsto (p_j \cdot \widetilde \mu') \circ \psi_{d_x}(x_{-d_x},\delta')/|\partial_{x_{d_x}} \Delta \circ \psi_{d_x}(x_{-d_x},\delta')|$ on $\widetilde{\mathcal{X}}^c_{d_x,\eta}(\delta,t_{d_x}) \times B_\eta(\delta)$, and
\begin{equation*}
\left|\int_{\widetilde{\mathcal{X}}_{d_x,\eta}(\delta,t_{d_x})}\frac{(p_j \cdot \widetilde \mu') \circ \psi_{d_x}(x_{-d_x},\delta)}{|\partial_{x_{d_x}} \Delta \circ \psi_{d_x}(x_{-d_x},\delta)|} (\widetilde G\circ \psi_{d_x}(x_{-d_x},\delta)-\zeta) dx_{-d_x}\right| \leq \int_{\widetilde{\mathcal{X}}_{d_x,\eta}(\delta,t_{d_x})}C dx_{-d_x} \leq  C\epsilon,
\end{equation*}
for
$$C:=\sup_{x_{-d_x}\in X_{i1} \times...\times X_{i,d_x-1}} \left|\frac{(p_j \cdot \widetilde \mu) \circ \psi_{d_x}(x_{-d_x},\delta)}{|\partial_{x_{d_x}} \Delta \circ \psi_{d_x}(x_{-d_x},\delta)|} (\widetilde G\circ \psi_{d_x}(x_{-d_x},\delta)-\zeta)\right|,$$ which is bounded from above, because all components in $C$ are bounded from above and $|\partial_{x_{d_x}} \Delta \circ \psi_{d_x}(x_{-d_x},\delta)|$ is bounded away from zero. Similarly, for $s_n$ large enough,
$$\left|\int_{\widetilde{\mathcal{X}}_{d_x,\eta}(\delta,t_{d_x})}{\int_{B_\eta(\delta)} (p_j \cdot \psi_{d_x}) \times \frac{(1\{ \delta\leq \Delta \leq \delta-s_n \widetilde G_n(x)\}\widetilde \mu') \circ \psi_{d_x}}{s_n|\partial_{x_{d_x}} \Delta \circ \psi_{d_x}|}}d{\delta'}dx_{-d_x} \right|\leq C\epsilon.$$

Therefore, combining the previous results
\begin{multline*}
\int_{\widetilde{\MM}_\Delta(B_{\eta}(\delta)) \cap \theta_i}p_j(x) \varphi(x){\frac{1\{ \delta\leq \Delta(x)\leq \delta-s_n \widetilde G_n(x)\}}{s_n}} \widetilde \mu'(x) dx \\
\leq -\int_{ \widetilde{\mathcal{X}}^c_{d_x,\eta}(\delta,t_{d_x})}\frac{(p_j \cdot \widetilde \mu') \circ \psi_{d_x}(x_{-d_x},\delta)}{|\partial_{x_{d_x}} \Delta \circ \psi_{d_x}(x_{-d_x},\delta)|} (\widetilde G\circ \psi_{d_x}(x_{-d_x},\delta)-\zeta) dx_{-d_x}+o(\eta) + C\epsilon \\
= -\int_{X_{i1} \times...\times X_{i,d_x-1} }\frac{(p_j\cdot  \widetilde \mu') \circ \psi_{d_x}(x_{-d_x},\delta)}{|\partial_{x_{d_x}} \Delta \circ \psi_{d_x}(x_{-d_x},\delta)|} (\widetilde G\circ \psi_{d_x}(x_{-d_x},\delta)-\zeta) dx_{-d_x} \\ + \int_{\widetilde{\mathcal{X}}_{d_x,\eta}(\delta,t_{d_x})}\frac{(p_j\cdot  \widetilde \mu') \circ \psi_{d_x}(x_{-d_x},\delta)}{|\partial_{x_{d_x}} \Delta \circ \psi_{d_x}(x_{-d_x},\delta)|} (\widetilde G\circ \psi_{d_x}(x_{-d_x},\delta)-\zeta) dx_{-d_x} + o(\eta)+C\epsilon\\
\leq - \int_{X_{i1} \times...\times X_{i,d_x-1}}\frac{(p_j \cdot  \widetilde \mu) \circ \psi_{d_x}(x_{-d_x},\delta)}{|\partial_{x_{d_x}} \Delta \circ \psi_{d_x}(x_{-d_x},\delta)|} (\widetilde G\circ \psi_{d_x}(x_{-d_x},\delta)-\zeta) dx_{-d_x}+o(\eta) +2C\epsilon \\
= -\int_{\theta_i \cap \widetilde{\MM}_{\Delta}(\delta)}p_j(x)\varphi(x)\cdot \widetilde \mu'(x) \frac{\widetilde G(x)-\zeta}{\| \partial
\Delta(x)\|}d\mathrm{Vol}+o(\eta) +2C\epsilon \\
= -\int_{\theta_i \cap \MM_{\Delta}(\delta)}p_j(x)\varphi(x)\cdot \widetilde \mu'(x) \frac{\widetilde G(x)-\zeta}{\| \partial
\Delta(x)\|}d\mathrm{Vol} +o(\eta) +2C\epsilon,
\end{multline*}
where $\zeta$, $\eta$ and $\epsilon$ can be arbitrarily small for large enough $n$.

Similarly, we can show that
\begin{multline*}
\int_{\widetilde{\MM}_\Delta(B_{\eta}(\delta)) \cap \theta_i}p_j(x) \varphi(x){\frac{1\{ \delta\leq \Delta(x)\leq \delta-s_n \widetilde G_n(x)\}}{s_n}} \widetilde \mu'(x) dx\\
\geq -\int_{\theta_i \cap \MM_{\Delta}(\delta)}p_j(x)\varphi(x)\cdot \widetilde \mu'(x) \frac{\widetilde G(x)-\zeta}{\| \partial
\Delta(x)\|}d\mathrm{Vol} -o(\eta) -2C\epsilon
\end{multline*}
Since we can choose $\eta$ and $\epsilon$ to be arbitrarily small, we conclude that for any $\varphi \in \mathcal{F}_I$,
$$
\frac{\Upsilon_{\Delta_n, \mu_n, \delta_n}(\varphi) -  \Upsilon_{\Delta, \mu, \delta}(\varphi)}{s_n} \to \int_{\mathcal{M}_\Delta(\delta)} \varphi(x)\frac{K-G(x)}{\|\partial \Delta(x)\|}\mu'(x)d\Vol  + H(\varphi(x) 1\{\Delta(x) \leq \delta\}).
$$

To show that the result holds uniformly in  $\varphi \in \mathcal{F}_I$, we use the equivalence between uniform convergence and continuous convergence (e.g., \citen[p.2]{Resnick87}). Take a sequence $\varphi^n \in \mathcal{F}_I$ that converges to $\varphi \in  \mathcal{F}_I$ in the $L^1(\mu)$ norm, i.e., $\int_{\mathcal{X}} |\varphi^n -\varphi | d\mu\rightarrow 0$ as $n \to \infty$. Then, the preceding argument applies to this sequence and $\partial_{\Delta,\mu,\delta} \Upsilon_{\Delta,\mu,\delta}(\varphi^n)[K,G,H] \to  \partial_{\Delta,\mu,\delta} \Upsilon_{\Delta,\mu,\delta}(\varphi)[K,G,H]$ by linearity of the map
$\varphi \mapsto \partial_{\Delta,\mu,\delta} \Upsilon_{\Delta,\mu,\delta}(\varphi)[K,G]$.

\end{proof}

\begin{proof}[Proof of Lemma \ref{lemma:Ratio-G}]

Note that $\Lambda^-_{\Delta,\mu,\delta}(\varphi) = \Upsilon_{\Delta,\mu,\delta}(\varphi)/ \Upsilon_{\Delta,\mu,\delta}(1)$, where $\Upsilon_{\Delta,\mu,\delta}(1) = \int 1(\Delta(x)\leq \delta) d\mu(x) = F_{\Delta,\mu}(\delta)$.

By Lemma \ref{lemma:had4}, $\Upsilon_{\Delta,\mu,\delta}(\varphi)$ and $\Upsilon_{\Delta,\mu,\delta}(1)$ are Hadamard-differentiable at $(\Delta,\mu,\delta)$ tangentially to $\widetilde{\mathbb{D}}_0$. Then, by the chain rule for Hadamard-differentiable mappings, $\Lambda^-_{\Delta,\mu,\delta}(\varphi)$ is Hadamard-differentiable at $(\Delta,\mu,\delta)$ tangentially to $\widetilde{\mathbb{D}}_0$ since $\Upsilon_{\Delta,\mu,\delta}(1) > 0$. The derivative map is obtained from
$$
\partial_{\Delta,\mu,\delta} \Lambda^-_{\Delta,\mu,\delta}(\varphi)  = \frac{\partial_{\Delta,\mu,\delta} \Upsilon_{\Delta,\mu,\delta}(\varphi)}{F_{\Delta,\mu}(\delta)} - \Lambda_{\Delta,\mu,\delta}(\varphi) \frac{\partial_{\Delta,\mu,\delta} \Upsilon_{\Delta,\mu,\delta}(1)}{F_{\Delta,\mu}(\delta)},
$$
after replacing the expressions of $\partial_{\Delta,\mu,\delta} \Upsilon_{\Delta,\mu,\delta}(\varphi)$ and $\partial_{\Delta,\mu,\delta} \Upsilon_{\Delta,\mu,\delta}(1)$ from Lemma \ref{lemma:had4} and grouping terms.

%

\end{proof}

\section{Sufficient Conditions for $\mu$-Donsker Properties in Section \ref{sec:theory2}}\label{app:C}

\begin{lemma}[Sufficient conditions for $\mathcal{G}$ being $\mu$-Donsker]\label{cor:donsker}
Suppose  ${\sf S}.1$-${\sf S}.2$ hold, and $\VV$ is the union of a finite number of compact intervals. Suppose that $\mathcal{F}$ satisfies:
$$\sup_{\widetilde{\Delta}\in \mathcal{F}} \sup_{x\in B(\XX)}\|\partial \widetilde{\Delta}(x)-\partial \Delta(x)\| + \sup_{\widetilde{\Delta} \in \mathcal{F}} \sup_{x\in B(\XX)} |\widetilde{\Delta}(x)- \Delta(x) |<c_0.$$
Let $N(\epsilon,\mathcal{F},\|\cdot\|_{\infty})$ be the $\epsilon$-covering number of the class $\mathcal{F}$ under $L_\infty$ norm. Suppose that $\int_{0}^1\sqrt{\log N(\epsilon^2,\mathcal{F},\|\cdot\|_{\infty}) } d\epsilon<\infty$. If $c_0$ is small enough, then $\mathcal{G}$ is $\mu$-Donsker.
\end{lemma}

\begin{proof}[Proof of Lemma \ref{cor:donsker}]
Since $\VV$ is a union of finite number of closed intervals, for any $\zeta>0$, we can construct a collection of closed intervals $\mathcal{I}:=\{[a_i,b_i]:i=1,2,...,r\}$ such that: (1) $|b_i-a_i|<\zeta$, (2)  $[a_i,b_i]\subset \VV$, (3) $\cup_{i=1}^r[a_i,b_i]=\VV$,
(4) $a_i\leq b_i\leq a_{i+1}\leq b_{i+1}$, for all $i=1,2,...,r-1$, and (5) $r\leq \frac{C_0}{\zeta}$, where $C_0$ is a constant.

Using ${\sf S}.1$ and ${\sf S}.2$ and the assumptions of the Lemma, there
exists $\eta>0$ small enough such that the following conditions hold:

 (1) There exist  constants $c$ and $C$ such that $\|\partial \Delta(x)\|\leq C$ for all $x\in \overline{\XX}$ and $\|\partial \Delta(x)\|\geq c$ in $\widetilde{\MM}_{\Delta}(B_{\eta}(\delta))$ for some small $\eta>0$ and all $\delta\in \DD$.

(2) Uniformly in $\widetilde{\Delta} \in \mathcal{F}$, $$\frac{c}{2}\leq \inf_{x\in \widetilde{\MM}_{\Delta}(B_{\eta}(\delta))}\|\partial \widetilde{\Delta}(x)\|\leq \sup_{x\in\widetilde{\MM}_{\Delta}(B_{\eta}(\delta))}\|\partial \widetilde{\Delta}(x)\|\leq \frac{c}{2}+C.$$

Moreover, using arguments similar to those used to show Lemma \ref{lemma:had0}, we can verify that:

(3) Uniformly in $\widetilde{\Delta} \in \mathcal{F}$, uniformly in $\delta\in \VV$,
$$f_{\widetilde{\Delta},\mu}(\delta)=\int_{\MM_{\widetilde{\Delta}}(\delta)}\frac{\mu'(x)}{\|\partial \widetilde{\Delta}(x)\|}d\mathrm{Vol}< K_1,$$
for some finite constant $K_1$.

Define the norm $\|g\|_{2,\mu}^2:=\int_{\XX} g(x)^2\mu'(x)dx$. For $\eta>0$ small enough, for any $\delta\in \VV$ and $\widetilde{\Delta} \in \mathcal{F}$,  $$\|1(\widetilde{\Delta} \leq \delta)-1(\widetilde{\Delta}\leq \delta+\eta)\|_{2,\mu}^2= \int 1(\delta\leq \widetilde{\Delta}(x)\leq \delta+\eta)\mu'(x)dx=\int_{\delta'\in B_{\eta}^+(\delta)} f_{\widetilde{\Delta},\mu}(\delta')d\delta'\leq K_1\eta.$$
Similarly, $\|1(\widetilde{\Delta} \leq \delta)-1(\widetilde{\Delta} \leq \delta-\eta)\|_{2,\mu}^2\leq K_1\eta$.

Let $B_{\zeta,\infty}(\Delta_1),...,B_{\zeta,\infty}(\Delta_{q_\zeta})$ be a set of $\zeta$-balls centered at $\Delta_1,...,\Delta_{q_\zeta}$ under sup norm that covers $\mathcal{F}$, where $q_\zeta=N(\zeta,\mathcal{F},\|\cdot\|_{\infty})$. Then, $[\Delta_{j}-\zeta, \Delta_{j}+\zeta]$ are covering brackets of $\mathcal{F}$, $j=1,2,...,q_{\zeta}$. For any $\widetilde{\Delta}\in [\Delta_{j}-\zeta, \Delta_{j}+\zeta]$ and $\delta\in [a_i,b_i]$, $i=1,2,...,r$, then the bracket $[1(\Delta_{j}+\zeta\leq a_i),1(\Delta_{j}-\zeta\leq b_i)]$ covers $1(\widetilde{\Delta}\leq \delta)$.
For $\zeta$ small enough, the size of the bracket $[1(\Delta_{j}+\zeta\leq a_i),1(\Delta_{j}-\zeta\leq b_i)]$ under the norm $\|\cdot\|_{2,\mu}$ is:
\begin{multline*}
\|1(\Delta_j+\zeta\leq a_i)-1(\Delta_j-\zeta\leq b_i )\|_{2,\mu}^2= \|1(\Delta_j\leq b_i+\zeta)-1(\Delta_j\leq a_i-\zeta)\|_{2,\mu}^2\leq  3K_1\zeta,
\end{multline*}
since $|b_i - a_i| < \zeta$ by construction.
Therefore, for $\zeta$ small enough, $\{[1(\Delta_{j}+\zeta\leq a_i),1(\Delta_{j}-\zeta\leq b_i)]: j=1,2,...,q_\zeta, i=1,2,...,r\}$, form a set of $\sqrt{3K_1\zeta}$-brackets under the norm $\|\cdot\|_{2,\mu}$ that covers $\mathcal{G}$. The total number of brackets is $rq_\zeta\leq \frac{C_0}{\zeta}N(\zeta, \mathcal{F}, \|\cdot\|_{\infty})$. Or equivalently, for $\zeta$ small enough,
$$N_{[]}(\zeta,\mathcal{G},\|\cdot\|_{2,\mu})\leq \frac{3K_1C_0}{\zeta^2}N(\zeta^2/(3K_1),\mathcal{F},\|\cdot\|_{\infty}).$$
Then by assumption,
\begin{multline*}
\int_{0}^1\sqrt{\log(N_{[]}(\zeta,\mathcal{G},\|\cdot\|_{2,\mu}))}d\zeta\leq \int_{0}^1\sqrt{\log\left(\frac{3K_1C_0}{\zeta^2}N(\zeta^2/(3K_1),\mathcal{F},\|\cdot\|_{\infty})\right)}d\zeta\\ \lesssim \int_{0}^1\sqrt{\log\left(\frac{3K_1C_0}{\zeta^2}\right)}d\zeta+\int_{0}^{1}\sqrt{\log(N(\zeta^2/(3K_1),\mathcal{F},\|\cdot\|_{\infty}))}d\zeta <\infty.
\end{multline*}

We conclude that $\mathcal{G}$ is $\mu$-Donsker by Donsker theorem \cite[Theorem 19.5]{vdV}.
\end{proof}

\begin{lemma}[Sufficient conditions for $\widetilde{\mathcal{G}}$ being $\mu$-Donsker]\label{lemma:Donsker_supp}
Suppose  ${\sf S}.1$-${\sf S}.2$ hold, and $\mathcal{V}$ is the union of a finite number of compact intervals. Suppose that $\mathcal{F}$ satisfies:
$$\sup_{\widetilde{\Delta}\in \mathcal{F}} \sup_{x\in B(\XX)}\|\partial \widetilde{\Delta}(x)-\partial \Delta(x)\| + \sup_{\widetilde{\Delta} \in \mathcal{F}} \sup_{x\in B(\XX)} |\widetilde{\Delta}(x)- \Delta(x) |<c_0.$$
Let $N(\epsilon,\mathcal{F},\|\cdot\|_{\infty})$ be the $\epsilon$-covering number of the class $\mathcal{F}$ under $L_\infty$ norm. Suppose that $\int_{0}^1\sqrt{\log N(\epsilon^2,\mathcal{F},\|\cdot\|_{\infty}) } d\epsilon<\infty$. If $c_0$ is small enough, then $\widetilde{\mathcal{G}}$ is $\mu$-Donsker.
\end{lemma}

\begin{proof}[Proof of Lemma \ref{lemma:Donsker_supp}]
First, $\mathcal{F}_I$ and $\mathcal{F}_M$ are both $\mu$-Donsker. By Lemma \ref{cor:donsker}, the class $\mathcal{F}$ is $\mu$-Donsker. Since the class of the product of two functions from Donsker classes is Donsker, $\widetilde{\mathcal{G}}$ is $\mu$-Donsker.
\end{proof}

\section{Extension of Theoretical Analysis to Discrete variables}\label{sec:discrete}


We consider the case where the covariate $X$ includes discrete components. Without loss of generality we assume that the first component of $X$ is discrete and the rest are continuous. Accordingly, we consider the  partition $X = (D,C)$. Let  $\XX_{c \mid d}$ denote the interior of the support of $C$ conditional on $D=d$,  $\XX_d$ denote the support of $D$,  $\mu_{c \mid d}$ denote the distribution of $C$ conditional on $D = d$, $\mu_d$ denote the distribution of $D$, and $\pi_d(d) = \Pr(D = d)$.
As above, $d_x = \dim(X)$, and $\DD$ is a compact set consisting of regular values of $\Delta$ on $\overline{\XX} := \cup_{d \in \XX_d} \{d\} \times \overline \XX_{c \mid d}$, where $\overline \XX_{c \mid d}$ is the closure of $ \XX_{c \mid d}$.

%

We adjust ${\sf S}.1$-${\sf S}.4$ to hold conditionally at each value of the discrete covariate.

${\sf S}.1'$. The set $\XX_d$ is finite. For any $d \in \XX_d$: the set $\mathcal{X}_{c \mid d}$
is open and its closure $\overline{\mathcal{X}}_{c \mid d}$ is compact;  the distribution $\mu_{c \mid d}$ is absolutely continuous with respect to the Lebesgue measure with density $\mu'_{c \mid d}$; and there exists an open set $B(\mathcal{X}_{c \mid d})$ containing $\overline{\mathcal{X}}_{c \mid d}$ such that $c \mapsto \Delta(d,c)$ is $\C^1$ on  $B(\mathcal{X}_{c \mid d})$,
and $c \mapsto \mu'_{c\mid d}(c)$ is continuous on  $B(\mathcal{X}_{c \mid d})$ and  is zero outside $\mathcal{X}_{c \mid d}$, i.e.  $\mu'(x) = 0$ for any $x \in B(\XX_{c \mid d}) \setminus \XX_{c \mid d}$.

${\sf S}.2'$. For any $d \in \XX_d$ and any regular value $\delta$ of $\Delta$ on $\overline{\mathcal{X}}_{c \mid d}$,  $\MM_{\Delta \mid d}(\delta) := \{c \in
\overline{\mathcal{X}}_{c \mid d} :  \Delta(d,c)= \delta\}$ is either a
$(d_x-2)-$ manifold without boundary on $\mathbb{R}^{d_x-1}$ of class $\C^1$ with finite number of connected branches, or an
empty set.


${\sf S}.3'$. $\widehat{\Delta}$, the estimator  of $\Delta$,  obeys a
functional central limit theorem, namely,
$$a_n(\widehat{\Delta}-\Delta)\rightsquigarrow G_{\infty} \text{ in } \ell^{\infty}
(B({\mathcal{X}})),$$ where $a_n$ is a sequence such that $a_n \to \infty$ as $n \to \infty$, and $c \mapsto G_{\infty}(d,c)$ is a tight process that has almost surely uniformly
continuous sample paths on $B({\mathcal{X}}_{c\mid d})$ for all $d \in \XX_d$.

Let $B(\XX):=
\cup_{d \in \XX_d} \{d\} \times B(\XX_{c \mid d})$;  $\mathcal{F}$ denote a set of continuous functions on $B(\mathcal{X})$ equipped with the sup-norm;  $\VV$ be any compact subset of $ \mathbb{R}$;  $\mathbb{H}$ be the set of all bounded operators $H: g \mapsto
H(g)$ uniformly continuous on $\mathcal{G} = \{1 ( f \leq \delta): f \in \mathcal{F}, \delta \in \VV \}$ with respect to the $L^2(\mu)$ norm, which are represented as:
$$
H(g) = \sum_{d \in \XX_d} H_{d}(d) \int g (c,d)  d\mu_{c \mid d}(c)  + \sum_{d \in \XX_d} \pi_d(d)
H_{c \mid d}(g (\cdot,d)),
$$
 where  $d \mapsto H_d(d)$ is a function that takes on finitely many values and  $g \mapsto H_{c \mid d}(g)$ is a bounded linear operator on $\mathcal{G}$.  Equip the space $\mathbb{H}$
with the sup norm $\| \cdot \|_{\mathcal{G}}$:
$\|H\|_{\mathcal{G}}=\sup_{g\in \mathcal{G}}{|H(g)|}$. Let $\mu(x) = \mu_d(d) \mu_{c\mid d}(c)$ and $\widehat \mu(x) = \widehat \mu_d(d) \widehat \mu_{c\mid d}(c)$.

${\sf S}.4'$. The function $x \mapsto \widehat \mu(x)$ is a distribution over $B(\XX)$ obeying in $\mathbb{H}$,
$$
b_n (\widehat \mu - \mu) \rightsquigarrow H_{\infty},
$$
where $ H_{\infty} \in  \mathbb{H}$ a.s.,  $b_n$ is a sequence such that $b_n \to \infty$ as $n \to \infty$, and  $H_{\infty}$ can be represented as:
$$
H_{\infty}(g) = \sum_{d \in \XX_d} H_{d,\infty}(d) \int g (c,d)  d\mu_{c \mid d}(c)  + \sum_{d \in \XX_d} \pi_d(d)
H_{c \mid d,\infty}(g (\cdot,d)).
$$

We generalize Lemmas \ref{lemma:had0} and \ref{lemma:had3} to the case where $X$ includes  discrete components.

Define $\mathbb{D}:=\mathbb{F}\times \mathbb{H}$ and
$\mathbb{D}_0:=\mathbb{F}_0\times  \mathbb{H}$, where $\mathbb{F}$ is the set of continuous functions on $B(\mathcal{X})$ and $\mathbb{F}_0$ is a subset of $\mathbb{F}$ containing uniformly
continuous functions.

\begin{lemma}[Properties  of $F_{\Delta,\mu}$ and $\Delta^{*}_{\mu}$ with discrete $X$]\label{lemma:dhad}
Suppose that ${\sf S}.1'$ and ${\sf S}.2'$ hold. Then, $\delta \mapsto F_{\Delta,\mu}(\delta)$ is differentiable at any $\delta \in \DD$, with
derivative function $f_{\Delta,\mu}(\delta)$ defined as:
$$f_{\Delta,\mu}(\delta):= \partial_{\delta} F_{\Delta,\mu}(\delta) =\sum_{d\in \XX_{d}} \pi_d(d) \int_{\MM_{\Delta \mid d}(\delta)} \frac{\mu'_{c \mid d}(c)}{\| \partial_c
\Delta(d,c)\|}d\mathrm{Vol}.$$
The map $\delta \mapsto f_{\Delta,\mu}(\delta)$ is uniformly continuous on $\DD$.

(1) The map $F_{\Delta,\mu}(\delta): \mathbb{D} \to \mathbb{R}$ is Hadamard differentiable uniformly in $d \in \DD$ at $(\Delta,\mu)$ tangentially  to $\mathbb{D}_0$, with  derivative map  $\partial_{\Delta,\mu} F_{\Delta,\mu}(\delta): \mathbb{D}_0 \to \mathbb{R}$ defined by:
\begin{eqnarray*}
(G,H)\mapsto \partial_{\Delta,\mu} F_{\Delta,\mu}(\delta)[G,H] := & - &\sum_{d \in
\XX_d}\pi_d(d) \int_{\MM_{\Delta \mid d}(\delta)}\frac{G(d,c)\mu'_{c\mid d}(c)}{\|\partial_c
\Delta(d,c)\|}d\mathrm{Vol}(c)\\
&+& \sum_{d\in\XX_d} H_d(d) \int 1\{\Delta(d,c) \leq \delta\}\mu'_{c|d}(c)dc \\ &+&
\sum_{d\in\XX_d}\pi_d(d)H_{c \mid d}(1\{\Delta(\cdot, d)\leq \delta\}).
\end{eqnarray*}
(2) The map $\Delta^*_{\mu}(u): \mathbb{D} \to \mathbb{R}$ is Hadamard differentiable uniformly in $u \in \UU$  at $(\Delta,\mu)$ tangentially  to $\mathbb{D}_0$, with  derivative map $\partial_{\Delta,\mu} \Delta_{\mu}^*(u): \mathbb{D}_0 \to \mathbb{R}$ defined by:
\begin{equation*}
(G,H)\mapsto \partial_{\Delta,\mu} \Delta_{\mu}^*(u)[G,H] := - \frac{\partial F_{\Delta,\mu}(\Delta^*_{\mu}(u))[G,H]}{f_{\Delta,\mu}(\Delta^*_{\mu}(u))},
\end{equation*}
where $\UU = \{\widetilde u \in [0,1] :  \Delta^*_{\mu}(\widetilde u) \in \DD, f_{\Delta,\mu}(\Delta_{\mu}^*(\widetilde u)) > \varepsilon\}$ for fixed $\varepsilon > 0$.
\end{lemma}

\begin{proof}[Proof of Lemma \ref{lemma:dhad}]
Note that $F_{\Delta,\mu}(\delta)=\sum_{d\in\mathcal{X}_d}\pi_d(d)\int_{c\in\mathcal{X}_d} 1(\Delta(d,c)\leq \delta)\mu'_{c|d}(c)dc$. Given the results of Lemma \ref{lemma:had0}, for each $d$,
 $$\partial_\delta \int_{\mathcal{X}_{c \mid d}} 1(\Delta(d,c)\leq \delta)\mu'_{c|d}(c)dc=\int_{\MM_{\Delta|d}(\delta)} \frac{\mu'_{c|d}(c)}{\|\partial_c \partial(d,c)\|}d{\mathrm{Vol}}.$$
Therefore, averaging over $d \in\XX_d$,
$$f_{\Delta,\mu}(\delta):=\partial_\delta F_{\Delta,\mu}(\delta)= \sum_{d\in\mathcal{X}_d}\pi_d(d)\int_{\MM_{\Delta|d}(\delta)} \frac{\mu'_{c|d}(c)}{\|\partial_c \Delta(d,c)\|}d{\mathrm{Vol}},$$
where we use that $\XX_d$ is a finite set. 

Next we prove the statements (1) and (2). Let $G_n \in \mathbb{F}$ and $H_n \in  \mathbb{H}$ such that $G_n \to G \in \mathbb{F}_0$ and $H_n \to H \in  \mathbb{H}$. Let $\Delta_n=\Delta+t_n  G_n$ and ${\mu}_n=\mu+t_n H_n$, where $t_n \to 0$ as $n \to \infty$.

As in the proof of Lemma \ref{lemma:had3}, we decompose
$$F_{\Delta_n,{\mu}_n}(\delta)-F_{\Delta,\mu}(\delta)=
[F_{\Delta_n,{\mu}_n}(\delta)-F_{\Delta_n,{\mu}}(\delta)]+[F_{\Delta_n,\mu}(\delta)-F_{\Delta,\mu}(\delta)].$$
Applying the same argument as in the proof of Lemma \ref{lemma:had3}
to each $d$ and averaging over $d \in \XX_d$, for any $\delta \in \DD$
 $$\frac{F_{\Delta_n,\mu}(\delta)-F_{\Delta,\mu}(\delta)}{t_n}=-\sum_{d \in
\XX_d}\mu_d(d) \int_{\MM_{\Delta \mid d}(\delta)}\frac{G(d,c)\mu'_{c\mid d}(c)}{\|\partial_c
\Delta(d,c)\|}d\mathrm{Vol}+o(1),$$
where we use that $\XX_d$ is a finite set.
By assumption {\sf S.}$4'$ and a similar argument to the proof of Lemma \ref{lemma:had3},
$$\frac{F_{\Delta_n,\mu_n}(\delta)-F_{\Delta_n,\mu}(\delta)}{t_n}= H(g_{\Delta, \delta})+o(1), \ \ \ \ g_{\Delta,\delta}(c,d) = 1\{\Delta(c,d) \leq \delta \}$$
We conclude that for any $\delta \in \mathcal{D}$,
$$\frac{F_{\Delta_n,\mu_n}(\delta)-F_{\Delta,\mu}(\delta)}{t_n}\to
-\sum_{d \in
\XX_d}\mu_d(d) \int_{\MM_{\Delta \mid d}(\delta)}\frac{G(d,c)\mu'_{c\mid d}(c)}{\|\partial_c
\Delta(d,c)\|}d\mathrm{Vol}+H(g_{\Delta, \delta}) = \partial_{\Delta,\mu} F_{\Delta,\mu}(\delta)[G,H].$$

By an argument similar to the proof of Lemma \ref{lemma:had3}, it can be shown that the convergence is uniform in $\delta \in \DD$. This shows statement (1).

Statement (2) follows from statement (1) and  Theorem 3.9.20 of
\citen{vdV-W} for inverse maps, using an argument analogous  to the proof of statement (b) in Lemma \ref{lemma:had3}.
\end{proof}

We are now ready to derive a functional central limit theorem for the empirical SPE-function. As in Theorem \ref{thm:fclt}, let $r_n := a_n\wedge
b_n$, the slowest of the rates of convergence of $\widehat \Delta$ and $\widehat \mu$, where  $r_n/a_n \to s_{\Delta} \in [0,1]$ and $r_n/b_n \to s_{\mu} \in [0,1]$.

\begin{theorem}[FCLT for $\widehat{\Delta_{\mu}^*}(u)$ with discrete $X$]\label{thm:dfclt}
Suppose that ${\sf S}.1'$-${\sf S}.4'$ hold,  the convergence in ${\sf S}.3'$ and ${\sf S}.4'$ holds jointly,  and $\widehat \Delta \in \mathcal{F}$ with probability
approaching  1. Then, the empirical SPE-process obeys a functional central limit theorem, namely in $\ell^{\infty}(\UU)$,
\begin{equation}\label{eq:fclt2}
r_n(\widehat\Delta^*_{\widehat \mu}(u)-\Delta^*_{\mu}(u))\rightsquigarrow  \partial_{\Delta,\mu} \Delta^*_{\mu}(u)[s_{\Delta} G_{\infty},s_{\mu} H_{\infty}] ,
\end{equation}
as a stochastic process indexed by $u \in \UU$, where $\UU$ is defined in Lemma \ref{lemma:dhad}.
%
%
%
%
%
%
%
\end{theorem}

\begin{remark}[Bootstrap FCLT for $\widehat{\Delta_{\mu}^*}(u)$ with discrete $X$] The exchangeable bootstrap is consistent to approximate the distribution of the limit process in \eqref{eq:fclt2} under the same conditions as in Theorem \ref{thm:bfclt}, replacing ${\sf S}.1$-${\sf S}.4$ by ${\sf S}.1'$-${\sf S}.4'$. Accordingly, we do not repeat the statement here. \qed
\end{remark}

\begin{remark}[CA with discrete covariates] The results of the classification analysis can also be extended to the case where $X$ contains discrete components following analogous arguments as for the SPE. We omit the details for the sake of brevity. \qed
\end{remark}

\begin{proof}[Proof of Theorem \ref{thm:dfclt}] The result follows from Lemma \ref{lemma:dhad} and Lemma \ref{lemma:HS}.
\end{proof}

\section{Some Numerical Illustrations}\label{subset:numerical} We evaluate the accuracy of the asymptotic approximations to the distribution of the empirical SPE in small samples using numerical simulations. In particular,  we compare pointwise 95\%  confidence intervals for the SPE based on the asymptotic and exact distributions of the empirical SPE. The exact distribution is approximated numerically  by simulation. The asymptotic distribution is obtained analytically from the FCLT of Theorem \ref{thm:fclt}, and approximated by  bootstrap using Theorem \ref{thm:bfclt}.  We first consider two simulation designs where the limit process in Theorem \ref{thm:fclt} has a  convenient closed-form analytical expression. The designs differ on whether the PE-function $x \mapsto \Delta(x)$ has critical points or not.  We hold fix the values of the covariate vector $X$ in all the calculations, and accordingly we  treat the distribution $\mu$ as known.  For the bootstrap inference, we use empirical bootstrap with $B = 3,000$ repetitions. All the results are based on $3,000$  simulations. The last design is calibrated to mimic the gender wage gap application.


\begin{design}[No critical points] We consider the PE-function
$$
\Delta(x) =   x_1 + x_2, \ \ x = (x_1,x_2),
$$
with the covariate vector $X$ uniformly distributed in $\XX = (-1,1)\times(-1,1)$.  The corresponding SPE is
$$
\Delta^*_{\mu}(u) = 2(\sqrt{2u} -1)1(u \leq 1/2) + 2(1 - \sqrt{2(1-u)})1(u > 1/2),
$$
where we use that $\Delta(X)$ has a triangular distribution with parameters $(-2,0,2)$. The sample size is $n = 441$ and the values of $X$ are held fixed in the grid $\{-1, -0.9, \ldots, 1 \}\times \{-1, -0.9, \ldots, 1 \}$.
Figure \ref{MC1-PE} plots $x \mapsto \Delta(x)$ on $\XX$, and $u \mapsto \Delta^*_{\mu}(u)$ on $(0,1)$. Here we see that  $x \mapsto \Delta(x)$ does not have critical values, and that $u \mapsto \Delta^*_{\mu}(u)$  is a smooth function.

\begin{figure}
\centering
\includegraphics[width=.49\textwidth,height=.6\textwidth]{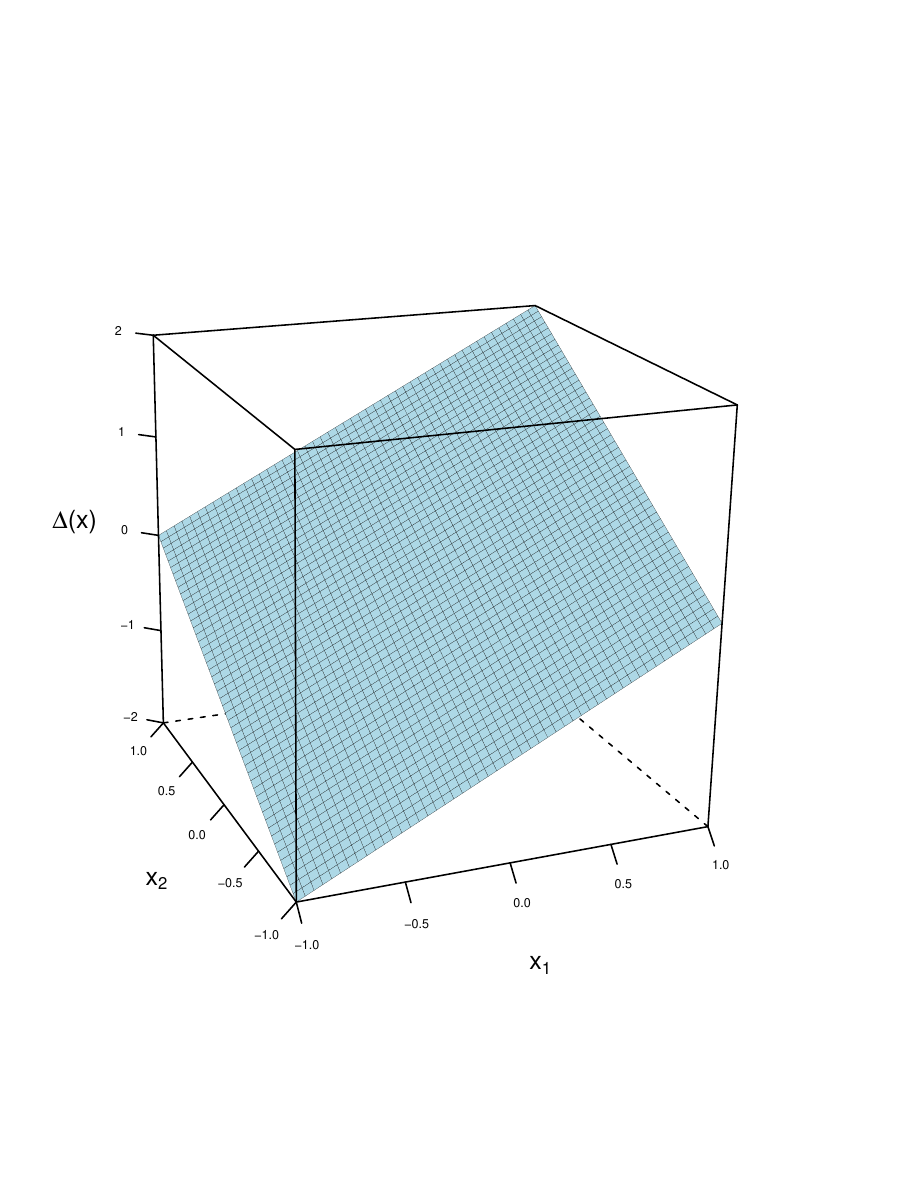}
\includegraphics[width=.49\textwidth,height=.6\textwidth]{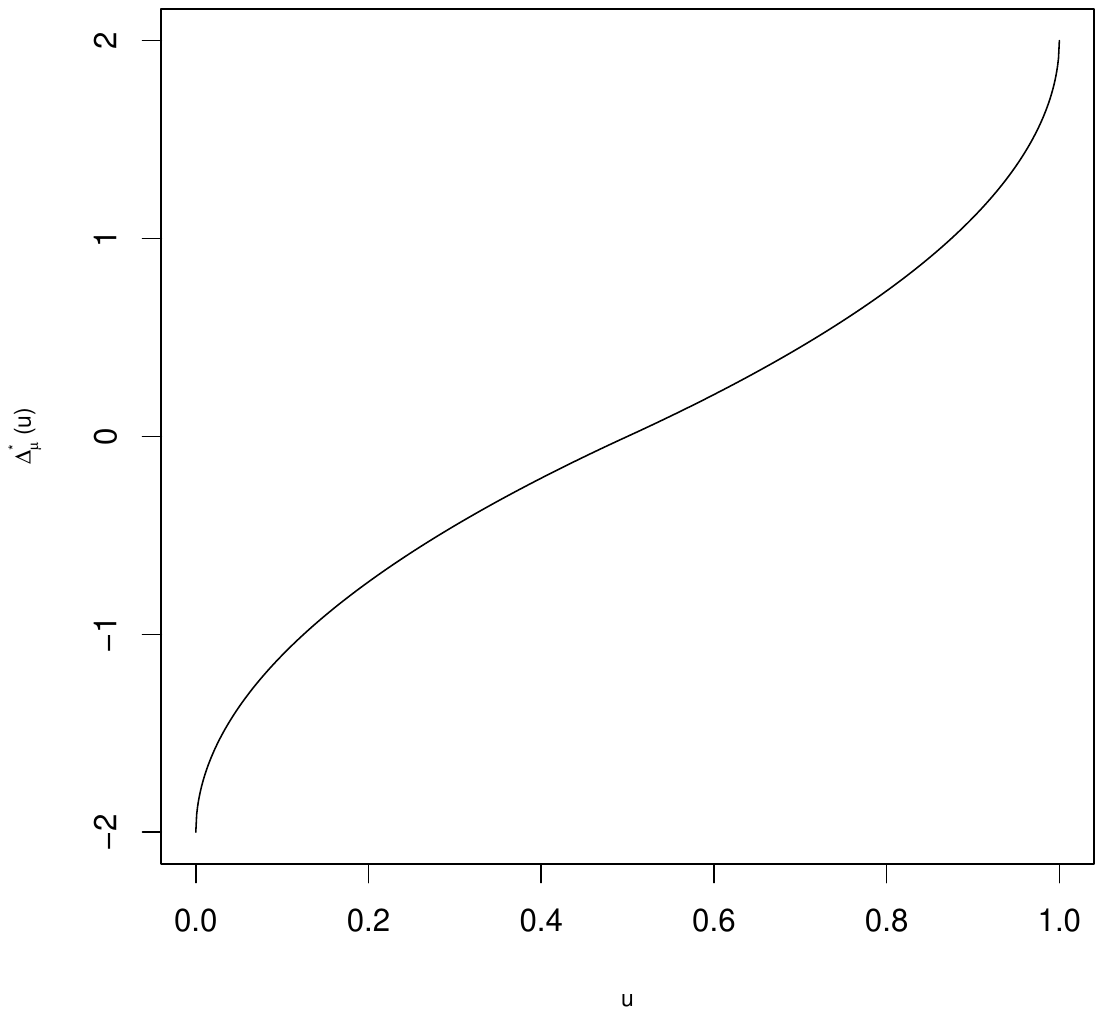}\caption{PE-function and SPE-function  in Design 1. Left: PE function $x \mapsto \Delta(x)$. Right: SPE function $u \mapsto \Delta_{\mu}^*(u)$.}\label{MC1-PE}
\end{figure}

To obtain an analytical expression of the limit  $Z_{\infty}(u)$ of Theorem \ref{thm:fclt}, we make the following  assumption on the  estimator of the PE:
$$
\sqrt{n} (\widehat \Delta(x) - \Delta(x)) = \exp[\Delta(x)]\sum_{i=1}^n Z_i/\sqrt{n},
$$
where $Z_1,\ldots,Z_n$ is an i.i.d. sequence of standard normal random variables. Hence
$$
Z_{\infty}(u) \sim N(0, \exp[2\Delta^*_{\mu}(u)]),
$$
so that $\widehat \Delta^*_{\mu}(u) \overset{a}{\sim} N(\Delta^*_{\mu}(u), \exp[2\Delta^*_{\mu}(u)]/n)$, where $ \overset{a}{\sim}$ denotes asymptotic approximation to the distribution.

Table \ref{MC1-SD} reports biases and compares the standard deviations of the empirical SPE  with the asymptotic standard deviations, $\exp[\Delta^*_{\mu}(u)]/\sqrt{n}$, at the  quantile indices  $u \in \{0.1, 0.2, \ldots, 0.9 \}$. The biases are small relative to dispersions and the asymptotic approximations are very close to the exact standard deviations. We also find that 95\% confidence intervals constructed using the asymptotic approximations, $\widehat \Delta^*_{\mu}(u)  \ \pm  1.96 \exp[\Delta^*_{\mu}(u)] /\sqrt{n}$, have  coverage probabilities close to their nominal levels at all indices.  These asymptotic confidence intervals are not feasible in general, either because  $\Delta^*_{\mu}(u)$ are unknown or more generally because it is not possible to characterize  analytically the distribution of $Z_{\infty}(u)$.  In practice we propose  approximating this distribution by bootstrap. In this case the empirical bootstrap version of the empirical SPE is constructed from the bootstrap PE
$$
\widetilde \Delta(x) =  \Delta(x) + \exp[\Delta(x)] \sum_{i=1}^n \omega_i Z_i/n,
$$
where $(\omega_1, \ldots, \omega_n)$ is a multinomial vector with dimension $n$ and probabilities $(1/n, \ldots, 1/n)$ independent of $Z_1, \ldots, Z_n$.
 The last column of the table shows that the empirical coverages of bootstrap 95\% confidence intervals are  close to their nominal levels at all quantile indices.

%
%
%
%

\begin{table}[ht] \caption{Properties of Empirical SPE in Design 1}
\centering
\begin{tabular}{cccccc}
  \hline\hline
  & \multicolumn{1}{c} {Bias } & \multicolumn{2}{c} {Std. Dev.}  & \multicolumn{2}{c}{Pointwise Coverage (\%)} \\
  \hline
 $u$ & ($\times$ 100) & Exact & Asymptotic & Asymptotic & Bootstrap$^{\dag}$ \\
  \hline
0.1 & 0.016 & 0.014 & 0.014 & 95.10 & 95.03 \\
  0.2 & 0.024 & 0.021 & 0.021 & 95.10 & 95.03 \\
  0.3 & 0.032 & 0.029 & 0.029 & 95.10 & 95.03 \\
  0.4 & 0.044 & 0.039 & 0.039 & 95.10 & 95.03 \\
  0.5 & 0.053 & 0.047 & 0.048 & 95.10 & 95.03 \\
  0.6 & 0.065 & 0.058 & 0.058 & 95.10 & 95.03 \\
  0.7 & 0.088 & 0.078 & 0.079 & 95.10 & 95.03 \\
  0.8 & 0.119 & 0.105 & 0.106 & 95.10 & 95.03 \\
  0.9 & 0.177 & 0.157 & 0.158 & 95.10 & 95.03 \\
   \hline\hline
    \multicolumn{6}{l} {\footnotesize{Notes: $3,000$ simulations with sample size $n = 441$.}}\\
    \multicolumn{6}{l} {\footnotesize{$^\dag$3,000 bootstrap repetitions. Nominal level is 95\%.}}\\
\end{tabular}\label{MC1-SD}
\end{table}

\end{design}

\begin{design}[Critical points] We consider the PE-function
$$
\Delta(x) =   x^3 - 3x,
$$
with covariate $X$ uniformly distributed on $\XX = (-3,3)$.  Figure \ref{MC2-PE} plots $x \mapsto \Delta(x)$ on $\XX$, and $u \mapsto \Delta^*_{\mu}(u)$ on $(0,1)$.\footnote{We obtain  $u \mapsto \Delta^*_{\mu}(u)$ analytically using the characterization of \citen{CFG-10} for the univariate case.} Here we see that  $x \mapsto \Delta(x)$ has two critical points at $x=-1$ and $x=1$ with corresponding critical values at $\delta = 2$ and $\delta = -2$. The SPE-function $u \mapsto \Delta^*_{\mu}(u)$  has two kinks at $u = 1/6$ and $u = 5/6$, the $\Delta^*_{\mu}$ pre-images of the critical values.

\begin{figure}
\centering
\includegraphics[width=.49\textwidth,height=.6\textwidth]{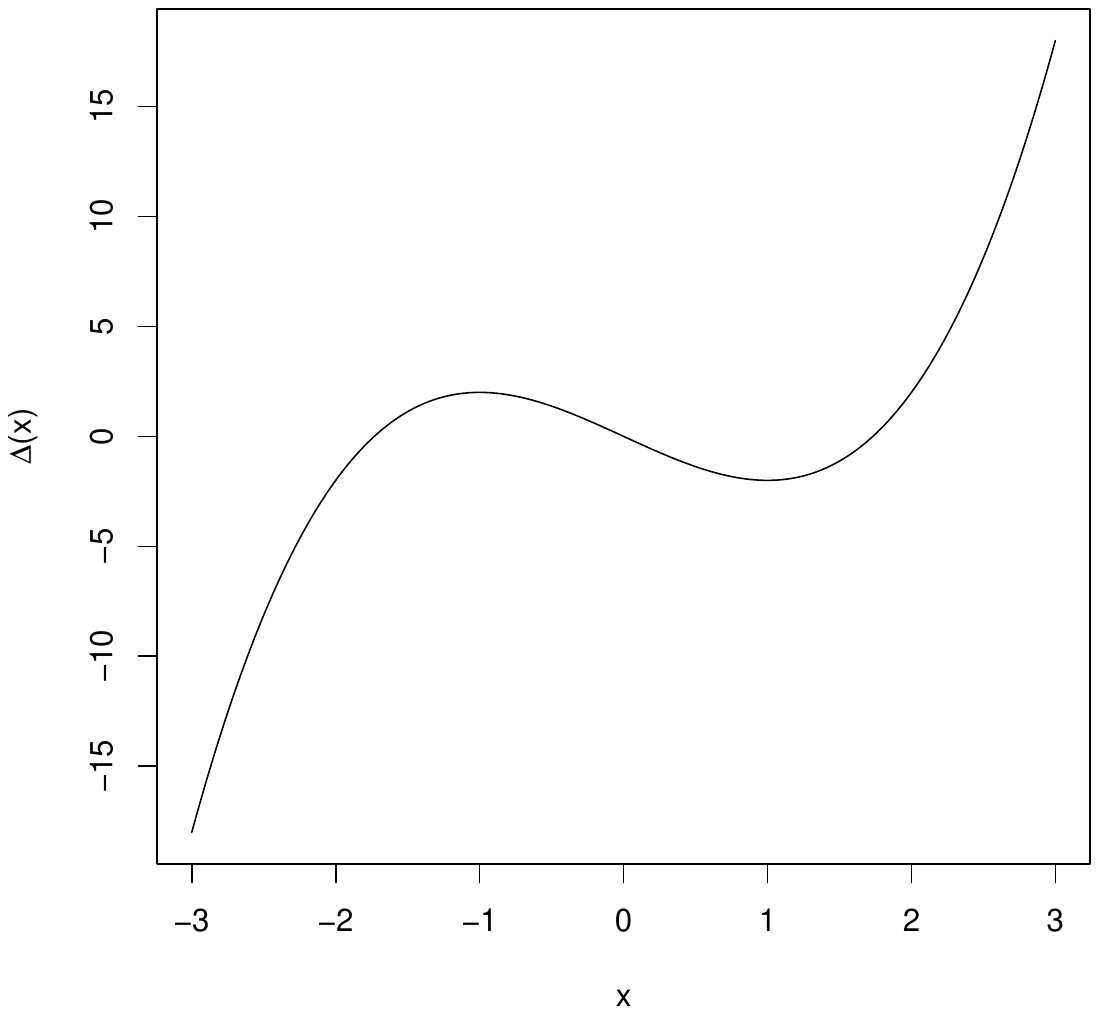}
\includegraphics[width=.49\textwidth,height=.6\textwidth]{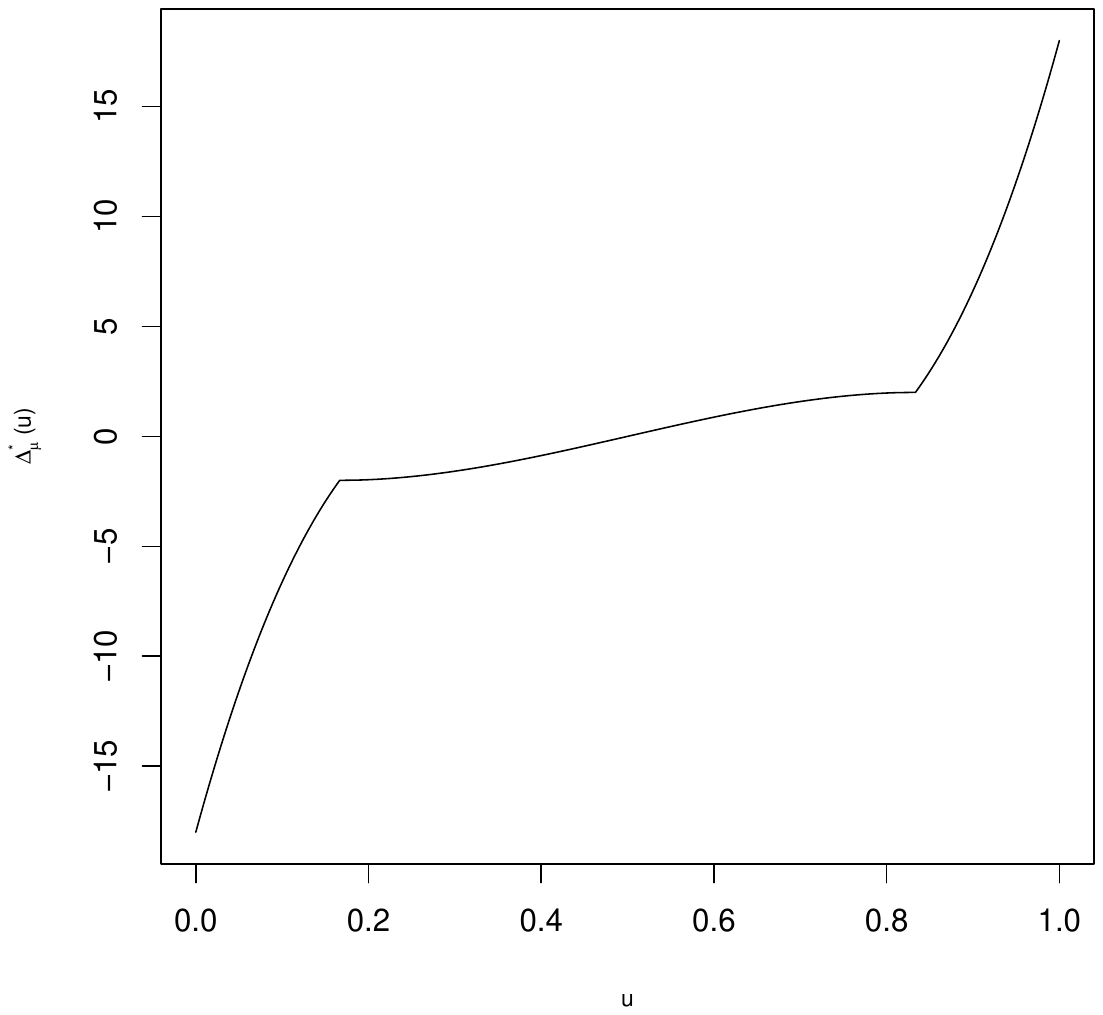}\caption{PE-function and SPE-function in Design 2. Left: PE function $x \mapsto \Delta(x)$. Right: SPE function $u \mapsto \Delta_{\mu}^*(u)$.}\label{MC2-PE}
\end{figure}

To obtain an analytical expression of the limit  $Z_{\infty}(u)$ of Theorem \ref{thm:fclt}, we make the following  assumption on the estimator of the PE:
$$
\sqrt{n} (\widehat \Delta(x) - \Delta(x))  =(x/2)^2 \sum_{i=1}^n Z_i/\sqrt{n},
$$
where $Z_1,\ldots,Z_n$ is an i.i.d. sequence of standard normal variables. This assumption is analytically convenient because  after some  calculations we find that for $u \notin \{1/6,5/6\}$,
$$
Z_{\infty}(u) \sim N(0, S(\Delta^*_{\mu}(u))^2/(4n)),
$$
where
$$S(\delta) = 1(\delta < -2) \breve{\Delta} _1(\delta)^2 + 1(-2 < \delta < 2) \sum_{k=1}^3 \frac{ \breve{\Delta} _k(\delta)^2 | \breve{\Delta} _k(\delta)^2 -1|^{-1}}{\sum_{j=1}^3 | \breve{\Delta} _j(\delta)^2 -1|^{-1}} + 1(\delta > 2) \breve{\Delta} _1(\delta)^2,
$$
and $\breve{\Delta} _1(\delta)$, $\breve{\Delta} _2(\delta)$ and $\breve{\Delta} _3(\delta)$ are  real roots of $\Delta(x) - \delta = 0$ sorted in increasing order.\footnote{The equation $\Delta(x) - \delta = x^3 -  3x - \delta = 0$ has three real roots when $\delta \in (-2,2)$, and one real root when $\delta < -2$ or $\delta > 2$.}
Hence,  $\widehat \Delta^*_{\mu}(u) \overset{a}{\sim} N(\Delta^*_{\mu}(u), S(\Delta^*_{\mu}(u))^2/(4n))$.

Table \ref{MC2-SD} reports biases and compares the standard deviations of the empirical SPE in samples of size $n = 601$ with the asymptotic standard deviations at the  quantile indices  $u \in \{1/12, 2/12, \ldots, 11/12 \}$, where the values of $X$ are held fixed in the grid $\{-3, -2.99, \ldots, 3\}$. The biases are small relative to dispersion except at the kinks $u = 1/6$ and $u = 5/6$ . The asymptotic approximation is  close to the exact standard deviation, except for the quantiles at the kinks where the asymptotic standard deviations are  not well-defined because $\breve{\Delta} _k(\delta)^2 -1 = 0$. We also find that pointwise 95\% confidence intervals constructed using the asymptotic distribution and empirical bootstrap  have  coverage probabilities close to their nominal levels.   Interestingly,  the bootstrap  provides coverages  close to the nominal levels even at the kinks.

\begin{table}[ht]\caption{Properties of Empirical SPE in Design 2}
\begin{center}
\begin{tabular}{cccccc}
  \hline\hline
  & \multicolumn{1}{c} {Bias } & \multicolumn{2}{c} {Std. Dev.}  & \multicolumn{2}{c}{Pointwise Coverage (\%)} \\
  \hline
 $u$ & ($\times$ 100) & Exact & Asymptotic & Asymptotic & Bootstrap$^{\dag}$ \\
  \hline
  1/12 & 0.068 & 0.126 & 0.127 & 95.67 & 95.80 \\
  1/6 & -2.393 & 0.054 & -- & -- & 95.67 \\
  1/4 & -0.005 & 0.025 & 0.025 & 95.83 & 95.77 \\
  1/3 & -0.016 & 0.028 & 0.028 & 95.80 & 95.90 \\
  5/12 & 0.045 & 0.030 & 0.030 & 95.63 & 95.47 \\
  1/2 & 0.023 & 0.030 & 0.031 & 92.73 & 97.53 \\
  7/12 & -0.020 & 0.030 & 0.030 & 95.20 & 95.80 \\
  2/3 & 0.049 & 0.028 & 0.028 & 95.53 & 95.67 \\
  3/4 & 0.039 & 0.025 & 0.025 & 95.53 & 95.73 \\
  5/6 & 2.447 & 0.053 & -- & -- & 95.73 \\
  11/12 & 0.068 & 0.126 & 0.127 & 95.67 & 95.80 \\
   \hline\hline
    \multicolumn{6}{l} {\footnotesize{Notes: $3,000$ simulations with sample size $n = 601$.}}\\
    \multicolumn{6}{l} {\footnotesize{$^\dag$3,000 bootstrap repetitions. Nominal level is 95\%.}}\\
\end{tabular}\label{MC2-SD}
\end{center}
\end{table}


\end{design}

\begin{design}[Calibration to CPS data] This design is calibrated to the interactive linear model with additive error  for the conditional expectation in the gender wage gap application of Section \ref{sec:empirics}. More specifically, we generate log wages as
$$
Y_i = P(T_i,W_i)'\beta + \sigma \varepsilon_i, \ \ i = 1,\ldots,n,
$$
where the covariates $X_i = (T_i,W_i)$ are fixed to the values in the 2015 CPS data set,  $P(T,W) = (TW, (1-T)W)$, $\beta$ and $\sigma^2$ are the least squares estimates of the regression coefficients and residual variance in the data set, $(\varepsilon_1,\ldots,\varepsilon_n)$ is a sequence of i.i.d. standard normal random variables independent of $X_i$, and $n = 32,523$, the sample size in the application. For each simulated sample $\{(Y_i,X_i) : 1 \leq i \leq n\}$, we reestimate the model by least squares, obtain the empirical   SPE-function on the treated over a grid of percentile indexes $\UU = \{0.02, 0.03, \ldots, 0.98\}$, and construct a 90\% uniform  confidence band for the SPE-function using Algorithm \ref{alg:computation} with standard exponential weights and $B=200$. We repeat this procedure $500$ times.

\begin{figure}
\centering
\includegraphics[width=.32\textwidth,height=.45\textwidth]{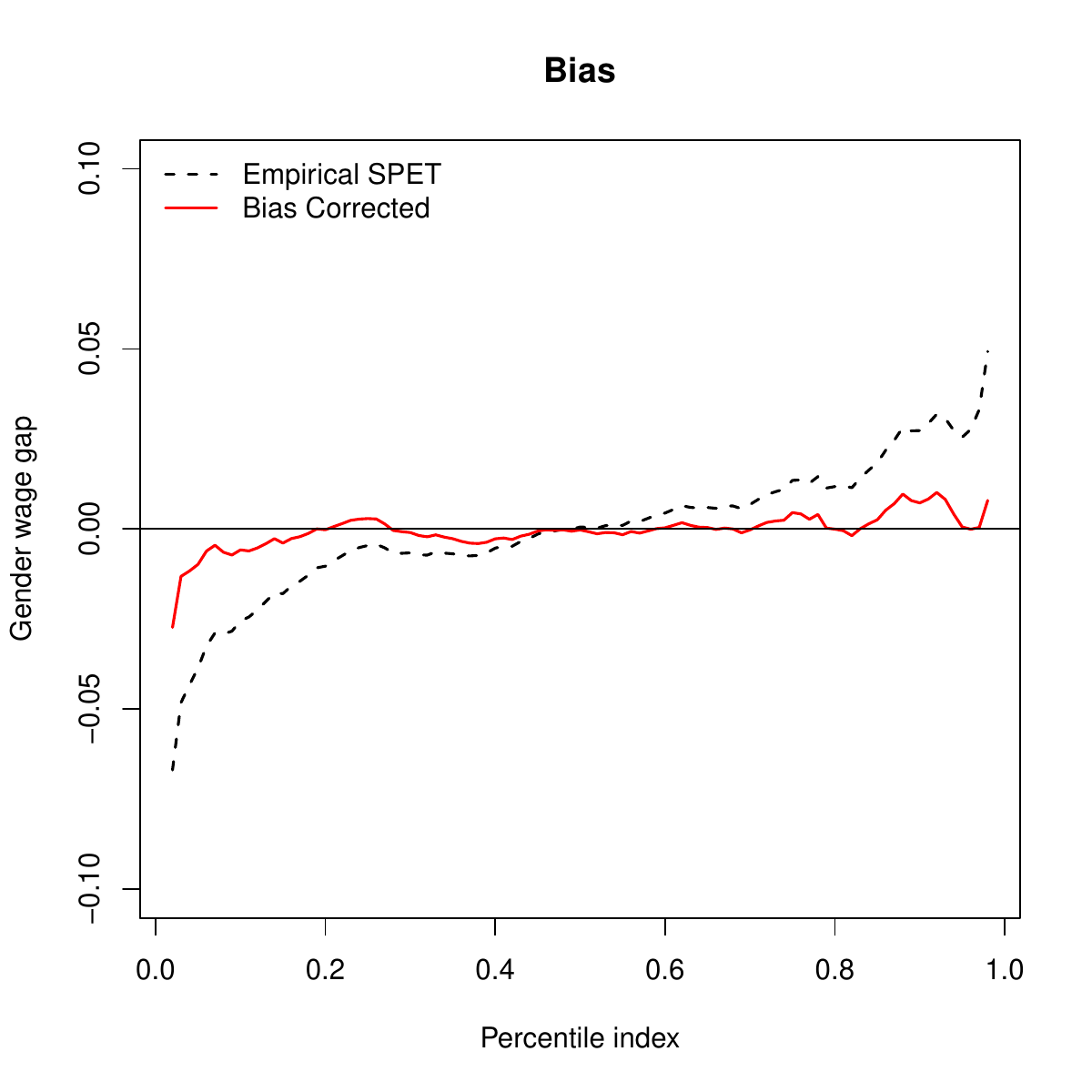}
\includegraphics[width=.32\textwidth,height=.45\textwidth]{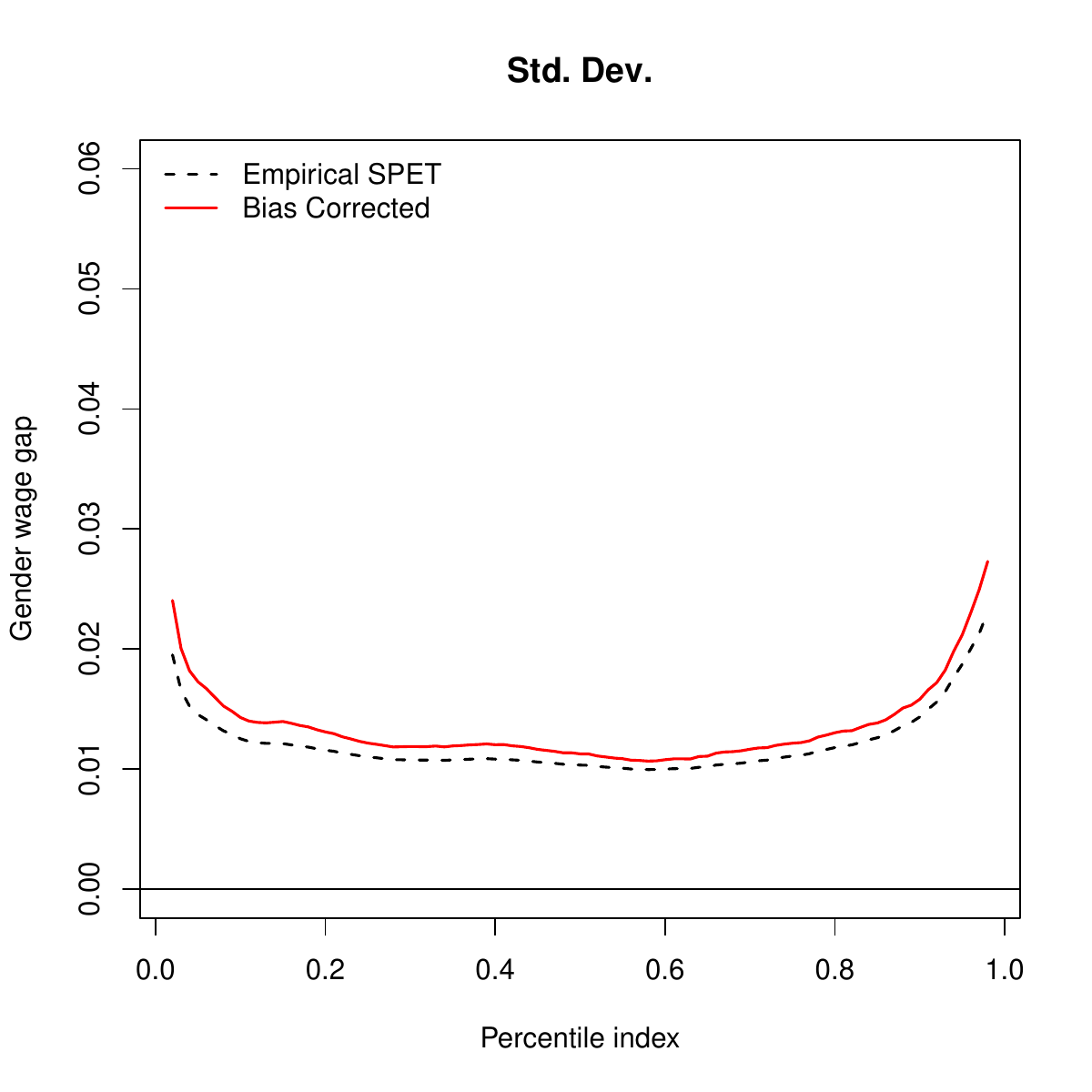}
\includegraphics[width=.32\textwidth,height=.45\textwidth]{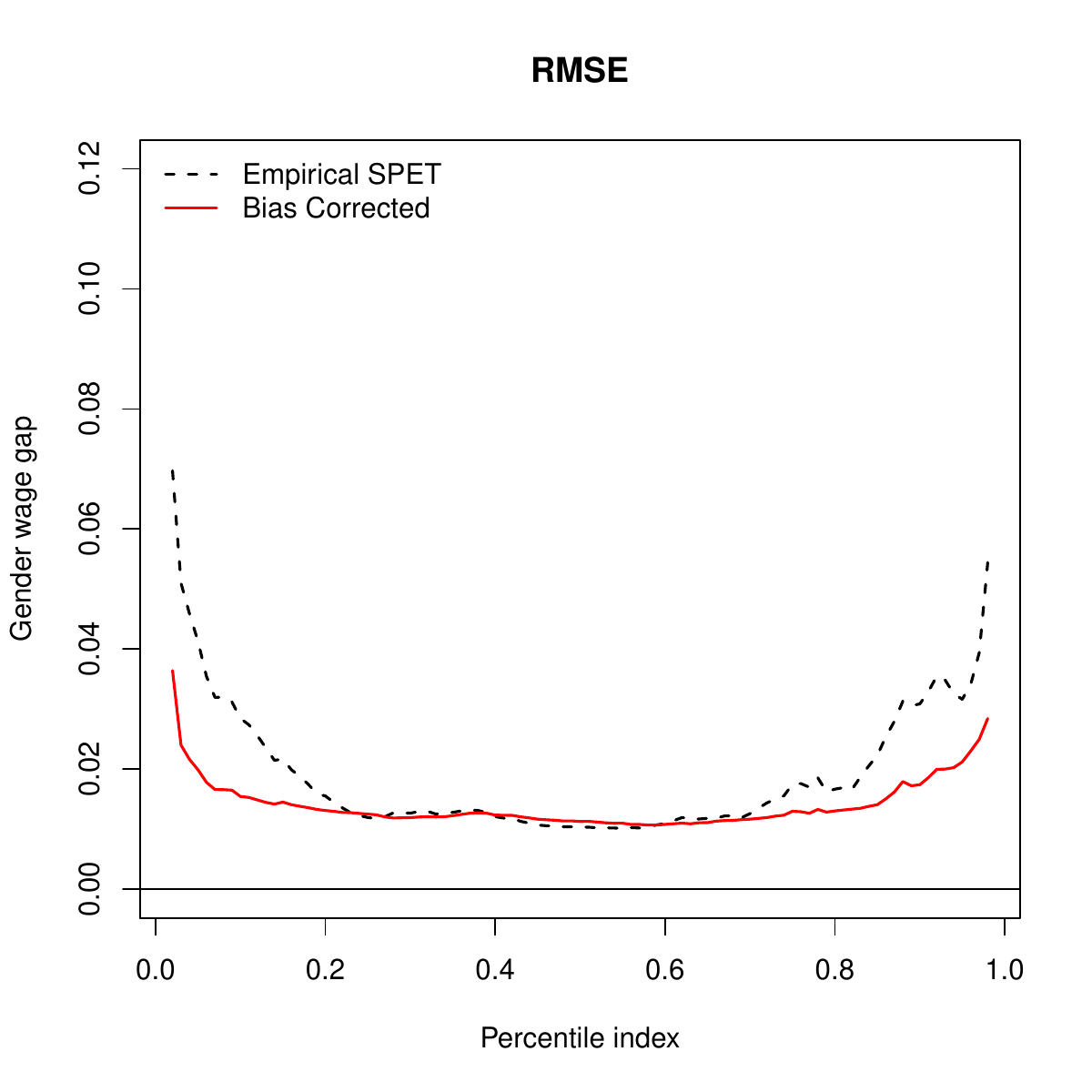}
\caption{Bias, standard deviation and root mean square error of empirical and bias corrected SPE functions. Results obtained from 500 repetitions of a design calibrated to the CPS 2015 data.}\label{MC3-bias}
\end{figure}

Figure  \ref{MC3-bias} reports bias, standard deviation (Std. Dev.) and root mean squared error (RMSE) of the empirical and bias corrected SPE functions, see Remark \ref{re:fsc}. We find that the empirical SPE displays negative bias in the lower tail and positive bias in the upper tail, which are reduced by the bootstrap bias correction. The correction slightly increases  dispersion, but reduces overall rmse for most percentiles, specially at the tails.   Table  \ref{MC3-bias}  reports the empirical coverage of 90\% confidence bands constructed around the empirical and
bias corrected SPE functions. Here we find that the uncorrected bands undercover the entire SPE function, whereas the corrected bands have coverage above the nominal level. One possible reason for the overcoverage of the bootstrap corrected bands is that we keep the covariates fixed across samples, which is not accounted by the bootstrap procedure. To sum up, we find that the bootstrap corrections of Remark \ref{re:fsc} reduce the bias of the empirical SPE and improve the coverage of the confidence bands in finite samples.


\begin{table}[ht]
\centering\caption{Coverage of 90\% Confidence Bands} \label{MC3-cover}
\begin{tabular}{rrr}
  \hline
 & Uncorrected & Bootstrap Bias Corrected \\
  \hline
Coverage & 0.80 & 0.98  \\
   \hline
   \multicolumn{3}{l}{
   \footnotesize{Notes: 500 simulations and 200 bootstrap repetitions.}}\\
   \multicolumn{3}{l}{
   \footnotesize{DGP calibrated to CPS 2015.}}
\end{tabular}
\end{table}

\end{design}

\section{Effect of Race on Mortgage Denials}\label{sec:mortgage}

To study the effect of race in the bank decisions of mortgage denials or racial mortgage denial gap, we use data on mortgage applications in Boston from 1990 (see \citen{MTBM1996}). The Federal Reserve Bank of Boston collected these data  in relation to the Home Mortgage Disclosure Act (HMDA), which was passed to monitor minority access to the mortgage market.
Providing better access to credit markets can arguably help the disadvantaged groups escape poverty traps.  Following \citen[Chap 11]{StockWatson}, we focus on white and black applicants for single-family residences. The sample includes $2,380$ observations corresponding to $2,041$ white applicants and $339$ black applicants.

We estimate a binary response model where the outcome variable $Y$ is an indicator for mortgage denial, the key covariate $T$ is an indicator for the applicant being black, and the controls $W$ contain financial and other characteristics of the applicant that banks take into account in the mortgage decisions. These include the monthly debt to income ratio; monthly housing expenses to income ratio; a categorial variable for ``bad'' consumer credit score with 6 categories (1 if no slow payments or delinquencies, 2 if one or two slow payments or delinquencies, 3 if more than two slow payments or delinquencies, 4 if insufficient credit history for determination, 5 if delinquent credit history with payments 60 days overdue, and 6 if delinquent credit history with payments 90 days overdue); a categorical variable for ``bad'' mortgage credit score with 4 categories (1 if no late mortgage payments, 2 if no mortgage payment history, 3 if one or two late mortgage payments, and 4 if more than two late mortgage payments); an indicator for public record of credit problems including bankruptcy, charge-offs, and collective actions; an indicator for denial of application for mortgage insurance; two indicators for medium and high loan to property value ratio, where medium is between .80 and .95 and high is above .95; and three indicators for self-employed, single, and high school graduate.

\begin{table}[htbp]\caption{\label{table:ds} Descriptive Statistics of Mortgage Applicants}
\centering
\begin{tabular}{lccc} \hline\hline
 & All & Black & White \\
 \hline
 Deny & 0.12 & 0.28 & 0.09 \\
  Black & 0.14 & 1.00 & 0.00 \\
  Debt-to-income ratio & 0.33 & 0.35 & 0.33 \\
  Expenses-to-income ratio & 0.26 & 0.27 & 0.25 \\
  Bad consumer credit & 2.12 & 3.02 & 1.97 \\
  Bad mortgage credit & 1.72 & 1.88 & 1.69 \\
  Credit problems & 0.07 & 0.18 & 0.06 \\
  Denied mortgage insurance & 0.02 & 0.05 & 0.02 \\
  Medium loan-to-value ratio & 0.37 & 0.56 & 0.34 \\
  High loan-to-value ratio & 0.03 & 0.07 & 0.03 \\
  Self-employed & 0.12 & 0.07 & 0.12 \\
  Single & 0.39 & 0.52 & 0.37 \\
  High school graduate & 0.98 & 0.97 & 0.99 \\
  \hline
   number of observations & 2,380 & 339 & 2,041 \\ \hline\hline
\end{tabular}
\end{table}

Table \ref{table:ds} reports the sample means of the variables used in the analysis. The probability of having the mortgage denied is $19\%$ higher for black applicants than for white applicants. However, black applicants are more likely to have socio-economic characteristics linked to a denial of the mortgage.

 \begin{figure}
\centering
\includegraphics[width=.45\textwidth,height=.45\textwidth]{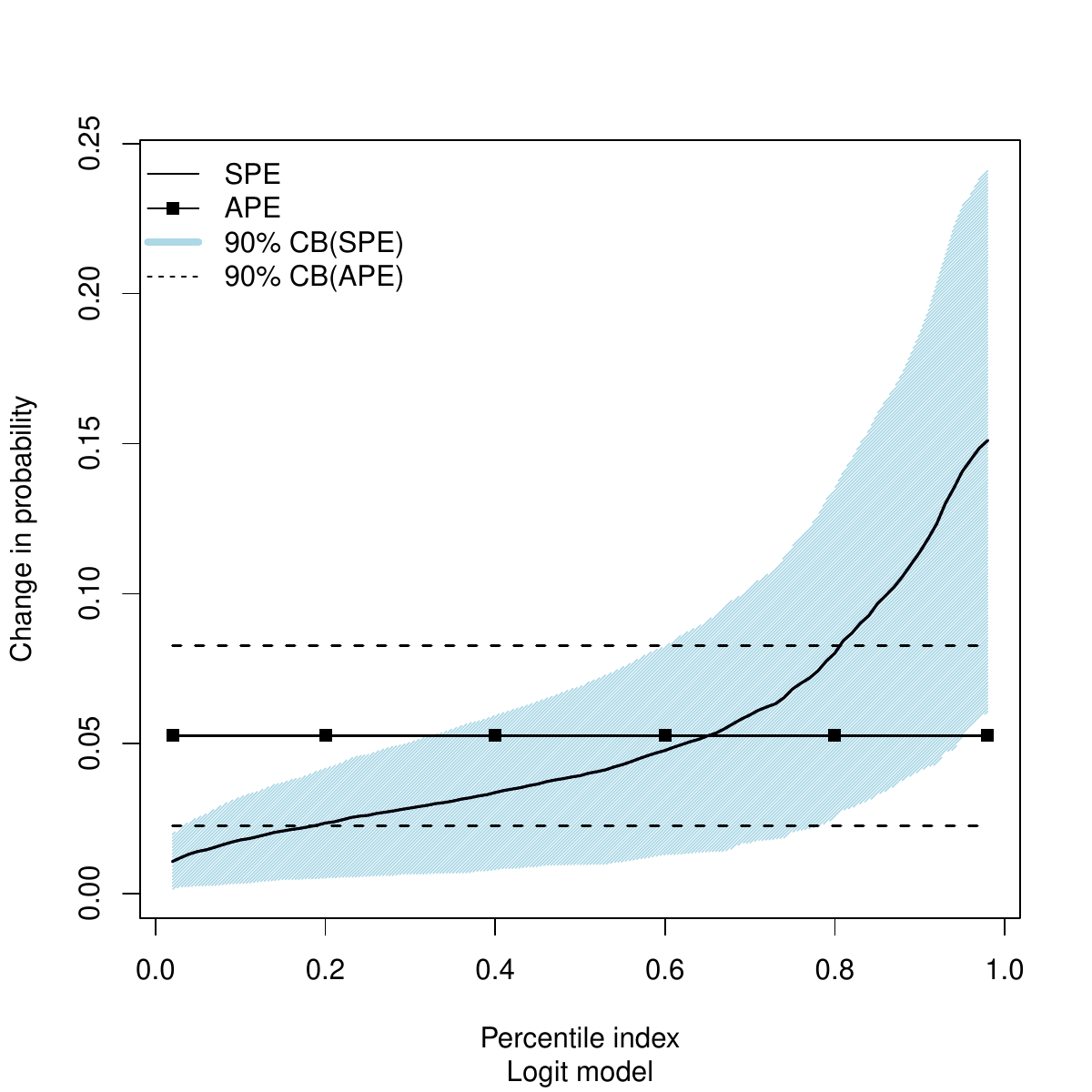}
\caption{APE and SPE (introduced in this paper) of being
black on the probability of mortgage denial.  Estimates and 90\% bootstrap uniform confidence bands (derived in this paper) based on a logit model are shown. }\label{fig:spe-black}
\end{figure}

Figure \ref{fig:spe-black} plots estimates and 90\% confidence sets of the population  APE and SPE-function of being black.  The  PEs are obtained as described in Example \ref{example:binary} of the main text using  a logit model with $P(X) = X = (T,W)$ and $\widehat \mu$ equal to the empirical distribution of $X$ in the whole sample. The confidence bands are constructed using Algorithm \ref{alg:computation} with multinomial weights (empirical bootstrap) and  $B=500$, and are uniform for the SPE-function over the grid $\mathcal{U} = \{.02, .03, \ldots, .98\}$.   We monotonize the bands using the rearrangement method of \citen{cfg-09}. 
After controlling for applicant characteristics, black applicants are still on average 5.3\% more likely to have the mortgage denied than white applicants.  Moreover, the SPE-function shows significant heterogeneity, with the PE ranging between 0 and  15\%.  Thus,
there exists a subgroup of applicants that is $15\%$ more likely to be denied a mortgage if they were black, and  there is a subgroup of  applicants that is not affected by racial mortgage denial gap.
Table \ref{table:class} shows  the results of the classification analysis, answering the question ``who is affected the most and who the least?"  The table shows
that the 10\% of the applicants \textit{most affected}  by racial mortgage denial gap are  \textit{more likely} to have either of the following characteristics relative to the 10\% of the \textit{least affected} applicants: self employed, single,  black, high debt to income ratio, high expense to income ratio, high loan to value ratio, medium or high loan-to-income ratio, bad consumer or credit scores, and credit problems.


\begin{table}[htbp]\caption{\label{table:class} Who is affected the most and who the least?  Classification Analysis -- Averages of Characteristics of the Mortgage  Applicants Least and Most Affected by Racial Discrimination}
\begin{center}
\begin{tabular}{llclc} \hline\hline 
Characteristics & \multicolumn{2}{c}{10\% Most Affected} & \multicolumn{2}{c}{10\% Least Affected}\\
 of the Group  & \multicolumn{2}{c}{PE  $> .11$} &  \multicolumn{2}{c}{PE $< .018$} \\ 
 \hline
  Deny & 0.44 & (0.03) & 0.11 & (0.04) \\ 
  Black & 0.37 & (0.04) & 0.07 & (0.02) \\ 
  Debt-to-income  & 0.39 & (0.01) & 0.25 & (0.02) \\ 
  Expenses-to-income & 0.28 & (0.01) & 0.21 & (0.02) \\ 
  Bad consumer credit & 4.64 & (0.25) & 1.31 & (0.09) \\ 
  Bad mortgage credit & 1.99 & (0.07) & 1.37 & (0.12) \\ 
  Credit problems & 0.45 & (0.05) & 0.05 & (0.02) \\ 
  Denied mortgage insurance & 0.01 & (0.01) & 0.06 & (0.04) \\ 
  Medium loan-to-house & 0.58 & (0.06) & 0.07 & (0.04) \\ 
  High loan-to-house& 0.13 & (0.03) & 0.02 & (0.01) \\ 
  Self employed & 0.18 & (0.05) & 0.05 & (0.03) \\ 
  Single & 0.59 & (0.05) & 0.11 & (0.06) \\ 
  High school grad & 0.93 & (0.03) & 1.00 & (0.01) \\ 
   \hline\hline
   \multicolumn{5}{l}{Std. errors in parentheses obtained by bootstrap with 200 repetitions.}
\end{tabular}\end{center}
\end{table}

\bibliographystyle{econometrica}
\bibliography{mybibVOLUME}

\end{document}